\documentclass[prb,twocolumn,showpacs,nofootinbib,preprintnumbers,amsmath,amssymb,aps,longbibliography,superscriptaddress]{revtex4-2}
\usepackage[T1]{fontenc}
\usepackage[utf8]{inputenc}
\usepackage{subfigure, lmodern, amsmath,amssymb, graphicx, pifont, adjustbox, bm, xcolor}
\usepackage{amsfonts}
\usepackage{amsthm}
\usepackage{comment}
\usepackage{mathtools}
\usepackage{float}
\usepackage{braket}
\usepackage{longtable}
\usepackage{graphicx, subfigure}
\usepackage{xcolor}
\usepackage{tikz}
\usepackage{enumitem}
\usepackage{dsfont}

\newtheorem{lemma}{Lemma}[section]

\makeatletter\g@addto@macro\bfseries{\boldmath}\makeatother

\usepackage[colorlinks=true, citecolor=blue, urlcolor=blue, linkcolor=blue, breaklinks=true, pdfpagelabels=false]{hyperref}

\newcommand{\appendixref}[1]{\hyperref[#1]{appendix~\ref{#1}}}
\def\equationautorefname~#1\null{eq.\,(#1)\null}

\newcommand{\ie}{\begin{equation}\begin{aligned}}
\newcommand{\fe}{\end{aligned}\end{equation}}
\newcommand{\figref}[1]{Fig.~\ref{#1}}
\newcommand{\tr}{\text{tr}}
\newcommand{\rank}{\text{rank}}

\usepackage{cleveref}
\crefname{appendix}{App.}{Apps.}
\crefname{equation}{Eq.}{Eqs.}
\crefname{figure}{Fig.}{Figs.}
\crefname{table}{Tab.}{Tabs.}
\crefname{section}{Sec.}{Secs.}

\begin{document}

\title{Quantum Fragmentation in the Extended Quantum Breakdown Model}
\date{\today}

\author{Bo-Ting Chen}
\affiliation{Department of Physics, Princeton University, Princeton, New Jersey 08544, USA}
\author{Abhinav Prem}
\affiliation{School of Natural Sciences, Institute for Advanced Study, Princeton, New Jersey 08540, USA}
\author{Nicolas Regnault}
\affiliation{Laboratoire de Physique de l’Ecole normale sup\'erieure, ENS, Universit\'e PSL, CNRS, Sorbonne Universit\'e, Universit\'e Paris-Diderot, Sorbonne Paris Cit\'e, 75005 Paris, France}
\affiliation{Department of Physics, Princeton University, Princeton, New Jersey 08544, USA}
\author{Biao Lian}
\affiliation{Department of Physics, Princeton University, Princeton, New Jersey 08544, USA}

\begin{abstract}
We introduce a one-dimensional (1D) extended quantum breakdown model comprising a fermionic and a spin degree of freedom per site, and featuring a spatially asymmetric breakdown-type interaction between the fermions and spins. We analytically show that, in the absence of any magnetic field for the spins, the model exhibits Hilbert space fragmentation within each symmetry sector into exponentially many Krylov subspaces and hence displays non-thermal dynamics. Here, we demonstrate that the fragmentation naturally occurs in an entangled basis and thus provides an example of ``quantum fragmentation." Besides establishing the nature of fragmentation analytically, we also study the long-time behavior of the entanglement entropy and its deviation from the expected Page value as a probe of ergodicity in the system. Upon introducing a non-trivial magnetic field for the spins, most of the Krylov subspaces merge and the model becomes chaotic. Finally, we study the effects of strong randomness on the system and observe behavior similar to that of many-body localized systems.
\end{abstract}

\maketitle

\section{Introduction}
\label{sec:intro}

Understanding the non-equilibrium dynamics of isolated many-body quantum systems is a key area of research in modern condensed matter physics, spurred in large part by experimental progress in manipulating ultra-cold atoms which can well approximate systems decoupled from any environment~\cite{weiss2006,gring2012,schreiber2015,smith2016,kaufman2016,kucsko2018}. On the theoretical front, there has been significant progress in understanding when (and how) a closed quantum system, evolving under its own dynamics, attains thermal equilibrium. A fundamental insight is provided by the Eigenstate Thermalization Hypothesis (ETH)~\cite{deutsch1991quantum,srednicki1994chaos,rigol2008thermalization, polkovnikov2011colloquium}, the strong version of which states that all eigenstates of an ergodic system behave like thermal states, as far as expectation values of local observables are concerned (see Refs.~\cite{d2016quantum,gogolin2016review,ueda2018review} for a review).

While generically one expects that most interacting quantum systems obey the ETH, many-body localised (MBL) systems~\cite{anderson1958,gornyi2005,baa2006} provide striking counterexamples. Unlike integrable systems, which are fine-tuned and have extensively many conserved quantities, the interplay of disorder and interactions in MBL systems leads to the emergence of extensively many local integrals of motion, endowing such systems with rich dynamical features~\cite{rahul2015review,altman2015review,abanin2018review}. A related question is whether non-ergodic dynamics can be obtained in non-integrable, disorder-free systems~\cite{deroeck2014trans,deroeck2014AMBL,groverfisher,schiulaz2015,papic2015nodisorder,yao2016,prem2017glass,smith2017a,smith2017b,michailidis2018,brenes2018}. Recently, the phenomenon of quantum many-body scars (QMBS)~\cite{bernien2017probing,turner2017quantum,vafek2017entanglement, moudgalya2018exact,schecter2018many,turner2018quantum, choi2018emergent,ho2018periodic,moudgalya2018nonint} i.e., the presence of a measure zero set of non-thermal eigenstates embedded in an otherwise thermal spectrum was discovered as a novel mechanism for weak ergodicity breaking which, for an appropriate choice of initial conditions, leads to late-time dynamics that deviates strongly from that of a thermal system~\cite{serbyn_review,chandran_review}. A closely related phenomenon is that of Hilbert space fragmentation (HSF)~\cite{sala2019ergodicity,khemani2020hilbert,rakovszky2019statistical,moudgalya2019thermalization,sanjay_review}, in which there exist exponentially many dynamically disconnected sectors--or Krylov subspaces--that are not merely distinguished by conventional global symmetries~\cite{moudgalya2021hilbert,moudgalya2022from}. In such systems, a physically motivated choice of basis leads to dynamically disconnected Krylov subspaces, each of which may be chaotic or non-ergodic, leading to a notion of ``Krylov restricted thermalization"~\cite{moudgalya2019thermalization}. 

Given the rapid theoretical progress in identifying various mechanisms for non-ergodicity in isolated many-body quantum systems, it is of much interest to identify experimentally feasible models which realize such physics and for which at least some results can be obtained in an exact manner. Recently, Ref.~\cite{Lian:2022nqj} introduced the ``quantum breakdown model" which provides a simplified quantum model for the breakdown of dielectrics in an electric field (the Townsend avalanche). The model consists of a one-dimensional (1D) fermionic chain with multiple degrees of freedom per site, where the conventional hopping is replaced by a spatially asymmetric interaction that converts between one fermion and multiple fermions (dubbed a ``breakdown" interaction), resulting in rich dynamical behaviour which includes HSF, MBL, and a many-body scar flat band. The quantum breakdown model has been generalized to similar models with spin and bosonic degrees of freedom~\cite{liu2023,bosonicqbhm}, and which include various types of spatially asymmetric breakdown interactions.

In this paper, we build upon Ref.~\cite{Lian:2022nqj} and further investigate the non-equilibrium dynamics of systems with both  hopping and breakdown-type spatially asymmetric interactions. Specifically, we consider a 1D chain with a fermionic and a spin degree of freedom on each site, where the spin-fermion interaction is spatially asymmetric and only the total fermionic charge $Q$ is conserved. Unlike the original breakdown model in Ref.~\cite{Lian:2022nqj}, here we additionally allow the fermions to hop freely, but we continue to find that the model displays non-thermal behavior across a wide range of parameters. In the absence of disorder and an external magnetic field, we analytically show that the system displays \textit{quantum HSF} (i.e., fragmentation in an entangled basis, in the lexicon of Ref.~\cite{moudgalya2021hilbert}) for even $Q$ at any system size but only for odd chain lengths when $Q$ is odd; this model hence displays ergodicity breaking without disorder. We further analyze the model in the presence of a random magnetic field, which causes the dynamically disconnected Krylov subspaces to mix and thus the system to thermalize; however, as the strength of the random magnetic field is increased even further, we observe a crossover from this chaotic state to an MBL state. Thus, this model provides a rich playground for exploring the non-equilibrium dynamics of interacting quantum systems.

The rest of this paper is organized as follows: In Sec.~\ref{sec:model}, we introduce the model and its associated symmetries. We first show the presence of HSF in our model in the absence of fermionic hopping in Sec.~\ref{section: disconnected subspaces by blocking}. We then turn on the symmetric hopping term for fermions: in Sec.~\ref{section: Without magnetic field: Q=1} and~\ref{section: Without magnetic field: Q>=2}, we study the model in the absence of a magnetic field and show that it displays Hilbert space fragmentation. Here, analytically identify the one and two dimensional Krylov subspaces within the $Q=1$ symmetry sectors and further discuss the manifestation of HSF in symmetry sectors with $Q > 1$. In Sec.~\ref{section: With magnetic field}, we add a random magnetic field to the model, and numerically show that the system displays a crossover from chaotic to MBL behavior. We conclude with a summary of our results and some avenues for future research in Sec.~\ref{section: Discussion}.


\section{Extended Quantum Breakdown Model}
\label{sec:model}

In this paper, we study a 1D chain with two degrees of freedom on each site -- a spinless fermion and a spin-$1/2$ -- which, in the absence of any hopping for the fermions, maps onto a particular sector of the original quantum breakdown model introduced in Ref.~\cite{Lian:2022nqj} (see Appendix~\ref{appendix: the original breakdown model} for details). In the model we consider here, the fermions can move either through an ordinary symmetric hopping term or through a spatially asymmetric interaction that does not conserve the total number of spins. This extension of the original breakdown model thus allows us to investigate how the competition between the symmetric fermionic hopping and the asymmetric spin-fermion interaction influences the long-time dynamics. We discuss the model and its various symmetries below. Note that we restrict attention to open boundary condition (OBC) in this paper.

\subsection{Model}

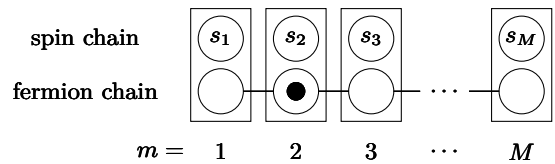
\begin{figure}[t]
    \centering
    \begin{tikzpicture}
\def\RectangleDistance{1}
\def\RectangleWidth{0.8}
\def\RectangleHeight{1.5}
\def\CircleDistance{0.7}
\def\CircleRadius{0.3}

\foreach \i in {0,1,2,4}{
    \draw (\i*\RectangleDistance-\RectangleWidth/2,0) rectangle (\i*\RectangleDistance+\RectangleWidth/2,\RectangleHeight);
    
    \draw (\i*\RectangleDistance, \RectangleHeight/2-\CircleDistance/2) circle (\CircleRadius);
    \draw (\i*\RectangleDistance,\RectangleHeight/2+\CircleDistance/2) circle (\CircleRadius);
    
    \foreach \i in {0,1,3}{
        \draw (\i*\RectangleDistance+\CircleRadius,\RectangleHeight/2-\CircleDistance/2) -- (\i*\RectangleDistance+\RectangleDistance-\CircleRadius,\RectangleHeight/2-\CircleDistance/2);
        }
    \draw (2*\RectangleDistance+\CircleRadius,\RectangleHeight/2-\CircleDistance/2) -- (3*\RectangleDistance-\CircleRadius-0.05,\RectangleHeight/2-\CircleDistance/2);
    
    \node at (\RectangleDistance*3,\RectangleHeight/2-\CircleDistance/2) {\small $\cdots$};

    \node at (\RectangleDistance*0,\RectangleHeight/2+\CircleDistance/2) {$s_1$};
    \node at (\RectangleDistance*1,\RectangleHeight/2+\CircleDistance/2) {$s_2$};
    \node at (\RectangleDistance*2,\RectangleHeight/2+\CircleDistance/2) {$s_3$};
    \node at (\RectangleDistance*4,\RectangleHeight/2+\CircleDistance/2) {$s_M$};

    \node at (-0.8,-0.4) {$m=$};
    \node at (\RectangleDistance*0,-0.4) {$1$};
    \node at (\RectangleDistance*1,-0.4) {$2$};
    \node at (\RectangleDistance*2,-0.4) {$3$};
    \node at (\RectangleDistance*3,-0.4) {$\cdots$};
    \node at (\RectangleDistance*4,-0.4) {$M$};

    \node at (-1.8,\RectangleHeight/2-\CircleDistance/2) {fermion chain};
    \node at (-1.8,\RectangleHeight/2+\CircleDistance/2) {spin chain};
    \fill (\RectangleDistance, \RectangleHeight/2-\CircleDistance/2) circle (3.5pt);
}
\end{tikzpicture}
\caption{The extended quantum breakdown model Eq.~\eqref{eqn: new hamiltonian} is defined on a chain of length $M$ where each site $m$ hosts a four-dimensional Hilbert space $\mathcal{H}_m = \mathbb{C}^2 \times \mathbb{C}^2$ consisting of a spinless fermion and a spin-$1/2$ degree of freedom.}
\label{fig: two chains}
\end{figure}

We consider a chain of length $M$ with total Hilbert space $\mathcal{H} = \otimes_{m=1}^M \mathcal{H}_m$, where the local Hilbert space $\mathcal{H}_m = \mathbb{C}^2 \otimes \mathbb{C}^2$ on site $m$ consists of a spin-$1/2$ and a spinless fermion, as illustrated in \figref{fig: two chains}. The Hamiltonian for our model is given by
\ie
\label{eqn: new hamiltonian}
H = H_J + H_\gamma + H_\mu + H_h 
.\fe
where the first term
\ie
H_J = \sum_{m=1}^{M} \left( J_m c_{m+1}^{\dagger} c_{m} \sigma_{m+1}^{+} + h.c. \right) 
,\fe
represents the spatially asymmetric spin-fermion interaction. Here, $\sigma_m^{\pm} = \sigma_m^x \pm i \sigma_m^y$ where $\sigma_m^\alpha$ are Pauli matrices acting on the spin on site $m$, and the fermionic creation/annihilation operators $c_m^\dagger, c_m$ obey canonical anti-commutation relations $\{c_m,c_{m^\prime}^\dagger\} = \delta_{m,m^\prime}$. This term permits fermions to hop between neighboring sites but only by simultaneously changing the local spin configuration on the right site. It resembles the breakdown phenomena in that a particle which is incident from the left can locally excite an atom \cite{Lian:2022nqj}, which we represent by a spin here for simplicity. For a similar interaction but for bosons, see Ref.~\cite{liu2023}. In this paper, we will consider both uniform ($J_m = J$) or random interaction strengths. The second term in the Hamiltonian
\ie
H_\gamma = - \sum_{m=1}^{M} \left( \gamma_m c_{m+1}^{\dagger} c_{m} + h.c. \right) 
,\fe
is the conventional nearest-neighbor hopping, which allows fermions to move without changing the spin configuration. We will consider both uniform hopping $\gamma_m = \gamma$ or random hopping strengths in this paper. The third term is the on-site potential for the fermions:
\ie
\label{eqn: H_mu}
H_\mu = \sum_{m=1}^{M} \mu_m \hat{n}_m\;,\;\;\; \hat{n}_m = c_{m}^{\dagger} c_m
.\fe
The on-site potential $\mu_m$ is a Gaussian random potential with mean value and variance given by
\ie
\label{eqn: mean std of mu}
\langle \mu_{m} \rangle = 0 \;, \;\;\; \langle \mu_{m}^2 \rangle = W^2
.\fe
Note that we do not have to separately consider $\langle \mu_{m} \rangle \ne 0$, since this would only induce a constant shift to the energy levels which can be eliminated through a re-definition of the ground state energy. Finally, we also consider the effect of a magnetic field:
\ie
H_h = \sum_{m=1}^{M} \vec{h}_m \cdot \vec{\sigma}_m
,\qquad \vec{h}_m=(h_m^x,h_m^y,h_m^z)\ ,\fe
where the field $\vec{h}_m$ represents an external magnetic field that couples to the spins. In this paper, the field $\vec{h}_m$ can either be uniform ($\vec{h}_m = \vec{h}$) or be generated from a Gaussian distribution
\ie
\langle h_{m}^i \rangle = \bar{h}^i \;, \;\;\; \langle \left( h_{m}^i - \bar{h}^i \right)^2 \rangle = \left( \Delta h_{m}^i \right)^2
.\fe
We note that only the terms $H_J$ and $H_\mu$ are present in the original breakdown model (see Appendix~\ref{appendix: the original breakdown model} for details on the mapping between the two models), which we have extended here by also including the terms $H_\gamma$ and $H_h$.

In what follows, we will denote terms in the Hamiltonian with uniform coefficients with a bar over the subscript, such as $H_{\bar{\gamma}}$ for uniform hopping $\gamma_m=\gamma$. Rather than exploring the entire four dimensional parameter space, in this paper we will restrict attention to certain subspaces that illustrate the rich dynamics of the model in Eq.~\eqref{eqn: new hamiltonian}; specifically, we consider the following Hamiltonians:
\begin{align}
\label{eq:hams}
H_{J \mu h} &= H_J + H_\mu + H_h \, , \nonumber \\
H_{\gamma J} &= H_\gamma + H_J \, , \nonumber \\
H_{\gamma J \mu} &= H_{\gamma} + H_{J} + H_\mu \, , \nonumber \\
H_{\gamma J h} &= H_\gamma + H_J + H_h \, .
\end{align}

\subsection{Symmetries}

Let us now turn to the symmetries of the Hamiltonian Eq.~\eqref{eqn: new hamiltonian} (which are also thus symmetries of each of the Hamiltonians in Eq.~\eqref{eq:hams}). While the extended breakdown model does not conserve the total spin, the total charge
\ie
\label{eqn: conserved charge}
Q = \sum_{m=1}^{M}\hat{n}_m
,\fe
which counts the total number of fermions, is conserved. This enables us to block diagonalize the Hamiltonian into charge $Q$ sectors and then analyze each $Q$ sector separately. The total Hilbert space dimension for a fixed charge $Q$ sector is $|\mathcal{H}_Q|=\text{dim}(\mathcal{H}_Q) = 2^M C^{M}_{Q}$, where $C^M_Q = \binom{M}{Q} =\frac{M!}{Q!(M-Q)!}$. Here, the factor of $2^M$ comes from the spin-$1/2$ per site, while the additional factor corresponds to filling a chain of length $M$ with $Q$ spinless fermions.

We can also define a charge conjugation transformation $\hat{C}$ as follows:
\ie
\hat{C}
&= \hat{C}_{\text{fermion}} \otimes \hat{C}_{\text{spin}}
\\
&= \left( \prod_{m=1}^{M} (c_m + c_m^\dagger) \right) 
\otimes \left( \prod_{m=1}^{M} \sigma_m^x \right)
.\fe
The charge conjugation operator $\hat{C}$ has the following properties:
\begin{itemize}
    \item It is its own inverse:
    \ie
        \hat{C} \hat{C} = \mathds{1} \Rightarrow \hat{C}^{-1} = \hat{C}
    .\fe
    \item It interchanges the roles of the fermions and the holes on the fermion chain:
    \ie
        \hat{C} c_m \hat{C}^{-1} = c_m^\dagger
        \\
        \hat{C} c_m^\dagger \hat{C}^{-1} = c_m
    .\fe
    \item It flips spin-up (spin-down) components to spin-down (spin-up), and changes the Pauli matrices in the Hamiltonian by:
    \ie
        \hat{C} \sigma_m^{x} \hat{C}^{-1} &= \sigma_m^{x}
        \\
        \hat{C} \sigma_m^{y} \hat{C}^{-1} &= -\sigma_m^{y}
        \\
        \hat{C} \sigma_m^{z} \hat{C}^{-1} &= -\sigma_m^{z}
        \\
        \hat{C} \sigma_m^{\pm} \hat{C}^{-1} &= \sigma_m^{\mp}
    \fe
\end{itemize}
In this paper, we will focus on Hamiltonians that are symmetric around $h_m^x=h_m^y=0$. Examples include $h_m^x=h_m^y=0$ and $h_m^x,h_m^y\sim \mathcal{N}(0,\sigma)$ (here, $\mathcal{N}$ denotes the normal distribution). Under $\hat{C}$, the Hamiltonian transforms as
\ie
\hat{C} H \hat{C}^{-1}
= \left(\sum_{m}\mu_m \right) - H_{\mu} - H_{\gamma} - H_J - H_h
=
-H
,\fe
which follows from Eq.~\eqref{eqn: mean std of mu}, while the charge $\mathcal{Q}$ transforms as
\ie
\hat{C} Q \hat{C}^{-1}
= M - Q
,\fe
such that for a simultaneous eigenstate $\ket{\Psi}$ of $H, Q$ with eigenvalues $E,q$, there exists another eigenstate $\left(\hat{C} \ket{\Psi} \right)$ with eigenvalues $-E, M- q$. In other words, the full spectrum of all charge sectors is particle-hole symmetric around $E = 0$ and, additionally, the energy spectrum in the charge $Q$ sector is opposite to that of the charge $M - Q$ sector.

Finally, the Hamiltonian Eq.~\eqref{eqn: new hamiltonian} hosts an additional SO(2) symmetry of the left boundary spin degree of freedom. Specifically, the operator $\vec{\sigma}_{m=1}$ is decoupled from all other operators, which we can see by first re-writing $H$ as
\ie
H = \vec{h}_1\cdot \vec{\sigma}_1 + H^\prime(\left\{ c_{m}, c_{m}^\dagger, \vec{\sigma}_{m^\prime} \right\})
,\fe
where $m=1, \dots, M$, $m^\prime = 2, \dots, M$. Note that $H^\prime$ does not act on the spin on the first site and hence trivially conserves it. The presence of this extra boundary symmetry allows us to block-diagonalize each $Q$-sector into two disconnected blocks, distinguished by the spin configuration on the first site. The total Hilbert space therefore decomposes as
\ie
\label{eqn: 2-block Hilbert space fragmentation}
\mathcal{H} = \mathcal{H}_{s_1} \otimes \mathcal{H}_{\text{sub}}
,\fe
where $\mathcal{H}_{s_1}$ is the two-dimensional Hilbert space for the first spin, and $\mathcal{H}_{\text{sub}}$ is the Hilbert space spanned by the remaining degrees of freedom. This Hilbert space $\mathcal{H}_{\text{sub}}$ is a tensor product Hilbert space of $M$ spinless fermion chain and $(M-1)$ spin-$1/2$ states:
\ie
\mathcal{H}_{\text{sub}} = \mathcal{H}_{\text{fermion}} \otimes \mathcal{H}_{s_2} \otimes \cdots \otimes \mathcal{H}_{s_M}
,\fe
and the states in $\mathcal{H}_{\text{sub}}$ can be expressed in terms of linear combinations of the following product states:
\ie
c_{f_Q}^\dagger \cdots c_{f_2}^\dagger c_{f_1}^\dagger | \Omega \rangle \otimes | s_2 \rangle \otimes \cdots \otimes | s_M \rangle
,\fe
where $f_1<f_2<\cdots<f_Q$. The state $|\Omega\rangle$ denotes the vacuum states of the fermion chain (no fermion), and $| s_m \rangle$ denotes the spin state on site $m$. The state $| s_m \rangle$ can be represented as
\ie
\ket{s_m} 
= \alpha  \ket{\uparrow} 
+ \beta \ket{\downarrow} 
\mapsto
\left(
\begin{matrix}
\alpha
\\
\beta
\end{matrix}
\right)_m
,\fe
or, in the computational basis 
\ie
\ket{s_m}
= \alpha  \ket{\uparrow} 
+ \beta \ket{\downarrow}
= \alpha  \ket{1} 
+ \beta \ket{0}
.\fe
In the rest of the paper, if not explicitly mentioned otherwise, the Hilbert space refers to $\mathcal{H}_{\text{sub}}$ with $|s_1\rangle=\ket{\uparrow}$, denoted as $\mathcal{H}_{\text{sub}}^{(\uparrow)}$, which has dimension $2^{M-1} C^{M}_{Q}$ in a given charge $Q$ sector. Additionally, we adopt the notation $s_m=1, 0$ instead of $s_m=\uparrow, \downarrow$ for consistency.

Typically, for a non-integrable Hamiltonian such as Eq.~\eqref{eqn: new hamiltonian}, once all global symmetries are resolved, one expects chaotic behavior in each symmetry sector. More specifically, for an ergodic system with some symmetry group $G$ and with $\ket{\Psi}$ a random many-body state of the symmetry, one expects that $\{H^n \ket{\phi} | n \in \mathbb{Z}_{\geq 0}\}$ ($\mathbb{Z}_{\geq 0}$ stands for positive integers) spans a basis for all states with the same symmetry quantum number as $\ket{\Psi}$. However, recent work has surprisingly revealed the phenomenon of Hilbert space fragmentation~\cite{sala2019ergodicity,khemani2020hilbert,rakovszky2019statistical,moudgalya2019thermalization,sanjay_review}, where the Hilbert space within each symmetry sector further fractures into dynamically disconnected Krylov subspaces, defined as follows: given a Hamiltonian $H$ acting on an $\mathcal{N}$-dimensional Hilbert space, a Krylov subspace $\mathcal{K}$ is the space spanned by all states $\{H^n \ket{\Psi} | n \in \mathbb{Z}_{\geq 0}\}$. If $n_K$ of them are linearly independent, we say the Krylov subspace has dimension $n_K$. If $n_K < \mathcal{N}$, we say the Krylov subspace $\mathcal{K}$ is non-trivial, while a Krylov subspace that spans the full Hilbert space (within the symmetry sector) is referred to as trivial.

Consider a given symmetry sector $\mathcal{H}_Q$ that fractures into $I$ distinct Krylov subspaces, where degenerate subspaces are treated as indistinct. We denote the degeneracy of the $i$-th subspace by $d_i$, and label each Krylov subspace by $\mathcal{K}_{i, \delta}^{(n_i)}$ ($i \in [1, I]$, $\delta \in [1, d_i]$). The superscript $n_i$ represents the dimension of the $i$-th (distinct) Krylov subspace: $n_i \equiv |\mathcal{K}_{i, \delta}^{(n_i)}|$. The number of $n$-dimensional Krylov subspaces is the sum of all $d_i$ with $n_i = n$:
\ie
D(\mathcal{K}^{(n)})
= \sum_{i | n_i=n} d_i
.\fe
The total number of all Krylov subspaces in the symmetry sector is denoted as:
\begin{equation}
D_{\text{tot}}(\mathcal{K}) \equiv \sum_n D(\mathcal{K}^{(n)})
= \sum_{i=1}^{I} d_i
.\end{equation}

Finally, we note that Hilbert space fragmentation can be categorized into so-called \emph{strong fragmentation} and \emph{weak fragmentation}. The former violates weak ETH, while the later violates strong ETH: the former holds that, for an ergodic system, most (aside from a measure zero set) eigenstates in the bulk of the energy spectrum behave as thermal states (as far as expectation values of local observables are concerned), while the latter holds that \textit{all} eigenstates behave thermally. Weak and strong fragmentation can thus be distinguished by the ratio of the dimension of the largest Krylov subspace $|\mathcal{K}^{\text{max}}|$ and the dimension of the total Hilbert space $|\mathcal{H}_Q|$ in the thermodynamic limit (system size $M\rightarrow \infty$):
\ie
\lim_{M\rightarrow \infty} \frac{|\mathcal{K}^{\text{max}}|}{|\mathcal{H}_Q|}
\begin{cases}
\rightarrow 0 & \text{: strong fragmentation}
\\
\rightarrow c & \text{: weak fragmentation}
\end{cases}
,\fe
where $0<c\le 1$. By definition, a weakly fragmented symmetry sector can only host a measure zero set of finite-dimensional closed Krylov subspaces, and therefore violates the strong form of the ETH. In the rest of this paper, we will discuss the conditions under which the extended breakdown model also exhibits Hilbert space fragmentation, along with the structure of its maximal Krylov subspaces.

\subsection{Notation}
\label{section: Notations}

Here, we introduce some notation that will prove convenient for analytically deriving the zero modes and Krylov subspaces.

\subsubsection{Hamiltonian Decomposition}

Let us begin with the Hamiltonian $H_{\gamma J}$ in \cref{eq:hams} (with a vanishing magnetic field and on-site potential), which can be decomposed as:
\ie
H_{\gamma J} 
= \sum_{m=1}^{M} \hat{h}_{m}
= \sum_{m=1}^{M} \left( \hat{h}_{m}^{(L)} + \hat{h}_{m}^{(R)} \right)
,\fe
where
\ie
\label{eqn: decomposition of the Hamiltonian}
&\hat{h}_{m}^{(L)} = - \gamma_{m-1} c_{m-1}^\dagger c_m + J_{m-1} c_{m-1}^\dagger c_m \sigma_{m}^{-}\ ,
\\
&\hat{h}_{m}^{(R)} = - \gamma_m c_{m+1}^\dagger c_m + J_m c_{m+1}^\dagger c_m \sigma_{m+1}^{+}
\ ,
\\
&\hat{h}_m =\hat{h}_{m}^{(L)}+\hat{h}_{m}^{(R)}\ .\fe
We find that this decomposition makes it easier to track the positions occupied by the fermions and facilitates a straightforward identification of the Krylov subspaces.

From the above definitions, we see that the operator $\hat{h}_{m}^{(L)}$ moves a fermion on site $m$ to site $(m-1)$ and includes a term that lowers the spin state $|s_m\rangle$ on site $m$. Similarly, $\hat{h}_{m}^{(R)}$ moves a fermion on site $m$ to site $(m+1)$ and includes a term that raises the spin state $|s_{m+1}\rangle$ on site $m+1$. With OBC, $\hat{h}_{1}^{(L)} = \hat{h}_{M}^{(R)} = 0$. Now, if we act with this Hamiltonian $H_{\gamma J}$ on a state with fermions occupying sites $\{f_1, f_2, \cdots, f_Q\}$, the components of the resulting state will generically have fermions on sites $\{f_1-1, f_1+1, f_2-1, f_2+1, \cdots, f_Q-1, f_Q+1\}$ (unless different components accidentally cancel each other).

\subsubsection{Notation for $H_{\gamma J}$: $\mathcal{K}^{(1)}$}
\label{subsubsection: Notations for K1}

An equivalent formulation of our model is that of two coupled systems: a 1D fermionic chain with hopping $\gamma_m$ and a spin chain subjected to magnetic field $\vec{h}_m$, where the two chains couple to each other through the asymmetric interaction with strength $J_m$. A useful graphical notation for Fock states in the Hilbert space $\mathcal{H}_{\text{sub}}^{(\uparrow)}$ is:
\begin{alignat}{1}
\label{eqn: graphical notation}
&| f_1, f_2, \cdots, f_Q ; s_1=1, s_2 \cdots s_M \rangle  \nonumber
\\
& = c_{f_Q}^\dagger \cdots c_{f_2}^\dagger c_{f_1}^\dagger | \Omega \rangle \otimes | s_1=1 \rangle \otimes | s_2 \rangle \otimes \cdots \otimes | s_M \rangle
\notag \\ \notag\\
&=\;\;\;
\begin{tikzpicture}[baseline={(0, -0.1)}]
\draw (0.0, 0) circle (.2);
\draw (0.8, 0) circle (.2);
\draw (2.4, 0) circle (.2);
\draw (4, 0) circle (.2);
\draw (5.6, 0) circle (.2);
\draw (6.4, 0) circle (.2);
\fill (2.4, 0) circle (2pt);
\fill (4.0, 0) circle (2pt);
\node at (0.0, -0.4) {\small $1$};
\node at (0.8, -0.4) {\small $2$};
\node at (2.4, -0.4) {\small $f_1$};
\node at (4.0, -0.4) {\small $f_2$};
\node at (6.4, -0.4) {\small $M$};
\node at (0.4, -0.2) {\small $s_2$};
\node at (1.2, -0.2) {\small $s_3$};
\node at (6.0, -0.2) {\small $s_M$};
\draw (0.2, 0) -- (0.6, 0);
\draw (1.0, 0) -- (1.3, 0);
\draw (1.9, 0) -- (2.2, 0);
\draw (2.6, 0) -- (2.9, 0);
\draw (3.5, 0) -- (3.8, 0);
\draw (4.2, 0) -- (4.5, 0);
\draw (5.1, 0) -- (5.4, 0);
\draw (5.8, 0) -- (6.2, 0);
\node at (1.6, 0) {$\cdots$};
\node at (3.2, 0) {$\cdots$};
\node at (4.8, 0) {$\cdots$};
\end{tikzpicture}
\notag\\
&=\;\;\;
c_{f_Q}^\dagger \cdots c_{f_2}^\dagger c_{f_1}^\dagger | 0 \rangle
\otimes
\begin{tikzpicture}[baseline={(0, -0.1)}]
\draw (0.0, 0) circle (.2);
\draw (0.8, 0) circle (.2);
\draw (2.4, 0) circle (.2);
\node at (0.0, -0.4) {\small $1$};
\node at (0.8, -0.4) {\small $2$};
\node at (2.4, -0.4) {\small $M$};
\node at (0.4, -0.2) {\small $s_2$};
\node at (1.2, -0.2) {\small $s_3$};
\node at (2.0, -0.2) {\small $s_M$};
\draw (0.2, 0) -- (0.6, 0);
\draw (1.0, 0) -- (1.3, 0);
\draw (1.9, 0) -- (2.2, 0);
\node at (1.6, 0) {$\cdots$};
\end{tikzpicture}
.\end{alignat}
Here, the indices $f_m = 1, \cdots, M$ denote the occupied sites on the fermion chain, and $s_m = \{1, 0\}$ denote the spin on the site $m \in \left[2, M \right]$. Recall that in our definition of $\mathcal{H}_{\text{sub}}^{(\uparrow)}$, the spin on the first site is fixed $s_1=1$, and we thus exclude $s_1$ in our graphical notation. Without loss of generality, we always require $f_1 < f_2 < \cdots < f_Q$. Graphically, the unoccupied sites on the fermion chain are denoted by an empty circle: $\begin{tikzpicture}[baseline={(0, -0.1)}]\draw (0.0, 0) circle (0.16);\end{tikzpicture}$, while the occupied sites are denoted by a filled circle: $\begin{tikzpicture}[baseline={(0, -0.1)}]\draw (0.0, 0) circle (0.16); \fill (0.0, 0) circle (1.6pt);\end{tikzpicture}$. In addition, the spin components are represented by the links connecting the circles, and the states are indicated by $s_m$ below the edges. For simplicity, we will sometimes omit the on-site indices and $s_m$ when there is no ambiguity.

Both $H_\gamma$ and $H_J$ enable the fermions to hop between neighboring sites; however, the former alters the spin states, while the later doesn't. To better keep track of the spin states, we define additionally 4 spin states with arrows ($s_m=0, 1$):
\begin{alignat}{2}
\label{eqn:spin-states-notation}
& 
\begin{tikzpicture}[baseline={(0, -0.1)}]
\draw (0.4, 0) circle (.2);
\draw (1.2, 0) circle (.2);
\draw (0.6, 0) -- (1.0, 0);
\node at (0, 0) {$\cdots$};
\node at (1.7, 0) {$\cdots$};
\node at (0.8, -0.2) {\small $s_m$};
\node at (1.2, -0.4) {\small $m$};
\draw[<-] (0.65, 0.15) -- (0.95, 0.15);
\end{tikzpicture}
,
\begin{tikzpicture}[baseline={(0, -0.1)}]
\draw (0.4, 0) circle (.2);
\draw (1.2, 0) circle (.2);
\draw (0.6, 0) -- (1.0, 0);
\node at (0, 0) {$\cdots$};
\node at (1.7, 0) {$\cdots$};
\node at (1.2, -0.4) {\small $m$};
\node at (0.8, -0.2) {\small $s_m$};
\draw[->] (0.65, 0.15) -- (0.95, 0.15);
\end{tikzpicture}
,\end{alignat}
such that when acted upon by the following terms of the Hamiltonian $H_{\gamma J}$, they transform into the eigenstates of $\sigma_{m}^{z}$:
\begin{alignat}{2}
\label{eqn: H acting on arrow states}
\hat{h}_{m-1}^{(R)} \;
\begin{tikzpicture}[baseline={(0, -0.1)}]
\draw (0.4, 0) circle (.2);
\draw (1.2, 0) circle (.2);
\fill (0.4, 0) circle (2pt);
\draw (0.6, 0) -- (1.0, 0);
\node at (0, 0) {$\cdots$};
\node at (1.7, 0) {$\cdots$};
\node at (0.8, -0.2) {\small $s_{m}$};
\node at (1.2, -0.4) {\small $m$};
\draw[<-] (0.65, 0.15) -- (0.95, 0.15);
\end{tikzpicture}
& =
\begin{tikzpicture}[baseline={(0, -0.1)}]
\draw (0.4, 0) circle (.2);
\draw (1.2, 0) circle (.2);
\fill (1.2, 0) circle (2pt);
\draw (0.6, 0) -- (1.0, 0);
\node at (0, 0) {$\cdots$};
\node at (1.7, 0) {$\cdots$};
\node at (0.8, -0.2) {\small $s_m$};
\node at (1.2, -0.4) {\small $m$};
\end{tikzpicture}
\notag\\
\hat{h}_{m}^{(L)} \;
\begin{tikzpicture}[baseline={(0, -0.1)}]
\draw (0.4, 0) circle (.2);
\draw (1.2, 0) circle (.2);
\fill (1.2, 0) circle (2pt);
\draw (0.6, 0) -- (1.0, 0);
\node at (0, 0) {$\cdots$};
\node at (1.7, 0) {$\cdots$};
\node at (1.2, -0.4) {\small $m$};
\node at (0.8, -0.2) {\small $s_m$};
\draw[->] (0.65, 0.15) -- (0.95, 0.15);
\end{tikzpicture}
& =
\begin{tikzpicture}[baseline={(0, -0.1)}]
\draw (0.4, 0) circle (.2);
\draw (1.2, 0) circle (.2);
\fill (0.4, 0) circle (2pt);
\draw (0.6, 0) -- (1.0, 0);
\node at (0, 0) {$\cdots$};
\node at (1.7, 0) {$\cdots$};
\node at (1.2, -0.4) {\small $m$};
\node at (0.8, -0.2) {\small $s_m$};
\end{tikzpicture}
.\end{alignat}
We say the spin states are $| \overleftarrow{s_m} \rangle$ or $| \overrightarrow{s_m} \rangle$, and the explicit form of the spin states with arrows in \cref{eqn:spin-states-notation} are given in Appendix~\ref{appendix: Explicit forms of the notations for K1}. Note that spin states with arrows are built on the eigenstates of $\sigma_{m}^{z}$, i.e., $| s_m=0 \rangle$ and $| s_m=1 \rangle$, so as to satisfy Eq.~\eqref{eqn: H acting on arrow states}. An intuitive way to understand these states is to imagine that when a fermion moves across the edge, it "pushes against" the arrow in the opposite direction and cancels the arrow. This notation will help us identify the invariant subspaces in Sec.~\ref{section: Without magnetic field: Q=1} and Sec.~\ref{section: Without magnetic field: Q>=2}.

\begin{widetext}

\subsubsection{Notation for $H_{\bar{\gamma}\bar{J}}$: $\mathcal{K}^{(2)}$}
\label{subsubsection: Notations for K2}

For models with uniform hopping coefficients $\gamma_m=\gamma$ and interactions $J_m=J$, we define the singly neutral-/left-/right-rectangularized spin states on a block of three consequtive sites $m-1,m,m+1$ (which determine the spin states $(s_m, s_{m+1})$):
\begin{alignat}{2}
\label{eqn: 3 kinds of rectangularized spin states}
\begin{cases}
    &\begin{tikzpicture}[baseline={(0, -0.1)}]
    \draw (0.0, 0.25) rectangle (2.4, -0.25);
    \draw (0.4, 0) circle (.2);
    \draw (1.2, 0) circle (.2);
    \draw (2.0, 0) circle (.2);
    \draw (0.6, 0) -- (1.0, 0);
    \draw (1.4, 0) -- (1.8, 0);
    \node at (0, 0) {$\cdots$};
    \node at (2.5, 0) {$\cdots$};
    \node at (1.2, 0.4) {\small $r_j$};
    \node at (1.2, -0.4) {\small $m$};
    \end{tikzpicture}
    \text{: neutral-rectangularized spin states}
    ,
    \\
    &\begin{tikzpicture}[baseline={(0, -0.1)}]
    \draw (0.0, 0.25) rectangle (2.4, -0.25);
    \draw (0.4, 0) circle (.2);
    \draw (1.2, 0) circle (.2);
    \draw (2.0, 0) circle (.2);
    \draw (0.6, 0) -- (1.0, 0);
    \draw (1.4, 0) -- (1.8, 0);
    \draw[<-] (0.65, 0.4) -- (0.95, 0.4);
    \node at (0, 0) {$\cdots$};
    \node at (2.5, 0) {$\cdots$};
    \node at (1.2, 0.4) {\small $r_j$};
    \node at (1.2, -0.4) {\small $m$};
    \end{tikzpicture}
    \text{: left-rectangularized spin states}
    ,
    \\
    &\begin{tikzpicture}[baseline={(0, -0.1)}]
    \draw (0.0, 0.25) rectangle (2.4, -0.25);
    \draw (0.4, 0) circle (.2);
    \draw (1.2, 0) circle (.2);
    \draw (2.0, 0) circle (.2);
    \draw (0.6, 0) -- (1.0, 0);
    \draw (1.4, 0) -- (1.8, 0);
    \draw[->] (1.45, 0.4) -- (1.75, 0.4);
    \node at (0, 0) {$\cdots$};
    \node at (2.5, 0) {$\cdots$};
    \node at (1.2, 0.4) {\small $r_j$};
    \node at (1.2, -0.4) {\small $m$};
    \end{tikzpicture}
    \text{: right-rectangularized spin states}
    .
\end{cases}
\end{alignat}
the explicit forms of which are given in Appendix~\ref{appendix: Explicit forms of the notations for K2}, such that the following properties are satisfied:
\begin{enumerate}
\item When acted by the Hamiltonian once, the neutral-rectangularized spin states transform as:
\begin{alignat}{1}
\label{eqn: neutral-rectangular acted by the Hamiltonian once}
\hat{h}_m \;
\begin{tikzpicture}[baseline={(0, -0.1)}]
\draw (0.0, 0.25) rectangle (2.4, -0.25);
\draw (0.4, 0) circle (.2);
\draw (1.2, 0) circle (.2);
\draw (2.0, 0) circle (.2);
\fill (1.2, 0) circle (2pt);
\draw (0.6, 0) -- (1.0, 0);
\draw (1.4, 0) -- (1.8, 0);
\node at (0, 0) {$\cdots$};
\node at (2.5, 0) {$\cdots$};
\node at (1.2, 0.4) {\small $r_j$};
\node at (1.2, -0.4) {\small $m$};
\end{tikzpicture}
&=
\begin{tikzpicture}[baseline={(0, -0.1)}]
\draw (0.0, 0.25) rectangle (2.4, -0.25);
\draw (0.4, 0) circle (.2);
\draw (1.2, 0) circle (.2);
\draw (2.0, 0) circle (.2);
\fill (0.4, 0) circle (2pt);
\draw (0.6, 0) -- (1.0, 0);
\draw (1.4, 0) -- (1.8, 0);
\draw[<-] (0.65, 0.4) -- (0.95, 0.4);
\node at (0, 0) {$\cdots$};
\node at (2.5, 0) {$\cdots$};
\node at (1.2, 0.4) {\small $r_j$};
\end{tikzpicture}
+
\begin{tikzpicture}[baseline={(0, -0.1)}]
\draw (0.0, 0.25) rectangle (2.4, -0.25);
\draw (0.4, 0) circle (.2);
\draw (1.2, 0) circle (.2);
\draw (2.0, 0) circle (.2);
\fill (2.0, 0) circle (2pt);
\draw (0.6, 0) -- (1.0, 0);
\draw (1.4, 0) -- (1.8, 0);
\draw[->] (1.45, 0.4) -- (1.75, 0.4);
\node at (0, 0) {$\cdots$};
\node at (2.5, 0) {$\cdots$};
\node at (1.2, 0.4) {\small $r_j$};
\end{tikzpicture}
;\end{alignat}
while the left-/right-rectangularized spin states transform as
\begin{alignat}{2}
\label{eqn: LR-rectangular acted by the Hamiltonian once}
\hat{h}_{m-1}^{(R)} \;
\begin{tikzpicture}[baseline={(0, -0.1)}]
\draw (0.0, 0.25) rectangle (2.4, -0.25);
\draw (0.4, 0) circle (.2);
\draw (1.2, 0) circle (.2);
\draw (2.0, 0) circle (.2);
\fill (0.4, 0) circle (2pt);
\draw (0.6, 0) -- (1.0, 0);
\draw (1.4, 0) -- (1.8, 0);
\draw[<-] (0.65, 0.4) -- (0.95, 0.4);
\node at (0, 0) {$\cdots$};
\node at (2.5, 0) {$\cdots$};
\node at (1.2, 0.4) {\small $r_j$};
\node at (1.2, -0.4) {\small $m$};
\end{tikzpicture}
& \in
\text{Span}\left(
\begin{tikzpicture}[baseline={(0, -0.1)}]
\draw (0.0, 0.25) rectangle (2.4, -0.25);
\draw (0.4, 0) circle (.2);
\draw (1.2, 0) circle (.2);
\draw (2.0, 0) circle (.2);
\fill (1.2, 0) circle (2pt);
\draw (0.6, 0) -- (1.0, 0);
\draw (1.4, 0) -- (1.8, 0);
\node at (0, 0) {$\cdots$};
\node at (2.5, 0) {$\cdots$};
\node at (1.2, 0.4) {\small $0$};
\end{tikzpicture}
,
\begin{tikzpicture}[baseline={(0, -0.1)}]
\draw (0.0, 0.25) rectangle (2.4, -0.25);
\draw (0.4, 0) circle (.2);
\draw (1.2, 0) circle (.2);
\draw (2.0, 0) circle (.2);
\fill (1.2, 0) circle (2pt);
\draw (0.6, 0) -- (1.0, 0);
\draw (1.4, 0) -- (1.8, 0);
\node at (0, 0) {$\cdots$};
\node at (2.5, 0) {$\cdots$};
\node at (1.2, 0.4) {\small $1$};
\end{tikzpicture}
\right)
,
\\ \notag
\hat{h}_{m+1}^{(L)} \;
\begin{tikzpicture}[baseline={(0, -0.1)}]
\draw (0.0, 0.25) rectangle (2.4, -0.25);
\draw (0.4, 0) circle (.2);
\draw (1.2, 0) circle (.2);
\draw (2.0, 0) circle (.2);
\fill (2.0, 0) circle (2pt);
\draw (0.6, 0) -- (1.0, 0);
\draw (1.4, 0) -- (1.8, 0);
\draw[->] (1.45, 0.4) -- (1.75, 0.4);
\node at (0, 0) {$\cdots$};
\node at (2.5, 0) {$\cdots$};
\node at (1.2, 0.4) {\small $r_j$};
\node at (1.2, -0.4) {\small $m$};
\end{tikzpicture}
& \in
\text{Span}\left(
\begin{tikzpicture}[baseline={(0, -0.1)}]
\draw (0.0, 0.25) rectangle (2.4, -0.25);
\draw (0.4, 0) circle (.2);
\draw (1.2, 0) circle (.2);
\draw (2.0, 0) circle (.2);
\fill (1.2, 0) circle (2pt);
\draw (0.6, 0) -- (1.0, 0);
\draw (1.4, 0) -- (1.8, 0);
\node at (0, 0) {$\cdots$};
\node at (2.5, 0) {$\cdots$};
\node at (1.2, 0.4) {\small $0$};
\end{tikzpicture}
,
\begin{tikzpicture}[baseline={(0, -0.1)}]
\draw (0.0, 0.25) rectangle (2.4, -0.25);
\draw (0.4, 0) circle (.2);
\draw (1.2, 0) circle (.2);
\draw (2.0, 0) circle (.2);
\fill (1.2, 0) circle (2pt);
\draw (0.6, 0) -- (1.0, 0);
\draw (1.4, 0) -- (1.8, 0);
\node at (0, 0) {$\cdots$};
\node at (2.5, 0) {$\cdots$};
\node at (1.2, 0.4) {\small $1$};
\end{tikzpicture}
\right)
.\end{alignat}
$\text{Span}(v_1, v_2, \cdots, v_n)$ denotes the vector space spanned by the vectors $v_1, \cdots v_n$, and $v \in \text{Span}(v_1, \cdots, v_n)$ implies that the vector $v$ is a linear combination of $v_1, \cdots v_n$. The exact coefficients of the linear combinations in Eq.~\eqref{eqn: LR-rectangular acted by the Hamiltonian once} are provided in \cref{eqn: The coefficients for spanned space}. The labels $r_j$ above the rectangles can take 2 values: $0$ and $1$. The subscript $j$ means it labels the $j$-th rectangle on the spin chain. We refer to those labels as \emph{rectangular labels}.

\item When acted by the Hamiltonian twice:
\begin{alignat}{1}
\label{eqn: H^2 acting on rectangularized states}
    \left( \hat{h}_{m-1}^{(R)} \hat{h}_m^{(L)} + \hat{h}_{m+1}^{(L)} \hat{h}_m^{(R)} \right) \;
    \begin{tikzpicture}[baseline={(0, -0.1)}]
    \draw (0.0, 0.25) rectangle (2.4, -0.25);
    \draw (0.4, 0) circle (.2);
    \draw (1.2, 0) circle (.2);
    \draw (2.0, 0) circle (.2);
    \fill (1.2, 0) circle (2pt);
    \draw (0.6, 0) -- (1.0, 0);
    \draw (1.4, 0) -- (1.8, 0);
    \node at (0, 0) {$\cdots$};
    \node at (2.5, 0) {$\cdots$};
    \node at (1.2, 0.4) {\small $r_j$};
    \node at (1.2, -0.4) {\small $m$};
    \end{tikzpicture}
    =
    (2 \gamma^2 + J^2)
    \begin{tikzpicture}[baseline={(0, -0.1)}]
    \draw (0.0, 0.25) rectangle (2.4, -0.25);
    \draw (0.4, 0) circle (.2);
    \draw (1.2, 0) circle (.2);
    \draw (2.0, 0) circle (.2);
    \fill (1.2, 0) circle (2pt);
    \draw (0.6, 0) -- (1.0, 0);
    \draw (1.4, 0) -- (1.8, 0);
    \node at (0, 0) {$\cdots$};
    \node at (2.5, 0) {$\cdots$};
    \node at (1.2, 0.4) {\small $r_j$};
    \node at (1.2, -0.4) {\small $m$};
    \end{tikzpicture}
.\end{alignat}
The detailed calculation can be found in \cref{eqn: The coefficients for spanned space}.
\end{enumerate}

The above properties of the singly rectangularized spin states in \cref{eqn: 3 kinds of rectangularized spin states} make them important three-site building blocks of the 2-dimensional Krylov subspaces $\mathcal{K}^{(2)}$, discussed in Sec.~\ref{section: Without magnetic field: Q=1}. In this paper, unless otherwise specified, we use the terms "singly rectangulared states" and "rectangulared states" interchangeably. The singly rectangulared states are different from the doubly rectangulared states that are defined in Appendix~\ref{appendix: Q=1, odd M, K=4}.

\end{widetext}


\section{$H_{J \mu h}$: Disconnected subspaces by blocking}
\label{section: disconnected subspaces by blocking}

We begin by considering the model in the absence of any hopping for the fermions (namely, $\gamma_m=0$) and in-plane magnetic fields $h_m^x=h_m^y=0$, and show that a given symmetry sector can be partitioned into an extensive number of smaller, dynamically disconnected sectors. The Hamiltonian  in this case consists of the spin-fermion interaction $H_J$, the on-site potential $H_\mu$, and a magnetic field in the $z$-direction:
\ie
\label{eq:HJmuh}
H = H_J + H_\mu + \sum_{m} h_{m}^{z} \sigma_{m}^{z}
.\fe

This Hamiltonian has an extensive number of disconnected subspaces. To see this, consider a root state $|\Psi\rangle$ with fermions on sites $f_1<f_2 < \cdots< f_Q$ and the spin state $s_{m_b} = \ket{1}$ (in the $z$-basis) on a site $m_b > f_Q$. The Hamiltonian cannot change the spin states on sites $m \ge m_b$, since no fermion can hop to a site $m \ge m_b$ via the spin-fermion interaction. Therefore, the spin state $s_{m_b}$ acts as a dynamical blockade, and the spin states $s_m$ with $m \ge m_b$ are dynamically frozen. On the other hand, if instead $s_{m_b} = \ket{0}$ with $m_b \le f_1$, we find that similarly $s_{m_b}$ acts as a dynamical blockade, and the spin states on sites $m \le m_b$ are frozen. The disconnected subspace(s) generated by these root states thus cannot span the Hilbert space (of dimension $\mathcal{N}$) of a fixed $Q$ sector, and there must exist an integer $n < \mathcal{N}$ such that $H^{n} | \Psi \rangle$ is linearly dependent on $\left\{ | \Psi \rangle, H | \Psi \rangle, \cdots, H^{n-1} | \Psi \rangle \right\}$.

The disconnected subspaces can be partially understood by observing that, in addition to Eq.~\eqref{eqn: conserved charge}, the restricted model in \cref{eq:HJmuh} has infinitely many conserved charges:
\ie
Q^{\text{(exp)}}(\alpha) = \sum_{m=1}^{M}  \left[ \alpha^m c_{m}^{\dagger} c_{m} + \alpha^{m-1}(1-\alpha)\left( \frac{\sigma_{m}^{z} + 1}{2} \right) \right]\ ,\fe
which holds for arbitrary number $\alpha$, since each fermion on site $m$ carrying charge $\alpha^m$ can split into a fermion with charge $\alpha^{m+1}$ and a spin up with $\alpha^m(1-\alpha)$ on site $m+1$. This then implies infinite local symmetries:
\ie
Q^{\text{(exp)}}(\alpha) = \sum_{m=0}^{M} \alpha^m P_m
,\fe
where each local operator $P_m$ is a conserved charge. Each of these charge sectors can further fragment into dynamically disconnected Krylov subspaces. This conserved charge can be understood as stemming from Eq.~\eqref{eqn: Q in the original model} in the original breakdown model. The presence of the conserved charge $Q^{\text{(exp)}}$ further subdivides a given symmetry sector corresponding to a particular value of $Q$ into smaller sectors, each corresponding to different values of $Q^{\text{(exp)}}$.

Note that if the magnetic field $\vec{h}_m$ has non-zero $x$ and $y$ components, the spin states will precess along the tilted magnetic field and, as a result of this Larmor precession, the initial alignment of spin states $s_{m_b}$ with the $z$-axis will change. Consequently, the blockades that initially prevent the fermions from passing through the site $m_b$ will no longer be effective and there will be mixing between the Krylov subspaces.


\section{Zero magnetic field: $Q = 1$ sector}
\label{section: Without magnetic field: Q=1}

We now investigate the HSF and quantum dynamics in the presence of both the conventional fermionic hopping term $H_\gamma$ and the spatially asymmetric interaction term $H_J$. In this Section, we focus on the charge sector with a single fermion, i.e., $Q = 1$ (as defined in Eq.~\eqref{eqn: conserved charge}) and consider the case with zero magnetic field $\vec{h}_m=0$. We will consider both uniform and random couplings here and also study the effect of a non-vanishing on-site potential $\mu_m$ on the dynamics, which we study using analytical arguments in the following.

\subsection{$H_{\gamma J}$: odd $M$, $\mathcal{K}^{(1)}$ subspaces}
\label{subsection: Q=1, odd M, K=1}

Consider the model with Hamiltonian $H_{\gamma J}$ and odd chain lengths $M = 2 Z + 1$  (where $Z \in \mathbb{Z}_{>0}$), and let the coefficients $\gamma_m$ and $J_m$ be either random or uniform. Here, within the $Q = 1$ charge sector, we can analytically identify all the zero-modes of $H_{\gamma J}$ with zero energy as follows, each of which is by definition a one-dimensional Krylov subspaces: 
\begin{enumerate}
\item We express a zero mode satisfying $H_{\gamma J}| \Psi \rangle=0$ as
\ie
\label{eq:root}
| \Psi \rangle = \sum_{i=0}^{Z} \eta_i | \Phi_i \rangle
,\fe
where $\eta_i= \pm 1$ are coefficients that will be fixed later. Each of the $Z+1$ states $\ket{\Phi_i}$ is a product state over the fermion number occupation and spin $\sigma_m^z$ basis. 

\item For each state $| \Phi_i \rangle$ ($i\in[0, Z]$), the only fermion occupies the $(2i+1)$-th site.

\item Arbitrarily choose the values of $s_2, \cdots, s_M$ to be either $0$ or $1$. For the $i$-th component $| \Phi_i \rangle$, where the fermion occupies the $f_1=(2i+1)$-th site, fix the spin states to be (the notations are defined in \cref{subsubsection: Notations for K1} and App.~\ref{appendix: Explicit forms of the notations for K1})
\ie
\begin{cases}
    |\overleftarrow{s_m} \rangle \text{, for $m\in 2 \mathbb{Z}_{>0}$ if $f_1 < m \le M-1$},
    \\
    |\overrightarrow{s_m} \rangle \text{, for $m\in (2 \mathbb{Z}_{>0} +1)$ if $3 \le m \le f_1$},
    \\
    |s_m\rangle \text{, otherwise}\ .
\end{cases}
\fe

\item For the $i$-th component $| \Phi_i \rangle$, fix the relative coefficients $\eta_i$ to be
\ie
\label{eqn: eta coefficients}
\begin{cases}
\eta_i = 1 &\text{, if } i \in 2 \mathbb{Z}_{\geq 0},
\\
\eta_i = -1 &\text{, if } i \in 2 \mathbb{Z}_{\geq 0} + 1.
\end{cases}
\fe
\end{enumerate}

Note that the spin states $s_2, \cdots, s_M$ can be arbitrarily assigned here, and since each $s_m$ can take two values: $1$ or $0$, there are in total 
\begin{equation}
D({\mathcal{K}^{(1)}})=2^{M-1} = 2^{2Z}
\end{equation}
such zero modes. Since the number of zero modes increases with the system size exponentially, the systems with odd $M$ exhibit HSF.

To understand why the state $|\Psi\rangle$ constructed by the above procedure is a zero mode, note that: first, the only terms of $H_{\gamma J}$ that do not act trivially on (namely, do not annihilate) $| \Phi_i \rangle$ are $\hat{h}_{2i+1}^{(L)}$ and $\hat{h}_{2i+1}^{(R)}$; second, the components of $| \Psi \rangle$ obey:
\ie
\hat{h}_{2i+1}^{(R)} | \Phi_i \rangle - \hat{h}_{2i+3}^{(L)} | \Phi_{i+1} \rangle = 0
,\fe
which ensures that $H_{\gamma J} | \Psi \rangle = 0$. Let us take $M=5$ as an example. The zero modes can be graphically represented by (the notations are defined in Appendix.~\ref{appendix: Explicit forms of the notations for K1})
\begin{alignat}{1}
| \Psi \rangle = \; & | \Phi_0 \rangle - | \Phi_1 \rangle + | \Phi_2 \rangle
\notag \\
= \;
&\begin{tikzpicture}[baseline={(0, -0.1)}]
\draw (0.0, 0) circle (.2);
\draw (0.8, 0) circle (.2);
\draw (1.6, 0) circle (.2);
\draw (2.4, 0) circle (.2);
\draw (3.2, 0) circle (.2);
\fill (0.0, 0) circle (2pt);
\draw (0.2, 0) -- (0.6, 0);
\draw (1.0, 0) -- (1.4, 0);
\draw (1.8, 0) -- (2.2, 0);
\draw (2.6, 0) -- (3.0, 0);
\draw[<-] (0.25, 0.15) -- (0.55, 0.15);
\draw[<-] (1.85, 0.15) -- (2.15, 0.15);
\end{tikzpicture}
\notag\\
-\;
&\begin{tikzpicture}[baseline={(0, -0.1)}]
\draw (0.0, 0) circle (.2);
\draw (0.8, 0) circle (.2);
\draw (1.6, 0) circle (.2);
\draw (2.4, 0) circle (.2);
\draw (3.2, 0) circle (.2);
\fill (1.6, 0) circle (2pt);
\draw (0.2, 0) -- (0.6, 0);
\draw (1.0, 0) -- (1.4, 0);
\draw (1.8, 0) -- (2.2, 0);
\draw (2.6, 0) -- (3.0, 0);
\draw[->] (1.05, 0.15) -- (1.35, 0.15);
\draw[<-] (1.85, 0.15) -- (2.15, 0.15);
\end{tikzpicture}
\notag\\
+\;
&\begin{tikzpicture}[baseline={(0, -0.1)}]
\draw (0.0, 0) circle (.2);
\draw (0.8, 0) circle (.2);
\draw (1.6, 0) circle (.2);
\draw (2.4, 0) circle (.2);
\draw (3.2, 0) circle (.2);
\fill (3.2, 0) circle (2pt);
\draw (0.2, 0) -- (0.6, 0);
\draw (1.0, 0) -- (1.4, 0);
\draw (1.8, 0) -- (2.2, 0);
\draw (2.6, 0) -- (3.0, 0);
\draw[->] (1.05, 0.15) -- (1.35, 0.15);
\draw[->] (2.65, 0.15) -- (2.95, 0.15);
\end{tikzpicture}
.\end{alignat}
If we apply $H_{\gamma J}$ on $|\Psi\rangle$, all fermions will move to their neighboring sites. In addition, the spin states of $|\Psi\rangle$ on the edges are fine-tuned such that, after applying the Hamiltonian, 
\ie
\hat{h}_{1}^{(R)} | \Phi_0 \rangle - \hat{h}_{3}^{(L)} | \Phi_1 \rangle &= 0
\\
-\hat{h}_{3}^{(R)} | \Phi_1 \rangle + \hat{h}_{5}^{(L)} | \Phi_2 \rangle &= 0
.\fe
Therefore, $H_{\gamma J} |\Psi\rangle = 0$. Since each $s_m$ ($m\in\left[2, 5\right]$) can take one of two values, we have $16$ zero modes in total.

For even chain lengths $M=2Z$ ($Z\in\mathbb{Z}_{>0}$), however, there is no way to find a state $| \Psi \rangle$ such that $H_{\gamma J} | \Psi \rangle = 0$, and there are no exact zero modes for models with even system size $M$.

\subsection{$H_{\bar{\gamma} \bar{J}}$: odd $M$, $\mathcal{K}^{(2)}$ subspaces and higher}
\label{subsection: Q=1, odd M, K>=2}

For the model Hamiltonian $H_{\bar{\gamma} \bar{J}}$ with uniform hopping coefficients $\gamma_m=\gamma$ and interaction strengths $J_m=J$, and no other terms, there exist even richer Krylov subspace structures beyond the one-dimensional Krylov subspaces $\mathcal{K}^{(1)}$ discussed above. We now show that the model in this limit exhibits higher-dimensional Krylov subspaces, including two-dimensional Krylov subspaces $\mathcal{K}^{(2)}$. In this section, we only consider models with chain lengths $M=4Z+3$ ($Z \in \mathbb{Z}_{>0}$).

We begin by analytically constructing the root states for the two-dimensional Krylov subspaces $\mathcal{K}^{(2)}$ for system size $M=4Z+3$ in the $Q = 1$ sector:
\begin{enumerate}
\item The root state $\ket{\Psi}$ is composed of $(Z+1)$ components $\ket{\Phi_i}$ as in Eq.~\eqref{eq:root}. For the $i$-th component $| \Phi_i \rangle$ ($i\in [0, Z]$), a fermion occupies the $f_1=(4i+2)$-th site.
\item For the $i$-th component $| \Phi_i \rangle$, where the fermion occupies the $f_1=(4i+2)$-th site, we construct $| \Phi_i \rangle$ from the rectangular building blocks introduced in Sec.~\ref{subsubsection: Notations for K2} and Eq.~\eqref{eqn: 3 kinds of rectangularized spin states} (also see Appendix~\ref{appendix: Explicit forms of the notations for K2} for explicit forms) as follows:
\ie
\label{eqn: singly rectangularize sites}
\begin{cases}
    \text{Right-rectangularize the group of sites:}
    \\
    \;\;\;\{ (m-1, m, m+1) | m \in 4\mathbb{Z}_{\geq 0}+2 \text{ and } m<f_1 \},
    \\
    \text{Neutral-rectangularize the group of sites:}
    \\
    \;\;\;(f_1-1, f_1, f_1+1),
    \\
    \text{Left-rectangularize the group of sites:}
    \\
    \;\;\;\{ (m-1, m, m+1) | m \in 4\mathbb{Z}_{\geq 0}+2 \text{ and } m>f_1 \}.
\end{cases}
\fe
There are then a total of $(Z+1)$ rectangles with unfixed rectangular labels $r_1, \cdots r_{Z+1}$, each of which can take one of two values.
\item Arbitrarily choose the spin states 
\ie
\label{eqn: zero modes of K2, spin states}
\{(s_m, s_{m+1}) | m \in 4\mathbb{Z}_{>0} \}
\fe
to be either $0$ or $1$ ($2Z$ in total). For the $i$-th component $| \Phi_i \rangle$, where the fermion occupies the $f_1=(4i+2)$-th site, fix the spin states to be
\ie
\label{eqn: spin states for rectangularized states}
\begin{cases}
    |\overleftarrow{s_m} \rangle \text{, for $m\in 4 \mathbb{Z}_{>0}$ if $m > f_1$}
    \\
    |\overrightarrow{s_m} \rangle \text{, for $m\in (4 \mathbb{Z}_{>0} + 1)$ if $m < f_1$}
    \\
    |s_m \rangle \text{, otherwise}
\end{cases}
.\fe

\item For the $i$-th component $| \Phi_i \rangle$, fix the relative coefficients $\eta_i$ as Eq.~\eqref{eqn: eta coefficients}.
\end{enumerate}

Since we can arbitrarily choose $2Z$ spin states in Eq.~\eqref{eqn: zero modes of K2, spin states} and the labels $r_j$ of $(Z+1)$ rectangles, the number of the root states constructed from the above procedure is $2^{2Z} 2^{Z+1} = 2^{3Z+1}$, showing that there are \textit{at least} 
\begin{equation}
D({\mathcal{K}^{(2)}})\ge 2^{3Z+1}
\end{equation}
two-dimensional Krylov subspaces. Thus, the number of two-dimensional Krylov subspaces $D({\mathcal{K}^{(2)}})$ increases at least exponentially with the system size, further confirming the presence of fragmentation in our model. 

To see the states constructed by the above procedure are root states for two-dimensional Krylov subspaces $\mathcal{K}^{(2)}$, note that the components of these states satisfy:
\begin{enumerate}[label=(\alph*)]
\item From Eq.~\eqref{eqn: H^2 acting on rectangularized states}, each component of the root state is proportional to itself after two successive actions of the Hamiltonian if the fermion returns to the original site:
\ie
\notag
\left( \hat{h}_{4i+3}^{(L)} \hat{h}_{4i+2}^{(R)} + \hat{h}_{4i+1}^{(R)} \hat{h}_{4i+2}^{(L)} \right) | \Phi_i\rangle = (2\gamma^2 + J^2) | \Phi_i\rangle
.\fe
\item From Eq.~\eqref{eqn: H acting on arrow states}, two consecutive components of the root state will cancel each other after undergoing two successive operations by the Hamiltonian if the fermion doesn't return:
\ie
\notag
\hat{h}_{4i+3}^{(R)} \hat{h}_{4i+2}^{(R)} | \Phi_i \rangle - \hat{h}_{4(i+1)+1}^{(L)} \hat{h}_{4(i+1)+2}^{(L)} | \Phi_{i+1} \rangle = 0
.\fe
\end{enumerate}
These two properties ensure that after applying the Hamiltonian on the root states twice, we have
\ie
H_{\bar{\gamma} \bar{J}}^2 | \Psi \rangle = (2\gamma^2 + J^2) | \Psi \rangle
.\fe
As an example, consider $M=7$. The root states are:
\begin{alignat}{2}
| \Psi \rangle = \;
&\begin{tikzpicture}[baseline={(0, -0.1)}]
\draw (-0.4, 0.25) rectangle (2.0, -0.25);
\draw (2.8, 0.25) rectangle (5.2, -0.25);
\draw (0.0, 0) circle (.2);
\draw (0.8, 0) circle (.2);
\draw (1.6, 0) circle (.2);
\draw (2.4, 0) circle (.2);
\draw (3.2, 0) circle (.2);
\draw (4.0, 0) circle (.2);
\draw (4.8, 0) circle (.2);
\fill (0.8, 0) circle (2pt);
\draw (0.2, 0) -- (0.6, 0);
\draw (1.0, 0) -- (1.4, 0);
\draw (1.8, 0) -- (2.2, 0);
\draw (2.6, 0) -- (3.0, 0);
\draw (3.4, 0) -- (3.8, 0);
\draw (4.2, 0) -- (4.6, 0);
\draw[<-] (1.85, 0.15) -- (2.15, 0.15);
\draw[<-] (3.45, 0.4) -- (3.75, 0.4);
\node at (0.8, 0.4) {\small $r_1$};
\node at (4.0, 0.4) {\small $r_2$};
\node at (2.0, -0.38) {\small $s_4$};
\node at (2.8, -0.38) {\small $s_5$};
\end{tikzpicture}
\notag \\ 
- \;
&\begin{tikzpicture}[baseline={(0, -0.1)}]
\draw (-0.4, 0.25) rectangle (2.0, -0.25);
\draw (2.8, 0.25) rectangle (5.2, -0.25);
\draw (0.0, 0) circle (.2);
\draw (0.8, 0) circle (.2);
\draw (1.6, 0) circle (.2);
\draw (2.4, 0) circle (.2);
\draw (3.2, 0) circle (.2);
\draw (4.0, 0) circle (.2);
\draw (4.8, 0) circle (.2);
\fill (4.0, 0) circle (2pt);
\draw (0.2, 0) -- (0.6, 0);
\draw (1.0, 0) -- (1.4, 0);
\draw (1.8, 0) -- (2.2, 0);
\draw (2.6, 0) -- (3.0, 0);
\draw (3.4, 0) -- (3.8, 0);
\draw (4.2, 0) -- (4.6, 0);
\draw[->] (2.65, 0.15) -- (2.95, 0.15);
\draw[->] (1.05, 0.4) -- (1.35, 0.4);
\node at (0.8, 0.4) {\small $r_1$};
\node at (4.0, 0.4) {\small $r_2$};
\node at (2.0, -0.38) {\small $s_4$};
\node at (2.8, -0.38) {\small $s_5$};
\end{tikzpicture}
.\end{alignat}
One can easily show that $| \Psi \rangle$ is a root state of $\mathcal{K}^{(2)}$ by verifying $H_{\bar{\gamma} \bar{J}}^2 | \Psi \rangle = (2\gamma^2 + J^2) | \Psi \rangle$. Moreover, this indicates all these two-dimensional Krylov subspaces $\mathcal{K}^{(2)}$ are degenerate, having a degenerate energy spectrum 
\begin{equation}
E_\pm=\sqrt{2\gamma^2 + J^2}\ .
\end{equation}

Besides these two-dimensional Krylov subspaces $\mathcal{K}^{(2)}$, there exist even higher dimensional Krylov subspaces for the $Q=1$ sector with uniform $\gamma_m=\gamma$ and $J_m=J$, leading to further fragmentation of the Hilbert space, such as $\mathcal{K}^{(4)}$ discussed in Appendix.~\ref{appendix: Q=1, odd M, K=4}. 

One way to numerically identify HSF is to check whether the characteristic polynomial of an integer-valued Hamiltonian can be factorized into integer polynomials~\cite{Regnault:2022ocy}. Each factor of the factorized integer polynomials implies a Krylov subspace. In Appendix~\ref{appendix: characteristic polynomial factorization}, we explicitly show the factors in the $Q=1$ sector for system sizes up to $M=11$. The factor degrees can be larger than $4$, showing the existence of Krylov subspaces with dimensions larger than $4$. However, we have not currently found an analytical method to count the number of Krylov subspaces with dimensions larger than $4$ or an analytic approach for calculating the dimension of the largest Krylov subspace, which we leave for future work. 

In Appendix~\ref{appendix: Q=1 numerical factorization}, we provide numerical evidence that the dimension of the largest Krylov subspace is $|\mathcal{K}^{\text{max}}| = 4\times \lfloor M/2 \rfloor$ by considering models with system sizes up to $M\le 11$. Although the factorization for models with $M>12$ is numerically inaccessible, we expect that this pattern will continue to persist, implying that
\ie
\label{eqn: limit of the ratio of Kmax and HQ}
\lim_{M\rightarrow \infty} \frac{|\mathcal{K}^{\text{max}}|}{|\mathcal{H}_{Q=1}|} 
=\frac{4\times \lfloor M/2 \rfloor}{M\times 2^{M-1}} \rightarrow 0
.\fe
Therefore, we expect the Hilbert space of the $Q=1$ sector exhibits strong fragmentation. We also emphasize that the root states for these Krylov subspaces are not simply direct products of the fermionic chain and the spin states, which can be clearly seen from our explicit construction for the two-dimensional Krylov subspaces. In the lexicon of Ref.~\cite{moudgalya2021hilbert}, our model thus exhibits \textit{quantum fragmentation} rather than classical fragmentation.

\subsection{$H_{\bar{\gamma} \bar{J}}$: even $M$}
\label{subsection: Q=1, even M}

\begin{figure}[t]
\centering
\includegraphics[width=80mm]{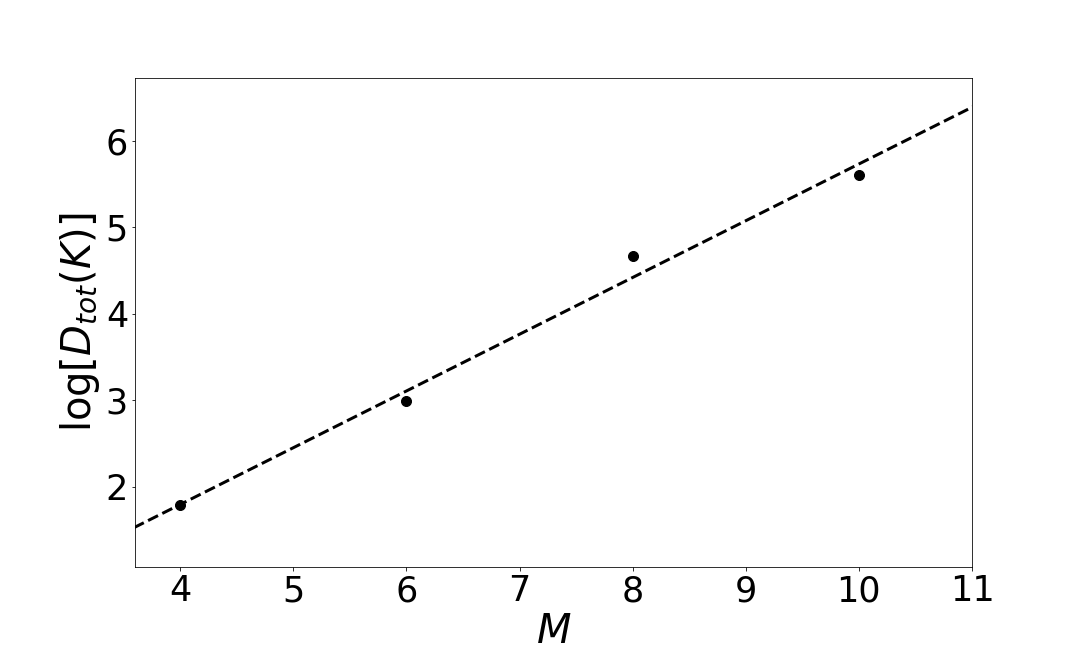}
\caption{Plot of (the logarithm of) the total number of Krylov subspaces $\log D_{\text{tot}}(\mathcal{K})$ with respect to the system size $M$ (focusing only on the even values).}
\label{fig: D(K) Q=1}
\end{figure}

For even system sizes $M$ in the $Q=1$ sector, we are not able to analytically identify the Krylov subspaces, and we resort to direct factorization of the integer characteristic polynomial. We focus on the models where the coefficients $\gamma_m=\gamma$ and $J_m=J$ are uniform. In \figref{fig: D(K) Q=1}, we show the growth of the total number of Krylov subspaces $D_{\text{tot}}(\mathcal{K})$ as $M$ increases. The dimensions of the factors for models with $M=4, 6, 8, 10$ are summarized in Appendix~\ref{appendix: Q=1 numerical factorization}. As shown in \figref{fig: D(K) Q=1}, the number of Krylov subspaces $D_{\text{tot}}(\mathcal{K})$ increases rapidly with $M$. This observation suggests that, in the thermodynamic limit ($M \rightarrow \infty$), the systems will display non-ergodic behavior as a consequence of Hilbert space fragmentation. Additionally, we find that the ratio $\frac{|\mathcal{K}^{\text{max}}|}{|\mathcal{H}_Q|}$ decreases monotonically as $M \rightarrow \infty$, see Eq.~\eqref{eqn: limit of the ratio of Kmax and HQ}, from which we infer that the fragmentation is also strong.

\subsection{$H_{\bar{\gamma} \bar{J} \mu}$: partially mixed Krylov subspaces}

We now consider model Hamiltonian $H_{\bar{\gamma} \bar{J} \mu}$ with uniform $\gamma_m=\gamma$, $J_m=J$, but include the random on-site potential $\mu_m$. The fragmentation structure we have discussed previously no longer holds true, and different invariant subspaces couple with each other due to this random on-site potential. Nonetheless, contrary to the expectation that generically all invariant subspaces will merge and couple, we find only partial merging of these subspaces and analytically show that fragmentation survives the addition of this random on-site potential $\mu_m$ term $H_\mu$.

Due to the presence of $H_\mu$, Eq.~\eqref{eqn: neutral-rectangular acted by the Hamiltonian once} must be revised as follows:
\begin{alignat}{1}
H_{\bar{\gamma} \bar{J}} \;
\begin{tikzpicture}[baseline={(0, -0.1)}]
\draw (0.0, 0.25) rectangle (2.4, -0.25);
\draw (0.4, 0) circle (.2);
\draw (1.2, 0) circle (.2);
\draw (2.0, 0) circle (.2);
\fill (1.2, 0) circle (2pt);
\draw (0.6, 0) -- (1.0, 0);
\draw (1.4, 0) -- (1.8, 0);
\node at (0, 0) {$\cdots$};
\node at (2.5, 0) {$\cdots$};
\node at (1.2, 0.4) {\small $r_j$};
\node at (1.2, -0.4) {\small $m$};
\end{tikzpicture}
&=
\begin{tikzpicture}[baseline={(0, -0.1)}]
\draw (0.0, 0.25) rectangle (2.4, -0.25);
\draw (0.4, 0) circle (.2);
\draw (1.2, 0) circle (.2);
\draw (2.0, 0) circle (.2);
\fill (0.4, 0) circle (2pt);
\draw (0.6, 0) -- (1.0, 0);
\draw (1.4, 0) -- (1.8, 0);
\draw[<-] (0.65, 0.4) -- (0.95, 0.4);
\node at (0, 0) {$\cdots$};
\node at (2.5, 0) {$\cdots$};
\node at (1.2, 0.4) {\small $r_j$};
\end{tikzpicture}
\notag \\
&+
\begin{tikzpicture}[baseline={(0, -0.1)}]
\draw (0.0, 0.25) rectangle (2.4, -0.25);
\draw (0.4, 0) circle (.2);
\draw (1.2, 0) circle (.2);
\draw (2.0, 0) circle (.2);
\fill (2.0, 0) circle (2pt);
\draw (0.6, 0) -- (1.0, 0);
\draw (1.4, 0) -- (1.8, 0);
\draw[->] (1.45, 0.4) -- (1.75, 0.4);
\node at (0, 0) {$\cdots$};
\node at (2.5, 0) {$\cdots$};
\node at (1.2, 0.4) {\small $r_j$};
\end{tikzpicture}
\notag \\
&+
\mu_m
\begin{tikzpicture}[baseline={(0, -0.1)}]
\draw (0.0, 0.25) rectangle (2.4, -0.25);
\draw (0.4, 0) circle (.2);
\draw (1.2, 0) circle (.2);
\draw (2.0, 0) circle (.2);
\fill (1.2, 0) circle (2pt);
\draw (0.6, 0) -- (1.0, 0);
\draw (1.4, 0) -- (1.8, 0);
\node at (0, 0) {$\cdots$};
\node at (2.5, 0) {$\cdots$};
\node at (1.2, 0.4) {\small $r_j$};
\end{tikzpicture}
.\end{alignat}
If we apply the Hamiltonian multiple times on this state, the fermion might move beyond the rectangularized sites $\left( m-1, m, m+1 \right)$. However, due to Eq.~\eqref{eqn: LR-rectangular acted by the Hamiltonian once}, whenever the fermion returns to the rectangularized sites, the rectangularized spin states will always be in one of the 6 following states:
\begin{alignat}{1}
\label{eqn: 6 choices of the rectangularized spin states}
\begin{cases}
&\begin{tikzpicture}[baseline={(0, -0.1)}]
\draw (0.0, 0.25) rectangle (2.4, -0.25);
\draw (0.4, 0) circle (.2);
\draw (1.2, 0) circle (.2);
\draw (2.0, 0) circle (.2);
\fill (0.4, 0) circle (2pt);
\draw (0.6, 0) -- (1.0, 0);
\draw (1.4, 0) -- (1.8, 0);
\draw[<-] (0.65, 0.4) -- (0.95, 0.4);
\node at (0, 0) {$\cdots$};
\node at (2.5, 0) {$\cdots$};
\node at (1.2, 0.4) {\small $r_j$};
\end{tikzpicture}
\\
&\begin{tikzpicture}[baseline={(0, -0.1)}]
\draw (0.0, 0.25) rectangle (2.4, -0.25);
\draw (0.4, 0) circle (.2);
\draw (1.2, 0) circle (.2);
\draw (2.0, 0) circle (.2);
\fill (1.2, 0) circle (2pt);
\draw (0.6, 0) -- (1.0, 0);
\draw (1.4, 0) -- (1.8, 0);
\node at (0, 0) {$\cdots$};
\node at (2.5, 0) {$\cdots$};
\node at (1.2, 0.4) {\small $r_j$};
\end{tikzpicture}
\\
&\begin{tikzpicture}[baseline={(0, -0.1)}]
\draw (0.0, 0.25) rectangle (2.4, -0.25);
\draw (0.4, 0) circle (.2);
\draw (1.2, 0) circle (.2);
\draw (2.0, 0) circle (.2);
\fill (2.0, 0) circle (2pt);
\draw (0.6, 0) -- (1.0, 0);
\draw (1.4, 0) -- (1.8, 0);
\draw[->] (1.45, 0.4) -- (1.75, 0.4);
\node at (0, 0) {$\cdots$};
\node at (2.5, 0) {$\cdots$};
\node at (1.2, 0.4) {\small $r_j$};
\end{tikzpicture}
\end{cases}
,\end{alignat}
where $r_j\in[0, 1]$.  For a fine-tuned root state $| \Psi \rangle$, this fact will make all the states $\{H^{n} | \Psi \rangle| n\in\mathbb{Z}_{\geq 0}\}$ span a non-trivial subspace.

The root state $\ket{\Psi}$ of the invariant subspace for $M=4Z+3$ ($Z \in \mathbb{Z}_{\geq 0}$) can be constructed as follows:
\begin{enumerate}
\item The root state $| \Psi \rangle$ is composed of $(Z+1)$ components as in Eq.~\eqref{eq:root}. For the $i$-th component $|\Phi_i\rangle$ ($i \in [0, Z]$), a fermion occupies the $(4i+2)$-th site.
\item For the $i$-th component $| \Phi_i \rangle$, where the fermion occupies the $f_1=(4i+2)$-th site, rectangularize the following sites as in Eq.~\eqref{eqn: singly rectangularize sites}:
\ie
\label{eqn: rectangularized sites of H gamma J mu model}
\left\{ (m-1, m, m+1) | m \in 4\mathbb{Z}_{\geq 0} + 2 \right\}
.\fe
There are a total of $(Z+1)$ rectangles with unfixed rectangular labels $r_1, \cdots r_{Z+1}$, each of which can take one of 2 values.
\item Arbitrarily choose the spin states $\left\{ s_m, s_{m+1} | m \in 4 \mathbb{Z}_{>0} \right\}$ ($2Z$ in total) to be either $0$ or $1$.
\end{enumerate}

By applying the Hamiltonian on this root state, the fermion will generically visit all sites, from $1$ to $M$. However, whenever the fermion is on any of the rectangularized sites i.e., Eq.~\eqref{eqn: rectangularized sites of H gamma J mu model}, the rectangularized spin states will be in one of the states in Eq.~\eqref{eqn: 6 choices of the rectangularized spin states}. This prevents $(H_{\bar{\gamma}\bar{J}\mu})^n | \Psi \rangle$ from generating the entire Hilbert space $\mathcal{H}_Q$. Instead, the dimension of the Krylov subspace is upper bounded by 
\begin{equation}
|\mathcal{K}|\le M \cdot 2^{2Z} \cdot 2^{Z+1}\ ,
\end{equation}
since there are $M$ possible sites in which the fermion can land, $Z+1$ rectangular labels, and $2Z$ undetermined spin states $\left\{ s_m, s_{m+1} | m \in 4 \mathbb{Z}_{>0} \right\}$. The bound is exponentially smaller than the dimension of the Hilbert space $|\mathcal{H}_Q|=M\cdot 2^{4Z+2}$ within a fixed $Q$ sector--therefore, the invariant subspace forms a non-trivial Krylov subspace.
\\~\\
To summarize, in this section we have explored the behavior of the breakdown model in the $Q=1$ sector and shown the following: 

(i) For model $H_{\gamma J}$ with odd system size $M$, there exist extensively many zero modes, which prevent the system from thermalizing. 

(ii) For model $H_{\bar{\gamma} \bar{J}}$ with uniform $\gamma_m=\gamma$ and $J_m=J$ and with odd system size $M=4Z+3$ ($Z\in \mathbb{Z}_{\geq 0}$, i.e., all non-negative integers), we identified extensive Krylov subspaces with dimension $2$. The Krylov subspaces with dimension $4$ is discussed in Appendix.~\ref{appendix: Q=1, odd M, K=4}. 

(iii) For model $H_{\bar{\gamma} \bar{J}}$ even $M$, we numerically established the existence of extensive number of Krylov subspaces by direct factorization. 

(iv) Finally, we showed that in the presence of a random on-site potential, namely, for model $H_{\bar{\gamma} \bar{J} \mu}$, fragmentation survives in this sector, although many of the Krylov subspaces analyzed above merge together.


\section{Zero Magnetic Field: $Q\ge 2$ Sectors}
\label{section: Without magnetic field: Q>=2}

Thus far, we have focused on the $Q=1$ sector of the breakdown model and argued that it hosts dynamically disconnected Krylov spaces, which would prevent this sector from thermalizing. A natural question is whether Hilbert space fragmentation persists as $Q$ is increased, which we consider in this section. Specifically, we analytically demonstrate that while models with odd $M$ and even $Q$ are ergodic, for even system sizes $M$ all $Q$ sectors are non-ergodic. This non-ergodicity stems from the presence of extensively many zero-modes, namely, one-dimensional Krylov subspaces (similar to the $Q=1$ case we previously studied).

\subsection{$H_{\gamma J}$: $Q=2$, $\mathcal{K}^{(1)}$ subspaces}
For charge sector $Q=2$ with either uniform or random $\gamma_m$ and $J_m$, we find there exist zero modes for all system-sizes $M$. In addition, the number of zero modes is lower bounded by $D(\mathcal{K}^{(1)}) \ge 2^{M-1}$. As a result, the $Q=2$ sector also exhibits HSF. We can construct (part of) the zero modes in the $Q=2$ sector as follows:
\begin{enumerate}
\item The zero modes $\ket{\Psi}$ are expressed in terms of $(M-1)$ components, $\ket{\Psi} = \sum_{i=1}^{M-1} \eta_i |\Phi_i\rangle $. For the $i$-th component $| \Phi_i \rangle$ ($i\in[1,M-1]$), two fermions occupy consecutive sites $(f_1, f_2)=(i, i+1)$.
\item Arbitrarily choose the values of $s_2, \cdots, s_M$. For the $i$-th component $| \Phi_i \rangle$, where two fermions occupy the $(f_1, f_2)=(i, i+1)$ sites, fix the spin states as
\ie
\begin{cases}
    | \overrightarrow{s_m} \rangle \text{, if $m \le f_1$,}
    \\
    | s_m \rangle \text{, if $m = f_2$,}
    \\
    | \overleftarrow{s_m} \rangle \text{, if $m>f_2$.}
\end{cases}
\fe
\item Fix the relative coefficients to be
\ie
\begin{cases}
\eta_i = 1, &\text{ if } i \in 2\mathbb{Z}_{>0}-1
\\
\eta_i = -1, &\text{ if } i \in 2\mathbb{Z}_{>0}
\end{cases}
\fe
\end{enumerate}
 
Take $M=4$ as an example. The zero modes can be graphically represented by
\begin{alignat}{1}
| \Psi \rangle = \;
&\begin{tikzpicture}[baseline={(0, -0.1)}]
\draw (0.0, 0) circle (.2);
\draw (0.8, 0) circle (.2);
\draw (1.6, 0) circle (.2);
\draw (2.4, 0) circle (.2);
\fill (0.0, 0) circle (2pt);
\fill (0.8, 0) circle (2pt);
\draw (0.2, 0) -- (0.6, 0);
\draw (1.0, 0) -- (1.4, 0);
\draw (1.8, 0) -- (2.2, 0);
\draw[<-] (1.05, 0.15) -- (1.35, 0.15);
\draw[<-] (1.85, 0.15) -- (2.15, 0.15);
\end{tikzpicture}
\notag\\
-\;
&\begin{tikzpicture}[baseline={(0, -0.1)}]
\draw (0.0, 0) circle (.2);
\draw (0.8, 0) circle (.2);
\draw (1.6, 0) circle (.2);
\draw (2.4, 0) circle (.2);
\fill (0.8, 0) circle (2pt);
\fill (1.6, 0) circle (2pt);
\draw (0.2, 0) -- (0.6, 0);
\draw (1.0, 0) -- (1.4, 0);
\draw (1.8, 0) -- (2.2, 0);
\draw[->] (0.25, 0.15) -- (0.55, 0.15);
\draw[<-] (1.85, 0.15) -- (2.15, 0.15);
\end{tikzpicture}
\notag\\
+\;
&\begin{tikzpicture}[baseline={(0, -0.1)}]
\draw (0.0, 0) circle (.2);
\draw (0.8, 0) circle (.2);
\draw (1.6, 0) circle (.2);
\draw (2.4, 0) circle (.2);
\fill (1.6, 0) circle (2pt);
\fill (2.4, 0) circle (2pt);
\draw (0.2, 0) -- (0.6, 0);
\draw (1.0, 0) -- (1.4, 0);
\draw (1.8, 0) -- (2.2, 0);
\draw[->] (0.25, 0.15) -- (0.55, 0.15);
\draw[->] (1.05, 0.15) -- (1.35, 0.15);
\end{tikzpicture}
.\end{alignat}
$H_{\gamma J} | \Psi \rangle = 0$ can then be easily verified. 

In the general construction delineated above, each $s_m$ can take 2 values and therefore the procedure produces $2^{M-1}$ zero modes. Note that this only produces a subset of all zero modes within the $Q=2$ sector, namely, 
\begin{equation}
D(\mathcal{K}^{(1)})\ge 2^{M-1}\ .
\end{equation}
However, as this provides a lower bound that increases exponentially with system-size $M$, it provides sufficient proof for fragmentation. 

\subsection{$H_{\gamma J}$: $Q \ge 3$, $\mathcal{K}^{(1)}$ subspaces}

The previous construction for identifying zero modes explicitly does not extend beyond the $Q=2$ sector. Nonetheless, we can give a more abstract construction of zero modes of $H_{\gamma J}$ that works for any $Q$ sector and system-size $M$, which we now discuss and use to establish Hilbert space fragmentation analytically for any $M$ when $Q$ is even and for odd $M$ when $Q$ is odd.

We start with the following lemma:
\begin{lemma}
\label{lemma: move to right cancelled by move to left}
Consider the $Q$-sector of a model with system-size $M$, where $(M-2) \ge Q \ge 1$. Given any state $\ket{\xi}$ in this sector with fermions at sites $f_1, \cdots, f_Q$ such that there exists at least one site $f_q$ ($1 \le q \le Q$ and $f_q \le M-2$) which is occupied while the neighboring two sites $(f_{q}+1)$ and $(f_{q}+2)$ are unoccupied:
\begin{alignat}{2}
|\xi\rangle &= \underbrace{\cdots c_{f_q}^\dagger \cdots  | \Omega \rangle}_{\text{fermion part}} \otimes \underbrace{ \cdots \otimes | s_{f_q+1} \rangle \otimes | s_{f_q+2} \rangle \otimes \cdots}_{\text{spin part}}
\\
\notag
&= 
\begin{tikzpicture}[baseline={(0, -0.1)}]
    \draw (1.4, 0) circle (.2);
    \draw (2.8, 0) circle (.2);
    \draw (4.2, 0) circle (.2);
    \fill (1.4, 0) circle (2pt);
    \draw (0.2, 0) -- (1.2, 0);
    \draw (1.6, 0) -- (2.6, 0);
    \draw (3.0, 0) -- (4.0, 0);
    \node at (-0.1, 0) {$\cdots$};
    \node at (4.7, 0) {$\cdots$};
    \node at (1.4, -0.5) {\small $f_q$};
    \node at (4.2, -0.5) {\small $f_q+2$};
    \node at (0.7, -0.25) {\small $s_{f_q}$};
    \node at (2.1, -0.25) {\small $s_{f_q+1}$};
    \node at (3.5, -0.25) {\small $s_{f_q+2}$};
\end{tikzpicture}
,\end{alignat}
there exists a state 
\begin{alignat}{2}
|\xi^{\prime} \rangle &= \cdots c_{f_{q}+2}^\dagger \cdots  | \Omega \rangle \otimes  \cdots \otimes | s_{f_q+1}^\prime \rangle \otimes | s_{f_q+2}^\prime \rangle \otimes \cdots
\\
\notag
&=
\begin{tikzpicture}[baseline={(0, -0.1)}]
    \draw (1.4, 0) circle (.2);
    \draw (2.8, 0) circle (.2);
    \draw (4.2, 0) circle (.2);
    \fill (4.2, 0) circle (2pt);
    \draw (0.2, 0) -- (1.2, 0);
    \draw (1.6, 0) -- (2.6, 0);
    \draw (3.0, 0) -- (4.0, 0);
    \node at (-0.1, 0) {$\cdots$};
    \node at (4.7, 0) {$\cdots$};
    \node at (1.4, -0.5) {\small $f_q$};
    \node at (4.2, -0.5) {\small $f_q+2$};
    \node at (0.7, -0.25) {\small $s_{f_q}$};
    \node at (0.7, -0.25) {\small $s_{f_q}$};
    \node at (2.1, -0.25) {\small $s_{f_q+1}^{\prime}$};
    \node at (3.5, -0.25) {\small $s_{f_q+2}^{\prime}$};
\end{tikzpicture}
,\end{alignat}
such that
\ie
\hat{h}_{f_{q}}^{(R)} |\xi\rangle + \hat{h}_{f_{q}+2}^{(L)} |\xi^{\prime} \rangle = 0
,\fe
where $\hat{h}_m^{(R)}$ and $\hat{h}_m^{(L)}$ are defined in Eq.~\eqref{eqn: decomposition of the Hamiltonian}.
\end{lemma}
\begin{proof}
Simply choose
\ie
| s_{f_q+1}^\prime \rangle = 
\left(
\begin{matrix}
-\gamma_{f_q} & J_{f_q}
\\
0 & -\gamma_{f_q}
\end{matrix}
\right)
| s_{f_q+1} \rangle
,\fe
and
\ie
| s_{f_q+2}^\prime \rangle = 
-
\left(
\begin{matrix}
-1/\gamma_{f_q+1} & 0
\\
-J_{f_q+1} / \gamma_{f_q+1}^2 & -1/\gamma_{f_q+1}
\end{matrix}
\right)
| s_{f_q+2} \rangle
,\fe
which is a direct consequence of Eq.~\eqref{eqn: decomposition of the Hamiltonian}. Note that the normalization factors are neglected here for convenience.
\end{proof}

Consider the $Q=2$ sector with even system-size $M$. Starting with a component $| \Phi_{(1, 2)} \rangle$ with $(f_1, f_2) = (1, 2)$ (the spin states are arbitrary, so we are essentially constructing $2^{M-1}$ zero modes), for which all other sites are empty (hence the above Lemma applies with $f_q = 2$), we can iteratively construct the rest of the components that constitute the zero mode (note that it is important to require $M$ being even):
\begin{enumerate}
\item Based on Lemma \ref{lemma: move to right cancelled by move to left}, the spin states for the component $| \Phi_{(1, 4)} \rangle$ with $(f_1, f_2) = (1, 4)$ can be inferred, such that Lemma \ref{lemma: move to right cancelled by move to left} with $f_q=2$ guarantees that 
\ie
h_{2}^{(R)} |\Phi_{(1, 2)}\rangle + h_{4}^{(L)} | \Phi_{(1, 4)} \rangle = 0
.\fe
Similarly, from $| \Phi_{(1, 4)} \rangle$ the spin states of the component with $(f_1, f_2) = (1, 6)$ can be inferred similar to the above, and so on, until the component with $(f_1, f_2) = (1, M)$ is constructed. In other words, once we know the spin states for the component with $(f_1, f_2) = (1, m)$, the spin states of the component with $(f_1, f_2) = (1, m+2)$ can be inferred.
\item Once we know the spin states for the component with $(f_1, f_2) = (m_1, m_2)$ with $m_2 > m_1 + 2$, the component with $(f_1, f_2) = (m_1+2, m_2)$ can be inferred base on Lemma \ref{lemma: move to right cancelled by move to left}.
\item We then iterate the above procedure until all components satisfying the conditions $f_2 > f_1$, where $f_2$ is even and $f_1$ is odd, are obtained.
\item The resulting zero mode is given by 
\ie
\ket{\Psi} = 
\sum_{m_2=m_1}^{M/2}
\sum_{m_1=1}^{M/2} \ket{\Phi_{(2m_1-1,2m_2)}}
.\fe
\end{enumerate}
The above procedure is diagramatically represented in \figref{fig: Graghically summary of constructing Q=2 zero modes}. Since we can arbitrarily choose $2^{M-1}$ spin states in $| \Phi_{(1, 2)} \rangle$, the procedure produces $2^{M-1}$ zero modes:
\ie
D(\mathcal{K}^{(1)}) \ge 2^{M-1}
.\fe

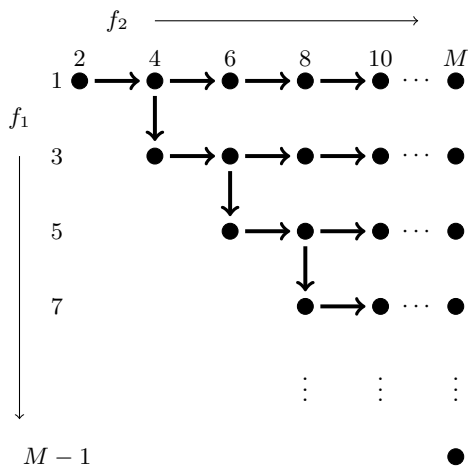
\begin{figure}[t]
\centering
\begin{tikzpicture}
\draw[fill=black] (0,0) circle (3pt);
\draw[fill=black] (1,0) circle (3pt);
\draw[fill=black] (2,0) circle (3pt);
\draw[fill=black] (3,0) circle (3pt);
\draw[fill=black] (4,0) circle (3pt);
\draw[fill=black] (5,0) circle (3pt);
\draw[fill=black] (1,-1) circle (3pt);
\draw[fill=black] (2,-1) circle (3pt);
\draw[fill=black] (3,-1) circle (3pt);
\draw[fill=black] (4,-1) circle (3pt);
\draw[fill=black] (5,-1) circle (3pt);
\draw[fill=black] (2,-2) circle (3pt);
\draw[fill=black] (3,-2) circle (3pt);
\draw[fill=black] (4,-2) circle (3pt);
\draw[fill=black] (5,-2) circle (3pt);
\draw[fill=black] (3,-3) circle (3pt);
\draw[fill=black] (4,-3) circle (3pt);
\draw[fill=black] (5,-3) circle (3pt);
\draw[fill=black] (5,-5) circle (3pt);
\draw[->, line width=1.5pt] (0.2,0) -- (0.8,0);
\draw[->, line width=1.5pt] (1.2,0) -- (1.8,0);
\draw[->, line width=1.5pt] (2.2,0) -- (2.8,0);
\draw[->, line width=1.5pt] (3.2,0) -- (3.8,0);
\draw[->, line width=1.5pt] (1.2,-1.0) -- (1.8,-1.0);
\draw[->, line width=1.5pt] (2.2,-1.0) -- (2.8,-1.0);
\draw[->, line width=1.5pt] (3.2,-1.0) -- (3.8,-1.0);
\draw[->, line width=1.5pt] (2.2,-2.0) -- (2.8,-2.0);
\draw[->, line width=1.5pt] (3.2,-2.0) -- (3.8,-2.0);
\draw[->, line width=1.5pt] (3.2,-3.0) -- (3.8,-3.0);
\draw[->, line width=1.5pt] (1.0,-0.2) -- (1.0,-0.8);
\draw[->, line width=1.5pt] (2.0,-1.2) -- (2.0,-1.8);
\draw[->, line width=1.5pt] (3.0,-2.2) -- (3.0,-2.8);
\draw[->] (1, 0.8) -- (4.5, 0.8);
\draw[->] (-0.8, -1) -- (-0.8, -4.5);
\node at (-0.8, -0.5) {$f_1$};
\node at (0.5, 0.8) {$f_2$};
\node at (-0.3, 0.0) {$1$};
\node at (-0.3, -1.0) {$3$};
\node at (-0.3, -2.0) {$5$};
\node at (-0.3, -3.0) {$7$};
\node at (-0.3, -5.0) {$M-1$};
\node at (0.0, 0.3) {$2$};
\node at (1.0, 0.3) {$4$};
\node at (2.0, 0.3) {$6$};
\node at (3.0, 0.3) {$8$};
\node at (4.0, 0.3) {$10$};
\node at (5.0, 0.3) {$M$};
\node at (4.5, 0.0) {$\cdots$};
\node at (4.5, -1.0) {$\cdots$};
\node at (4.5, -2.0) {$\cdots$};
\node at (4.5, -3.0) {$\cdots$};
\node at (3.0, -4.0) {$\vdots$};
\node at (4.0, -4.0) {$\vdots$};
\node at (5.0, -4.0) {$\vdots$};
\end{tikzpicture}
\caption{A diagrammatic representation of the procedure for constructing zero modes via Lemma $1$. The upper-left point corresponds to the component with $(f_1, f_2) = (1, 2)$, whose spin states can be randomly chosen (there are $2^{M-1}$ independent choices). An arrow pointing from component $A$ to component $B$ implies that the spin states of component $B$ can be inferred from component $A$ by Lemma \ref{lemma: move to right cancelled by move to left}. Note that the arrows always point from top to down and from left to right.}
\label{fig: Graghically summary of constructing Q=2 zero modes}
\end{figure}

However, there is an ambiguity in drawing the arrows in \figref{fig: Graghically summary of constructing Q=2 zero modes}: when considering components that form a square and given the knowledge of the spin state of the upper-left component (corresponding to $(f_1, f_2)$, where $f_2 > f_1 + 2$), we can choose either of the following steps for constructing the bottom-right component (corresponding to $(f_1+2, f_2+2)$):
\begin{alignat}{1}
\begin{tikzpicture}
\draw[fill=black] (0,0) circle (2pt);
\draw[fill=black] (1,0) circle (2pt);
\draw[fill=black] (1,-1) circle (2pt);
\draw[fill=black] (0,-1) circle (2pt);
\draw[->, line width=1.pt] (0.2,0) -- (0.8,0);
\draw[->, line width=1.pt] (0.2,-1) -- (0.8,-1);
\draw[->, line width=1.pt] (0,-0.2) -- (0,-0.8);
\end{tikzpicture}
\;\;\;\;\;\;\;\;\;\;\;\;\;\;\;\;
\begin{tikzpicture}
\draw[fill=black] (0,0) circle (2pt);
\draw[fill=black] (1,0) circle (2pt);
\draw[fill=black] (1,-1) circle (2pt);
\draw[fill=black] (0,-1) circle (2pt);
\draw[->, line width=1.pt] (0.2,0) -- (0.8,0);
\draw[->, line width=1.pt] (0,-0.2) -- (0,-0.8);
\draw[->, line width=1.pt] (1,-0.2) -- (1,-0.8);
\end{tikzpicture}
\end{alignat}
The bottom-right component constructed from either procedure is identical because moving horizontally and vertically along the arrows alters the spin states on distinct sites, thus the two moves commute. Namely, moving horizontally alters $s_{f_1+1}$ and $s_{f_1+2}$, while moving vertically alters $s_{f_2+1}$ and $s_{f_2+2}$. The final outcome is invariant under interchanging the order.

The above procedure extends to systems with even $M$ and sectors with even $Q$ and to systems with odd $M$ and with odd $Q$ sectors. In these cases, the 2D lattices depicted in \figref{fig: Graghically summary of constructing Q=2 zero modes} are replaced with $Q$-dimensional cubic lattices. 

Finally, let us recall from the discussion in Sec.~\ref{sec:model} that models with odd system size $M$ in the odd $Q$ sectors share their spectrum with the $Q^\prime = M - Q$ sector. Consequently, we can deduce the existence of zero modes in the even $Q^\prime$ sector from that in the even $Q$ sector. Thus, aside from the case of even $M$ and odd $Q$, we have analytically constructed zero modes for any system size and any symmetry sector. Since our procedure produces $2^{M-1}$ zero modes, the number of zero modes in such symmetry sectors are lower bounded by $2^{M-1}$, which grows exponentially with $M$ and thus establishes HSF.


\section{Krylov subspaces and the Entanglement Entropy}

The von Neumann entanglement entropy (EE) is an important indicator of non-ergodicity in classically fragmentated systems~\cite{khemani2020hilbert,sala2019ergodicity}. By dividing the system into two subregions $A$ and $B$, the EE of subregion $A$ is defined as:
\ie
S_A(\rho) = -\tr \left[ \rho_A \log \rho_A \right]
,\fe
where $\rho_A$ is the reduced density matrix of subregion $A$. In this paper, the subregion $A$ is defined as the left half part of the system consisting of sites $1 < m < \lfloor M/2 \rfloor$. The long-time behavior of the EE, when starting from a random Haar product state as the initial configuration, sheds light on the system's ergodicity (or lack thereof): a random Haar product state (of which on-site tensor product states are a subset) has zero EE:
\ie
| \psi_0 \rangle = | \psi_A \rangle \otimes | \psi_B \rangle
\Rightarrow
S_A(| \psi_0 \rangle) = 0
,\fe
and as time evolves, the EE of the state increases. For an ergodic system, the state will explore the entire Hilbert space, and is expected to behave like a random Haar state in the entire Hilbert space in the long time limit. In the presence of global symmetries, if we pick a random Haar product state with a fixed symmetry charge and, for an ergodic system, it is expected to explore the entire Hilbert space within that symmetry sector. Consequently, the EE at $t \rightarrow \infty$ saturates the Page value~\cite{Page_1993}, which is defined as the average EE of random Haar states between two bi-partitioned systems. On the other hand, if the EE fails to saturate the Page value at $t \rightarrow \infty$, it indicates that the system is non-ergodic.

In this Section, we study the time-evolution from random Haar product states of the EE within different symmetry sectors.  We will focus on equally bi-partitioned systems. The results clearly deviate from the expected Page values, validating the analytic results we established in prior sections.

\subsection{Entanglement Entropy: $Q=1$}

\begin{figure}[t]
\centering
\label{}
\subfigure[]{\label{fig: expEE_Q1_t10_J8_W0_hmu0_odd}\includegraphics[width=80mm]{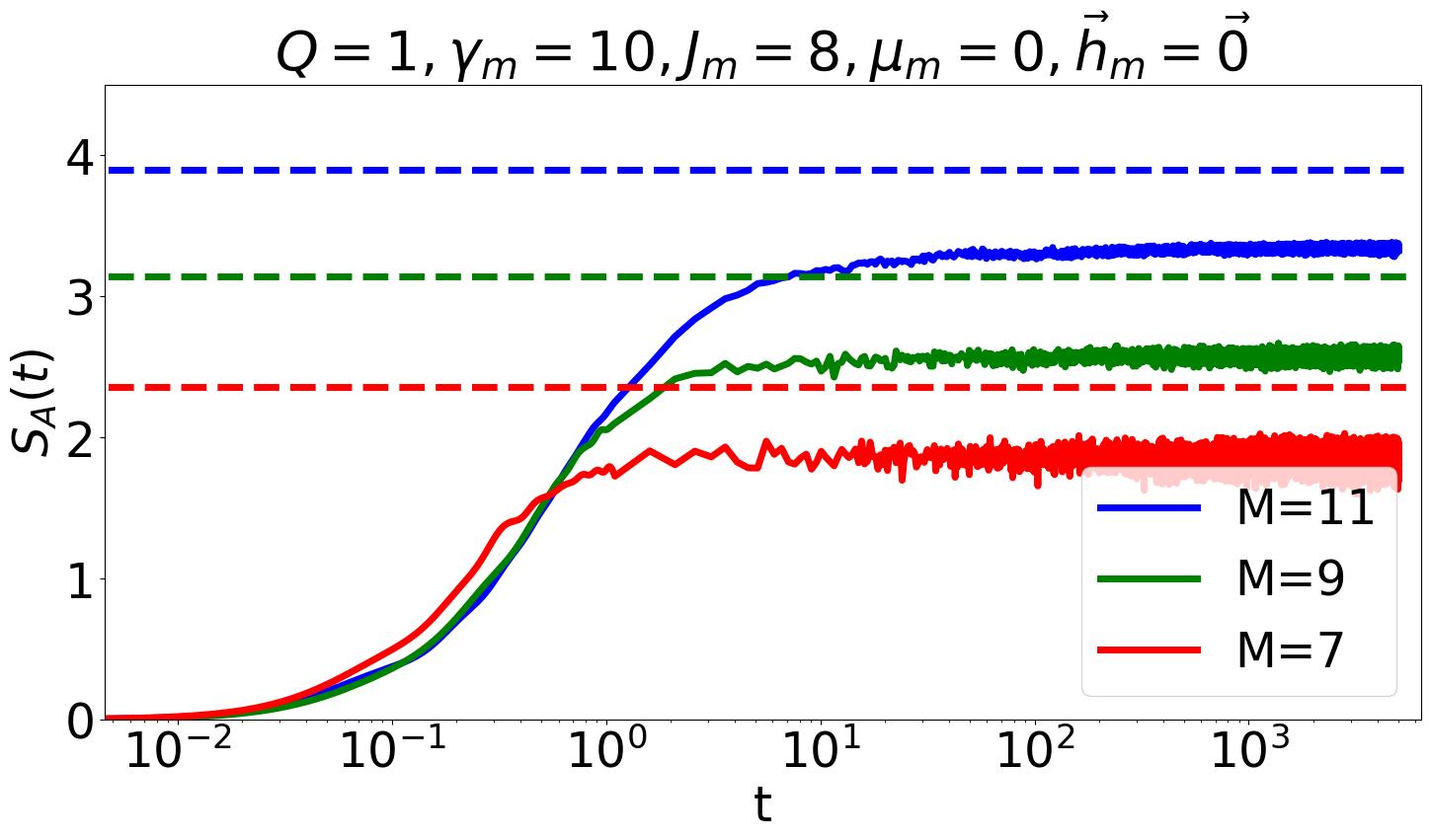}}
\subfigure[]{\label{fig: expEE_diff_Q1_t10_J8_odd}\includegraphics[width=80mm]{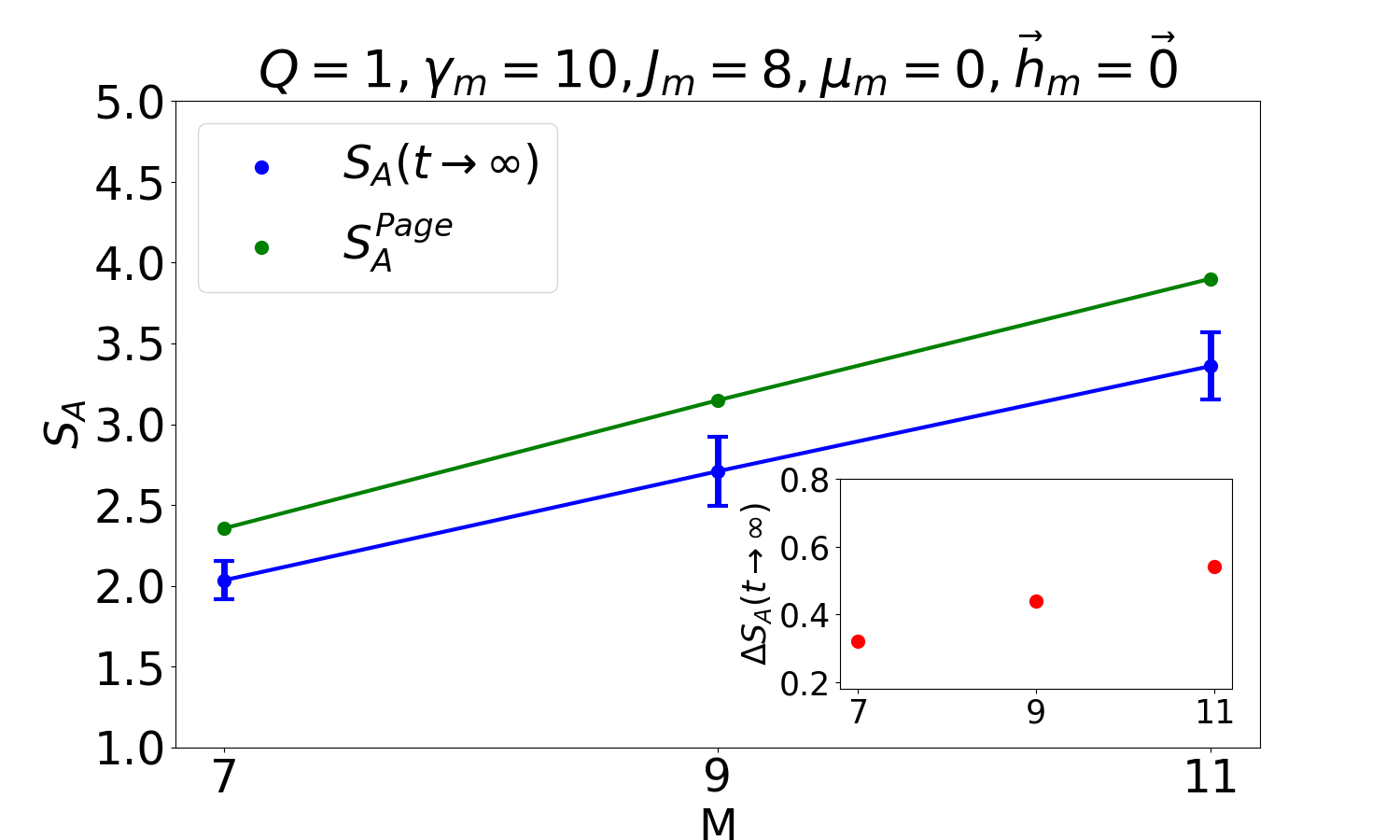}}
\caption{(a) Time evolution of EE for system sizes $M=7, 9, 11$ in the $Q=1$ sector: $\gamma_m=10, J_m=8, W=0, \vec{h}_m=\vec{0}$. The dashed lines are the corresponding Page values, obtained by averaging the EE over 50 random many-body states. The asymptotic value of all the curves are smaller than their corresponding Page values. (b) Comparison of the late time EE and the Page values. Each data point and its associated standard deviation are based on 20 realizations of random initial Haar product states. The inset shows the deviation between the average late time EE's and the Page values.}
\label{fig: expEE Q=1 odd}
\end{figure}

\begin{figure}[t]
\centering
\label{}
\subfigure[]{\label{fig: expEE_Q1_t10_J8_W0_hmu0_even}\includegraphics[width=80mm]{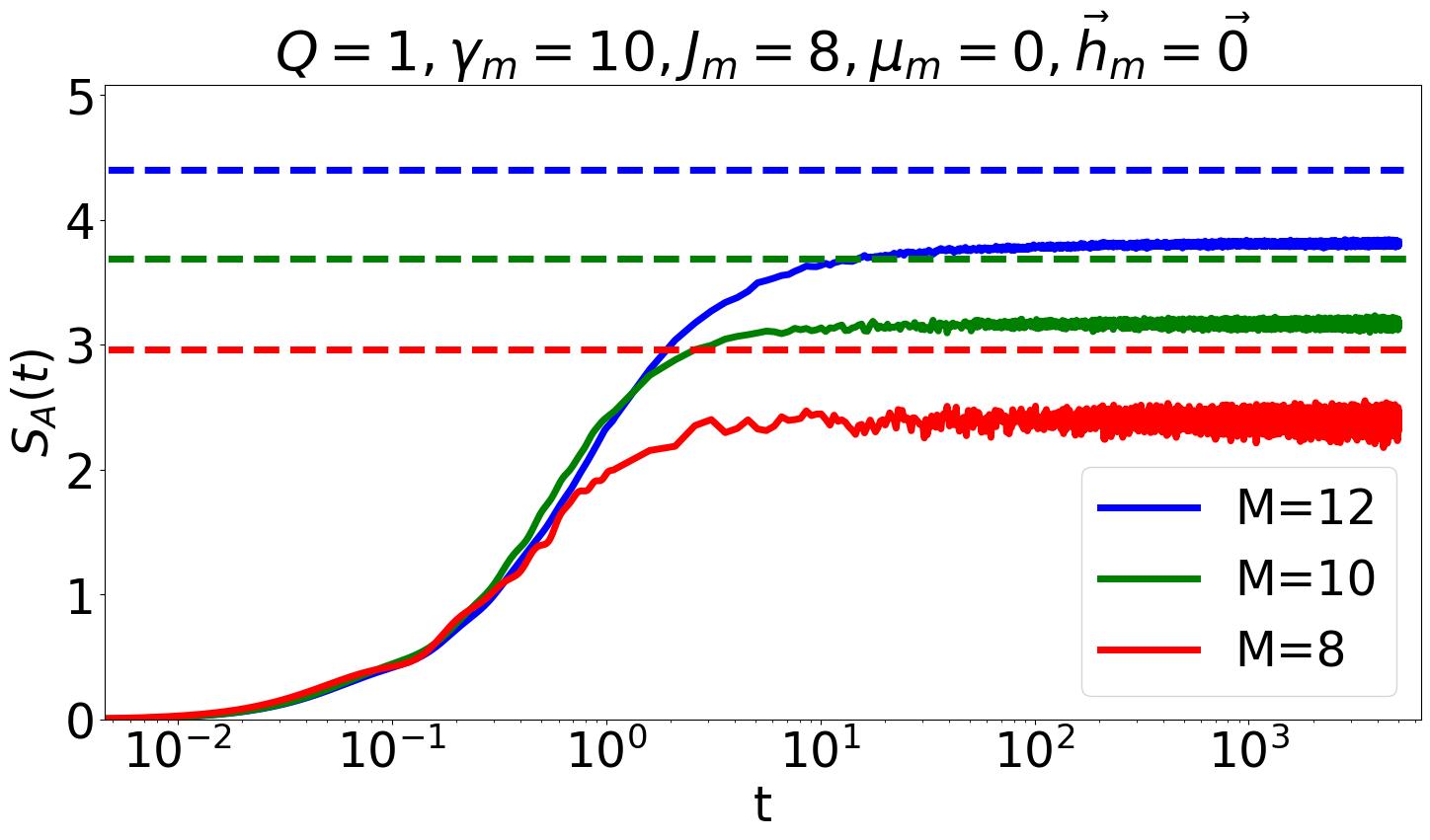}}
\subfigure[]{\label{fig: expEE_diff_Q1_t10_J8_even}\includegraphics[width=80mm]{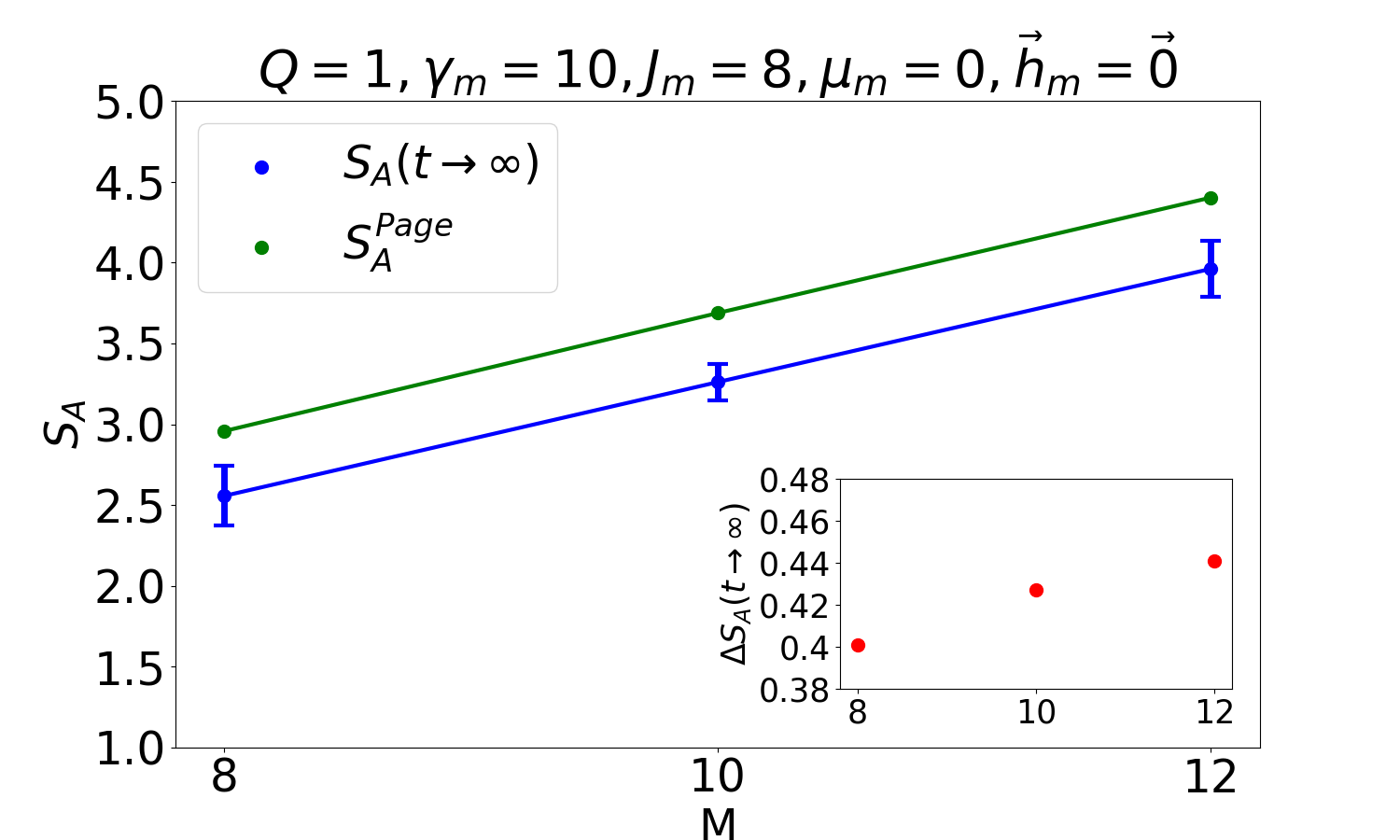}}
\caption{(a) Time evolution of EE for system sizes $M=8, 10, 12$ in the $Q=1$ sector: $\gamma_m=10, J_m=8, W=0, \vec{h}_m=\vec{0}$. The Page values are obtained by averaging over 50 random many-body states. (b) Comparison of the late time EE and the Page values. Each data point and its associated standard deviation are based on 20 realizations of random initial Haar product states.}
\label{fig: expEE Q=1 even}
\end{figure}

In the symmetry charge $Q = 1$ sector of model $H_{\bar{\gamma} \bar{J}}$ in~\figref{fig: expEE_Q1_t10_J8_W0_hmu0_odd}, we show the evolution of the EE for systems with odd $M$ starting from arbitrary Haar random product states. All the EE curves (solid curves) fail to saturate the corresponding Page values (dashed curves); this numerical observation is consistent with our analytic identification of exponentially many Krylov subspaces within the $Q = 1$ sector for odd system sizes $M$. 

We note that while the precise value of the EE at $t \rightarrow \infty$ has some initial state dependence, the EE fails to saturates the corresponding Page value for any choice of random initial state, \emph{despite the initial state having random components in all the Krylov subspaces}. In \figref{fig: expEE_diff_Q1_t10_J8_odd}, we present the EE at late times $S_A(t \rightarrow \infty)$ for various system sizes, and compare the differences with their respective Page values. In the main figure, each data point and its associated standard deviation averaged over 20 random realizations of random initial Haar product states. For each value of $M$, the Page values consistently exceed the first standard deviation, providing a a clear signal that a product state is unable to explore the entire symmetry sector. Additionally, the inset of \figref{fig: expEE_diff_Q1_t10_J8_odd}, shows the difference between the Page value and the observed late time EE:
\ie
\Delta S_A(t \rightarrow \infty) \equiv S_A^{\text{Page}} - S_A(t \rightarrow \infty)
.\fe
We observe that $\Delta S_A(t \rightarrow \infty)$ increases with increasing system sizes $M$ (for odd $M$), consistent with strong fragmentation.

For models with even $M$, we do not have analytical arguments to estimate the number of all Krylov subspaces. However, \figref{fig: D(K) Q=1} suggests that the number of Krylov subspaces increases rapidly with $M$, and we therefore conjecture that the system displays non-thermal behavior. In \figref{fig: expEE_Q1_t10_J8_W0_hmu0_even}, we show the EE evolution for models with even $M$ starting from arbitrary Haar random product states. One can clearly observe finite deviations of $S_A(t \rightarrow \infty)$ from the Page values for the system sizes we can numerically access, where the size-dependent saturation value can be attributed to the exponentially many fragmented subspaces~\cite{sala2019ergodicity}. We also observe a slight increase in the deviation $\Delta S_A(t\rightarrow\infty)$ with $M$ (see inset of \figref{fig: expEE_diff_Q1_t10_J8_even}), but it is not as significant as to in the case of odd $M$. Nevertheless, the numerical results we have found are consistent with strong fragmentation in the $Q=1$ sector for even $M$, with the caveat that larger system sizes should be analyzed to conclusively determine the nature of non-ergodicity in the thermodynamic limit.

\subsection{Entanglement Entropy: $Q \ge 2$}
\label{subsection: Q >= 2, Entanglement entropy growth}

\begin{figure}[t]
\centering
\subfigure[]{\label{fig: expEE_Q2_t10_J8_W0_hmu0}\includegraphics[width=80mm]{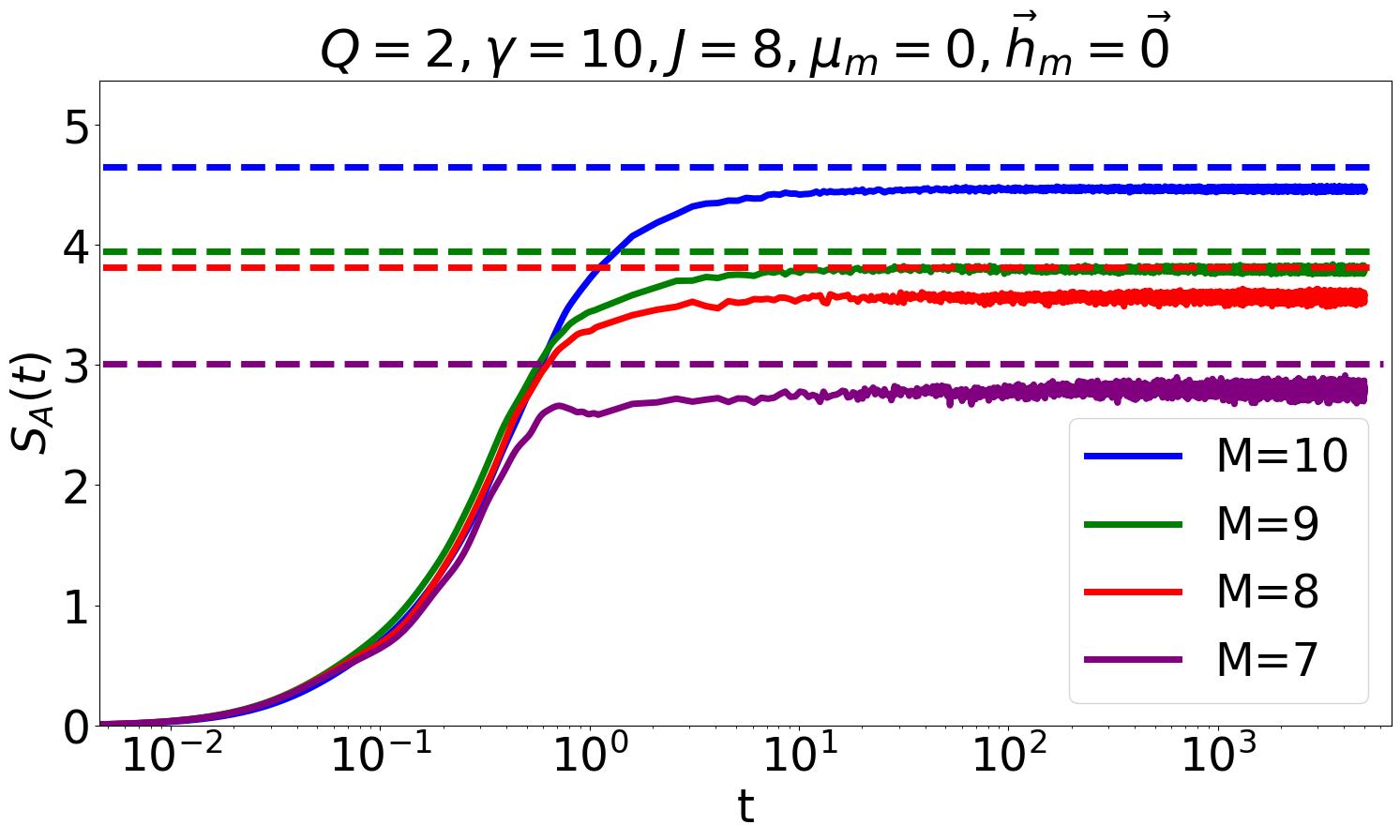}}
\subfigure[]{\label{fig: expEE_diff_Q2_t10_J8}\includegraphics[width=80mm]{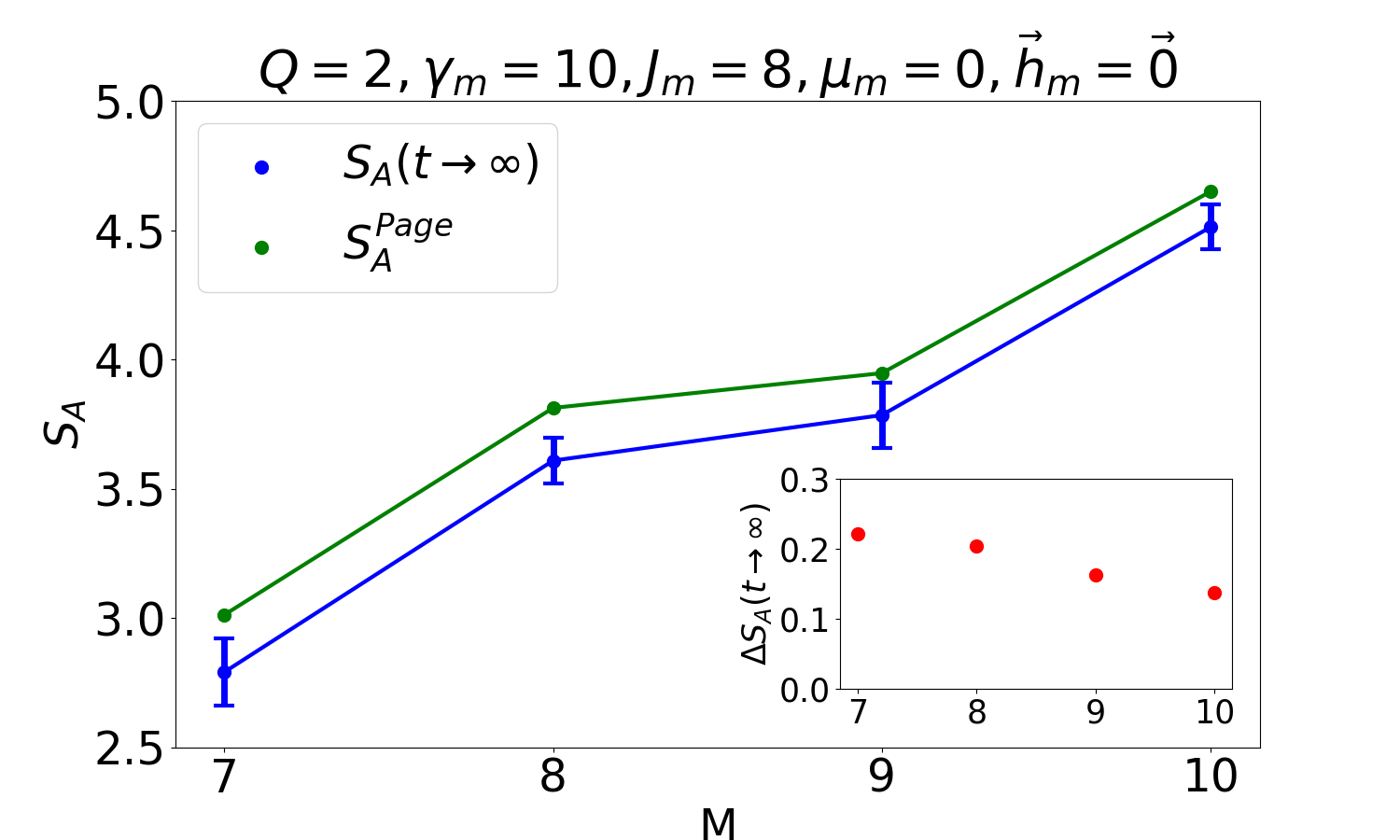}}
\caption{(a) Time evolution of EE for system sizes $M=7, 8, 9, 10$ in the $Q=2$ sector: $\gamma_m=10, J_m=8, W=0, \vec{h}_m=\vec{0}$. The Page values are obtained by averaging over 50 random many-body states. (b) Comparison of the late time EE and the Page values. Each data point and its associated standard deviation are based on 20 realizations of random initial Haar product states. The deviation decreases as $M \rightarrow \infty$, indicating thermalization in the large $M$ limit.}
\label{fig: expEE Q=2}
\end{figure}

\begin{figure}[t]
\centering
\subfigure[]{\label{fig: expEE_Q3_t10_J8_W0_hmu0}\includegraphics[width=80mm]{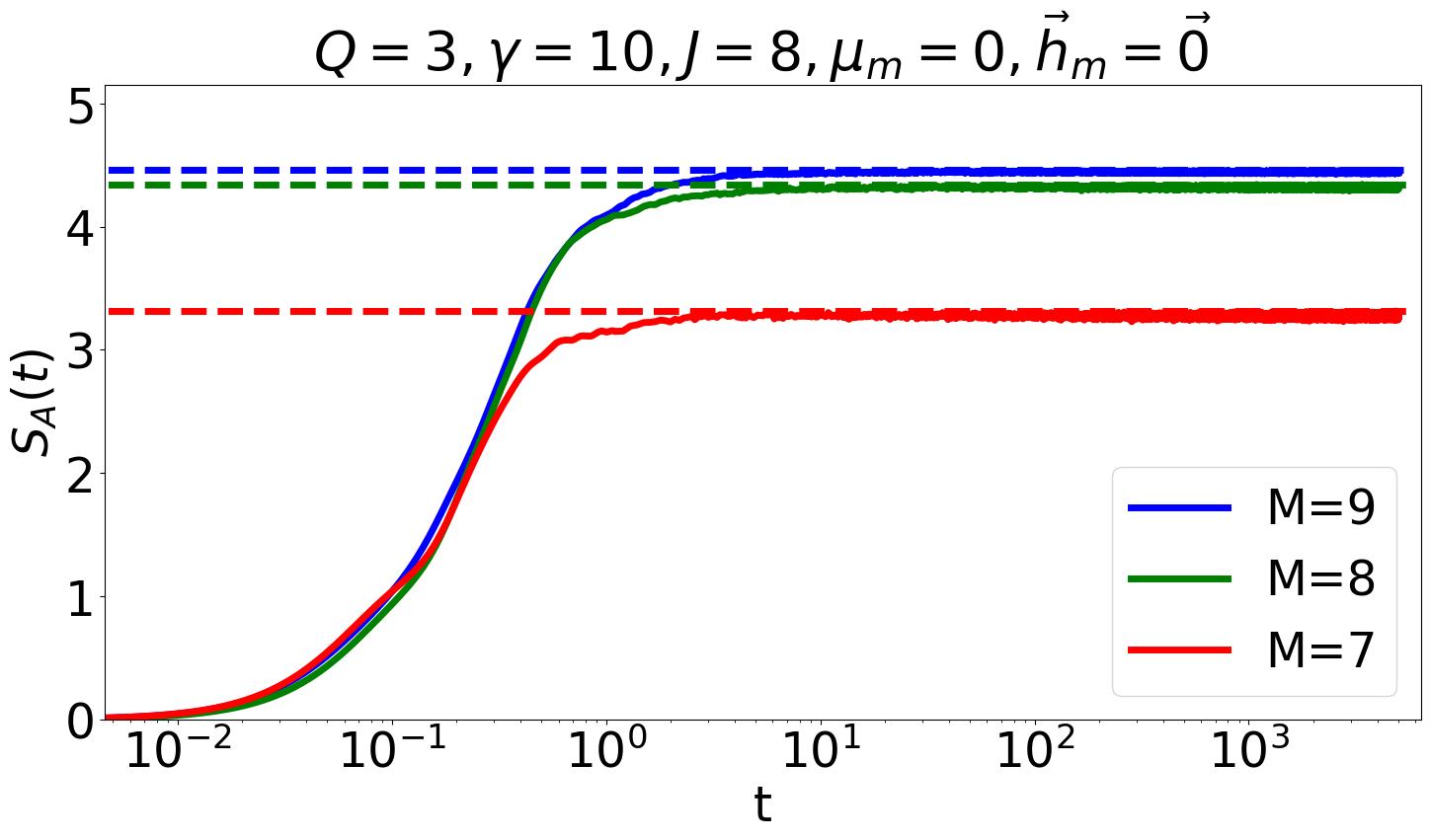}}
\subfigure[]{\label{fig: expEE_diff_Q3_t10_J8}\includegraphics[width=80mm]{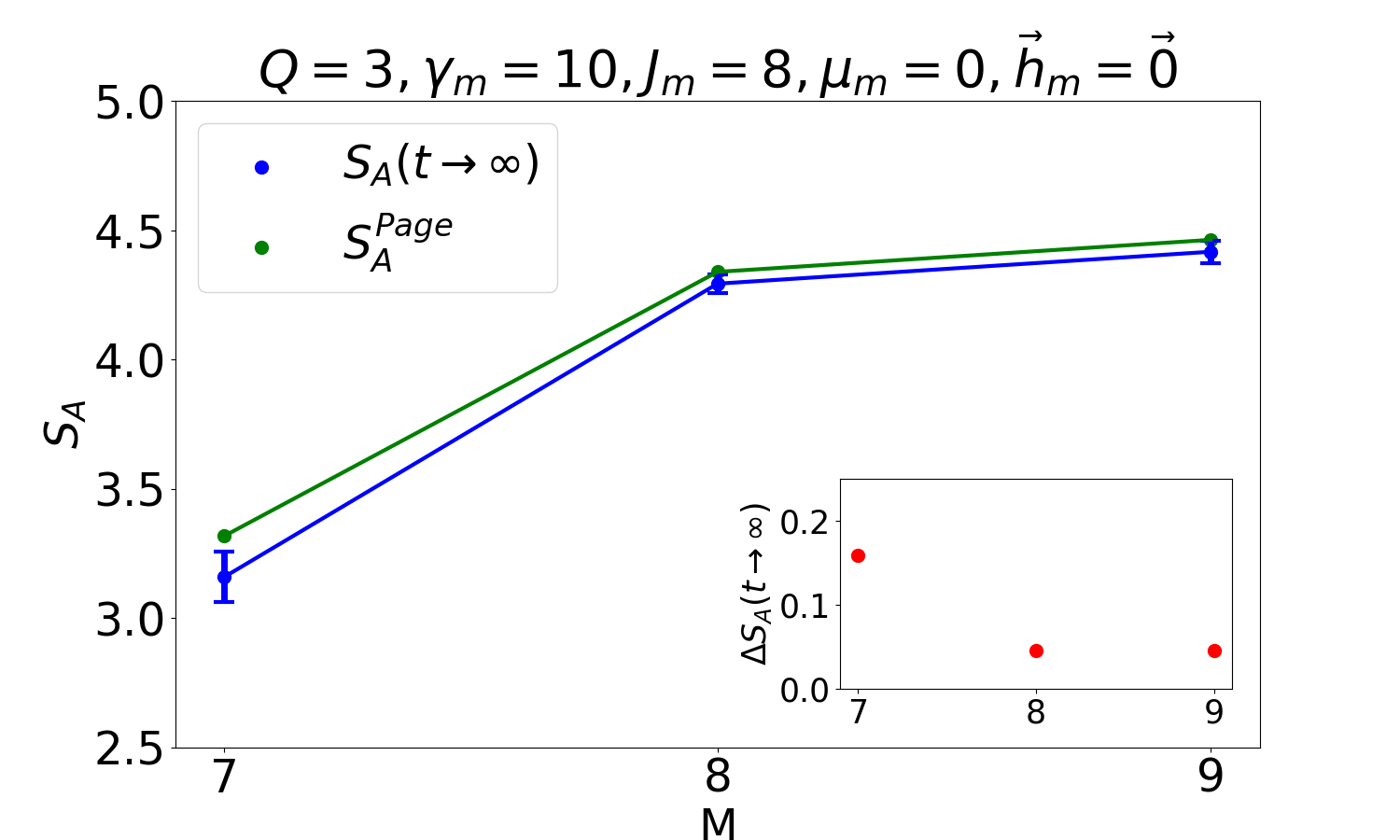}}
\caption{(a) Time evolution of EE for system sizes $M=7, 8, 9$ in the $Q=3$ sector: $\gamma_m=10, J_m=8, W=0, \vec{h}_m=\vec{0}$. The Page values are obtained by averaging over 50 random many-body states. (b) Comparison of the late time EE and the Page values. Each data point and its associated standard deviation are based on 20 realizations of random initial Haar product states. The deviation decreases as $M \rightarrow \infty$, indicating thermalization in the large $M$ limit.}
\label{fig: expEE Q=3}
\end{figure}

In this subsection, we analyze Krylov subspaces in $Q\ge 2$ charge sectors of model $H_{\bar{\gamma} \bar{J}}$ by direct numerical simulations of the EE. In these symmetry sectors, although the system has exponentially many Krylov subspaces, the ratios of the dimension of the largest Krylov subspace and the dimension of the symmetry sector no longer vanish, indicating that the system exhibits only weak fragmentation (i.e., only a violation of strong ETH). In such cases, although the system has non-thermalizing eigenstates, the number of those states is measure-zero within the entire energy spectrum, and the EE is expected to saturate the Page value in the thermodynamic limit $M \rightarrow \infty$.

In \figref{fig: expEE_Q2_t10_J8_W0_hmu0} and \figref{fig: expEE_Q3_t10_J8_W0_hmu0}, the EE curves either saturate the Page values or have only small offsets at $t \rightarrow \infty$. These small offsets can be attributed to the fact that there still exist exponentially many Krylov subspaces even in a weakly fragmented system and, strictly speaking, the offsets are expected to vanish only in the thermodynamic limit. Indeed, as shown in the insets in \figref{fig: expEE_diff_Q2_t10_J8} and \figref{fig: expEE_diff_Q3_t10_J8}, we see that the offset vanishes as the value of $M$ increases. Hence, in the thermodynamic limit, we expect that these offsets will vanish, resulting in an EE that is consistent with an ergodic system.

\subsection{Non-saturation of EE from degeneracy}

The non-saturation of EE in the previous subsection results from the level degeneracy due to the degenerate Krylov subspaces. In a given symmetry sector labeled by $Q$, if we have $I$ distinct Krylov subspaces $\mathcal{K}_{i, \delta}^{(n_i)}$ with degeneracies $d_i$ and dimensions $n_i$ ($i \in [1, I]$, $\delta \in [1, d_i]$), the total Hilbert space dimension is 
\begin{equation}
|\mathcal{H}_Q|=\sum_{i=1}^{I} n_i d_i. 
\end{equation}
A pure state $|\Psi\rangle$ in $\mathcal{H}_Q$ under time evolution can only access one state in each degenerate level subspace, thus the maximal sub-Hilbert space a pure state can access has a dimension (see Appendix.~\ref{appendix: Krylov subspace generated by arbitrary state})
\begin{equation}
|\mathcal{H}^{\text{sub}}_Q(\Psi)|=\sum_{i}^{I} n_i. 
\end{equation}
The degenerate levels $d_i>1$ result in 
\ie
|\mathcal{H}^{\text{sub}}_Q(\Psi)| < |\mathcal{H}_Q|
,\fe
leading to $S_A(t\rightarrow\infty) < S_A^{\text{Page}}$.

Here, we have shown that the presence of degenerate Krylov subspaces can also lead to an offset of the entanglement entropy away from the expected Page values. More generically, such an offset can be attributed to such degenerate subspaces, the existence of exponentially many Krylov subspaces~\cite{Sala_2020, moudgalya2019thermalization}, or a combination thereof.


\section{Effect of a Non-Trivial Magnetic Field}
\label{section: With magnetic field}

We now turn to the effects of turning on a non-trivial magnetic field for the spins, which we generically expect to break the fragmentation structure we have established in prior sections, thereby restoring ergodicity within each charge sector. Here, we will focus only on the $Q = 1$ sector as it captures the richest essential HSF physics of the model. We begin by considering the influence of a uniform magnetic field in the $z$ direction, after which we investigate the effect of a random magnetic field on the dynamics, where we find evidence of an MBL transition. We note that the interplay of fragmentation and disorder was previously studied in other models in Ref.~\cite{herviouMBL2021, liu20232d}, which found signatures of an MBL transition within several Krylov subspaces.

\subsection{$H_{\bar{\gamma} \bar{J} \bar{h}}$: The breakdown cross-over}
\label{subsection: The breakdown cross-over}

\begin{figure}[t]
\centering
\subfigure[]{\label{fig: expEE_t10_J8_W0_hmu5}\includegraphics[width=80mm]{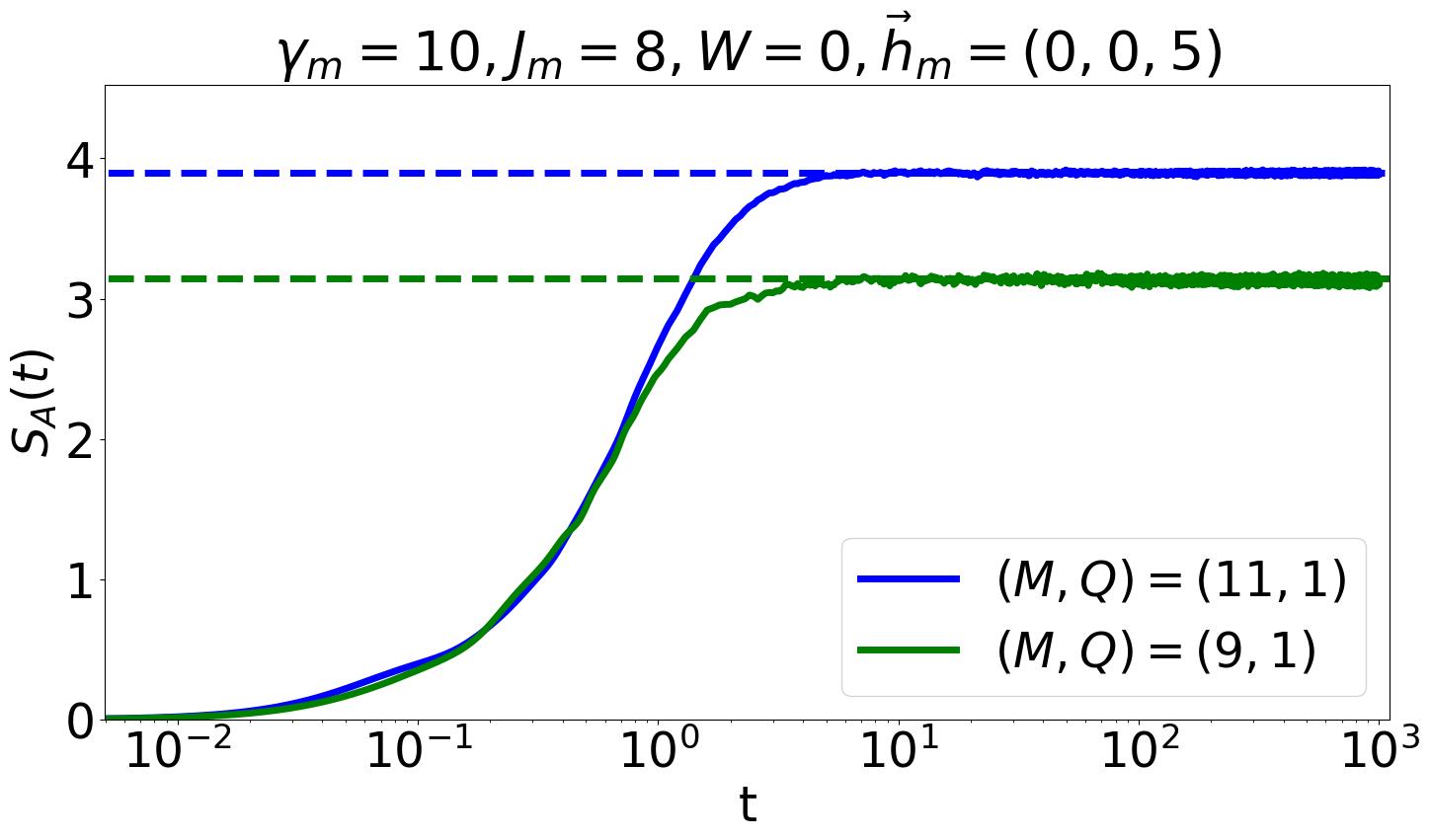}}
\subfigure[]{\label{fig: expNum_fermion_M11_Q1_t10_J8_W0_hmu0}\includegraphics[width=60mm]{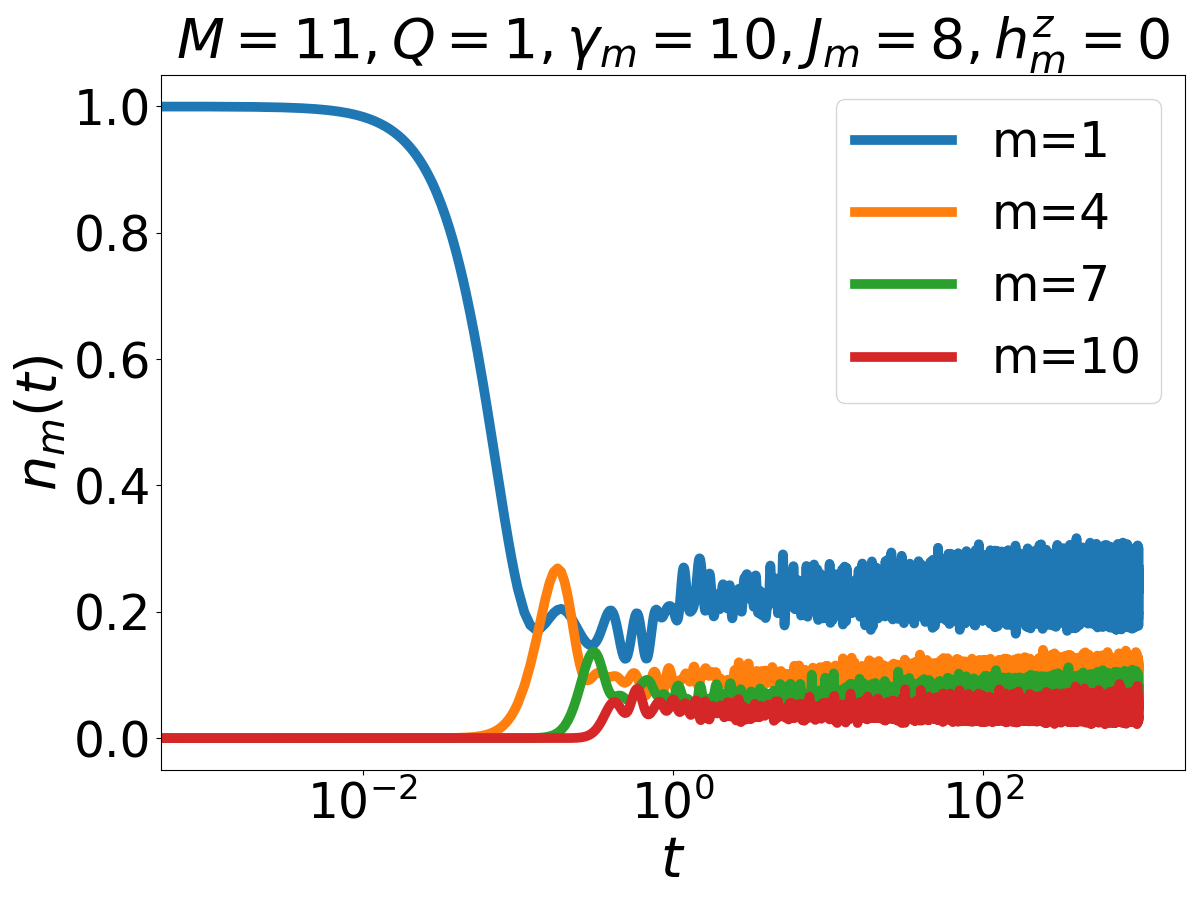}}
\subfigure[]{\label{fig: expNum_fermion_M11_Q1_t10_J8_W0_hmu5}\includegraphics[width=60mm]{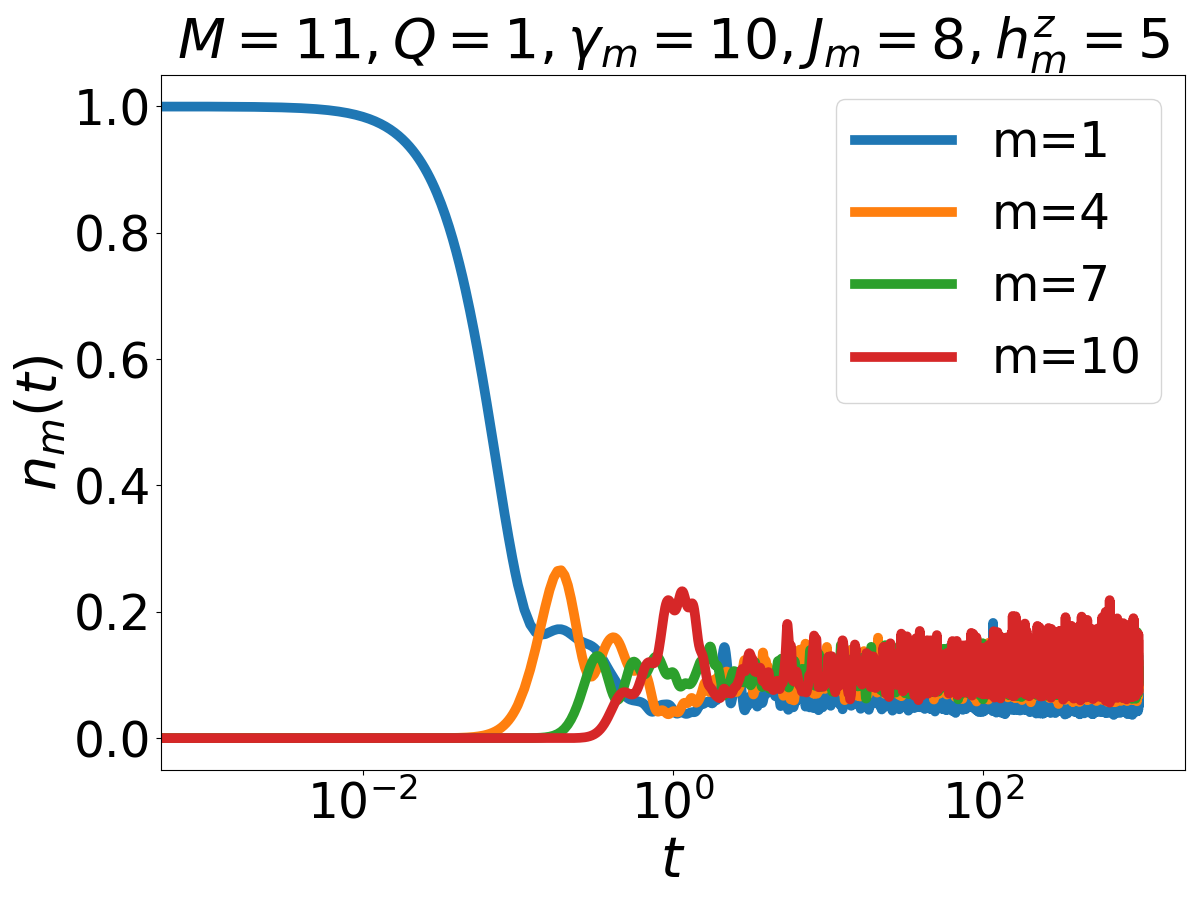}}
\caption{(a) Time evolution of EE of a model with $M=11, \gamma_m=10, J_m=8, h_m^z=5$ in the $Q=1$ sector. The EE saturates the Page value (dashed line), indicating the system is thermal. (b)-(c) Time evolution of $n_m(t)$ at different sites ($m=1, 4, 7, 10$) of a model with $M=11, \gamma_m=10, J_m=8$ in the $Q=1$ sector, with magnetic fields (b) $h_m^z=0$ and (c) $h_m^z=5$. Note that the subscript $m$ has a different meaning for $J_m$ (all $m$) versus $n_m$ (fixed $m$).}
\label{fig: expEE and expNum for h!=0}
\end{figure}

We consider first the role of switching on a small ({with ``smallness" to be defined further on in this section}) \emph{uniform} magnetic field in the z-direction: $\vec{h}_m = (0, 0, h^z)$. As anticipated, the presence of such a non-zero $h^z$ breaks the Hilbert space fragmentation structures that we showed are present for $\vec{h}=\vec{0}$. To see this directly, observe that the root states $|\Psi\rangle$ shown in Sec.~\ref{section: Without magnetic field: Q=1} comprise components with different numbers of spin-up (spin-down) states $N_\uparrow$ ($N_\downarrow$) in the $z$ direction. Once we apply the Hamiltonian $H$ with $h^z \neq 0$ onto a root state, these different components will acquire different coefficients, and $H |\Psi\rangle$ will mix with states in other Krylov subspaces, making the original Krylov subspaces no longer closed by themselves. This argument applies generally to systems of any length $M$. Therefore, we may expect that a non-zero magnetic field $h_z$ causes the system to thermalize.

Numerically, we can verify the thermalization and the breaking of HSF by studying the evolution of the EE. \figref{fig: expEE_t10_J8_W0_hmu5} shows the evolution of EE from a random Haar product state as a function of time. Contrary to $h^z=0$ models, both curves reach their expected maximum Page value as $t\rightarrow\infty$, indicating that the state will explore the entire symmetry sector and that the system is thermalizing.

Thermalization can also be inferred by examining the time evolution of the fermion number 
\ie
n_m(t)=\langle \psi(t) | \hat{n}_m | \psi(t) \rangle
\fe
starting from an initial state with one fermion at the first site (i.e., $n_1(t=0)=1$, $n_{m>1}(t=0)=0$), and with an all-down spin configuration:
\ie
\label{eqn: initial state for spin configuration evolution}
| \psi(t=0) \rangle = c_1^\dagger | \Omega \rangle \otimes |s_2=0\rangle \otimes \cdots \otimes |s_M=0\rangle
.\fe
\figref{fig: expNum_fermion_M11_Q1_t10_J8_W0_hmu5} shows the evolution of $n_m(t)$ in the $(M, Q)=(11, 1)$ sector, whose $\vec{h}_m=(0,0,5)$. At $t\rightarrow \infty$, all sites have approximately equal $n_m(t)$, consistent with the system having reached a thermal state. In contrast, \figref{fig: expNum_fermion_M11_Q1_t10_J8_W0_hmu0} shows the evolution of $n_m(t)$ with $\vec{h}_m=\vec{0}$. At $t\rightarrow \infty$, the system retains the memory of its initial state i.e., the occupation on the first site $n_1(t\rightarrow\infty)$ has more weight compared to that on other sites. This non-uniform distribution of $n_m(t)$ across different $m$ is consistent with a non-ergodic state at vanishing magnetic field and does not persist once the magnetic field is switched on.

\begin{figure}[t]
\centering
\subfigure[]{\label{fig: expNum_M11_Q1_t10_J8_hmu_pauli}\includegraphics[width=80mm]{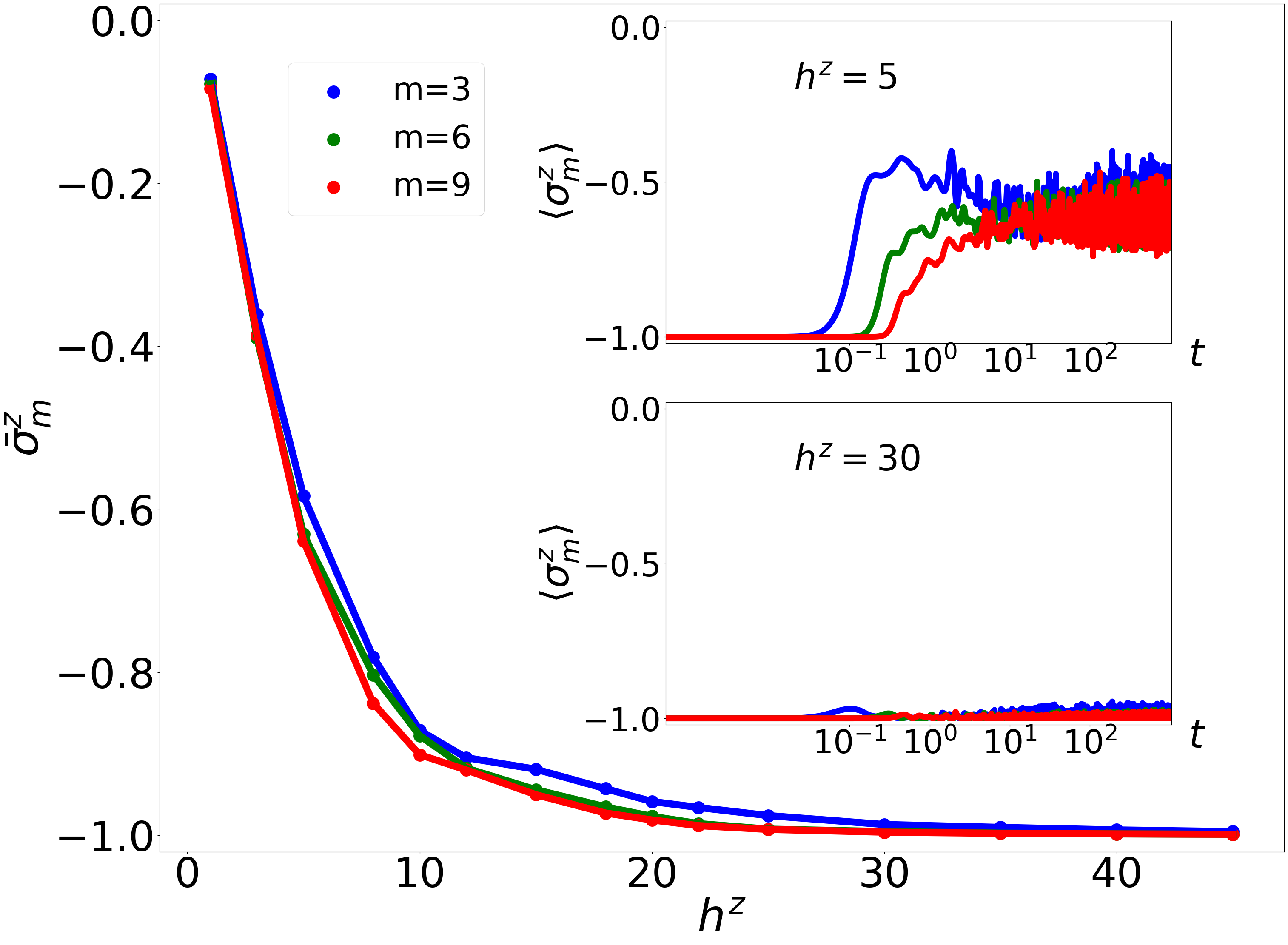}}
\subfigure[]{\label{fig: expNum_M11_Q1_t10_J8_hmu_fermion}\includegraphics[width=85mm]{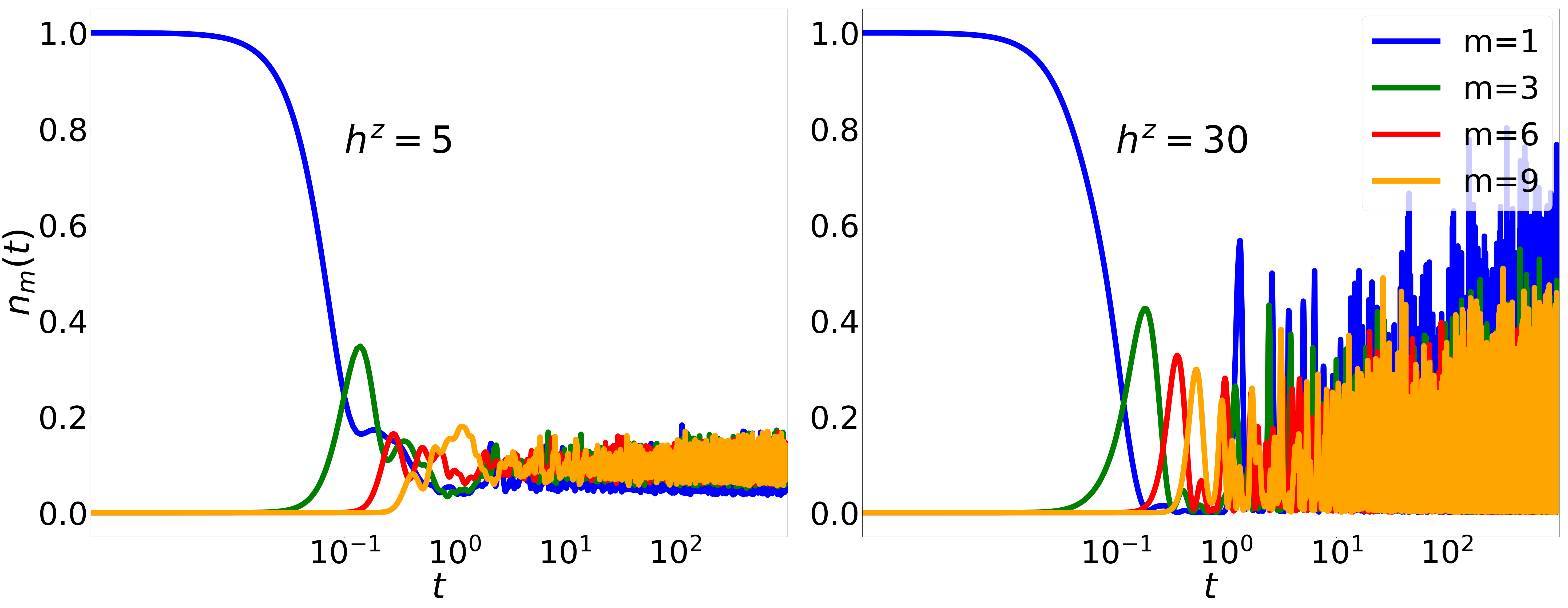}}
\caption{The dynamics of a model with $(M, Q)=(11, 1)$, $\gamma_m=10$, $J_m=8$. The initial state contains a fermion at the first site and a all-down spin chain, see Eq.~\eqref{eqn: initial state for spin configuration evolution}. (a) The main figure shows the time-averaged spin angular momentum along the $z$-axis on the sites $m=3, 6, 9$ as a function of $h^z$. The averaging time range is $t \in [900, 1000]$. The upper inset shows the time evolution of the spin chain for small $h^z=5$, where the the spin configuration evolves from all-down to $\langle \sigma_m^z \rangle \sim -0.45$. The lower inset shows the time evolution of the spin chain for large $h^z=30$, where the spin configuration remains all-down. (b)-(c) The figures show the time evolution of $n_m$ on the sites $m=1, 3, 6, 9$ for $h^z=5$ and $h^z=30$. The fermionic chain remains dynamic regardless of the value of $h^z$.}
\label{fig: expNum_M12_Q1_t10_J8_hmu}
\end{figure}

Let us now examine in more detail how the dynamics of the system depends on the strength $h^z$ of the uniform magnetic field. For small $h^z$, starting with the initial state Eq.~\eqref{eqn: initial state for spin configuration evolution} with an spins pointing down, the evolution of the spin chain is shown in the upper inset of \figref{fig: expNum_M11_Q1_t10_J8_hmu_pauli}. At $t>0$, the spin chain flips and becomes a mixture of spin-up and spin-down orientations. As $h^z$ increases, the system crosses over from a dynamical spin chain to a frozen spin configuration: at larger values of $h^z$, the evolution of the spin chain is shown in the lower inset of \figref{fig: expNum_M11_Q1_t10_J8_hmu_pauli}, where the spin configuration remains all-down at $t>0$, and is hence non-dynamical.

To quantitatively study the dependence of the dynamics of the spin chain on $h^z$, we define the late-time time-averaged spin angular momentum along the $z$-axis:
\ie
\bar{\sigma}_{m}^z = \frac{1}{T} \int_{\tilde{t}}^{\tilde{t}+T} \langle \psi(t) | \sigma_{m}^z | \psi(t) \rangle dt
,\fe
where $\tilde{t}$ is large enough to capture the late-time behavior. \figref{fig: expNum_M11_Q1_t10_J8_hmu_pauli} shows $\bar{\sigma}_{m}^z$ at selected sites $m$ for different $h^z$: for small $h^z$, $\bar{\sigma}_{m}^z > -1$, indicating that the spin chain remains dynamical at late-times (that is, the spins are not frozen in some fixed configuration). As $h^z$ increases, $\bar{\sigma}_{m}^z$ approaches $-1$, implying that the spin chain is static since all spins remain frozen in the spin-down configuration at late times. Therefore, there is a cross-over from a dynamical spin chain to an almost static spin chain as $h^z$ is increased.

On the other hand, let us consider the fermion chain. In \figref{fig: expNum_M11_Q1_t10_J8_hmu_fermion}, we show the evolution of $n_m(t)$ from the initial state Eq.~\eqref{eqn: initial state for spin configuration evolution}. Initially, only the first site is occupied by a fermion, hence $n_1=1$ and $n_{i>1}=0$. As time progresses, both large values of $h^z$ and small values of $h^z$ demonstrate a decrease in $n_1$ and an increase in $n_{i>1}$. Thus, unlike the spins, the fermions remain dynamical no matter what value $h^z$ takes.

This phenomena can be understood by examining the energy levels of the Hamiltonian. First, note that $H_\gamma$ drives the time evolution of the fermions, while $H_J$ drives both the time evolution of the fermions and spins. Therefore, to have non-trivial dynamics for the spins, $H_J$ must contribute non-trivially. We decompose the Hamiltonian into two parts:
\ie
H = (H_h + H_\gamma) + H_J = H_0 + H_J
,\fe
and treat $H_J$ as a perturbation in the limit $h^z \gg \gamma, J$. The energy eigenstates of $H_0$ are
\ie
\label{energy eigenstates of H_0}
| \psi_{\vec{s}, k}^{(0)} \rangle = \sum_{m=1}^{M} \sqrt{\frac{2}{M+1}} \sin \left( \frac{\pi k m}{M+1} \right) | m ; s_2, \cdots, s_M \rangle
,\fe
which corresponds to the energy level
\ie
\label{eqn: eigen-energy of H0}
E_{\vec{s}, k}^{(0)} = h^z(N_\uparrow - N_\downarrow) + 2 \gamma \cos\left(\frac{\pi k}{M+1}\right)
,\fe
where $N_\uparrow$ denotes the number of spin-up states in $s_2, \cdots, s_M$, and $N_\uparrow + N_\downarrow = M-1$. Since every configuration of spins $s_2, \cdots, s_M$ with the identical $N_\uparrow$ and $N_\downarrow$ has the same energy, the degeneracy of Eq.~\eqref{eqn: eigen-energy of H0} is $\frac{(M-1)!}{N_\uparrow! N_\downarrow!}$. From Eq.~\eqref{energy eigenstates of H_0}, we see that the eigenstate is an extended state in the fermionic sector independent of the ratio between $h^z$ and $\gamma$. Therefore, the fermion chain is dynamical for all values of $h^z$, consistent with the numerics. On the other hand, since $H_J$ flips one spin state, the value of $J$ has to be large enough to overcome the energy gap between the energy levels that differs by one spin state in order to activate the spins. According to Eq.~\eqref{eqn: eigen-energy of H0}, the presence of $\gamma$ lowers the energy gap from $2h^z$ to $2h^z - 4\gamma$; therefore, the crossover from a dynamical spin chain to a static spin chain happens approximately when
\ie
J \sim 2 h^z - 4\gamma
.\fe
At $h^z \lesssim 2\gamma + \frac{J}{2}$, the spin chain is dynamical, while at $h^z \gtrsim 2\gamma + \frac{J}{2}$, the spin chain is static.

Finally, we point out that the contributions to the shifts in the energy level due to $H_\gamma$ and $H_J$ are different. Since the effect of $H_J$ changes the spin configuration, any eigenket $| \psi_{\vec{s}, k}^{(0)} \rangle$ (see Eq.~\eqref{energy eigenstates of H_0}) has a different spin configuration from $H_J | \psi_{\vec{s}, k}^{(0)} \rangle$, and we have
\ie
\langle \psi_{\vec{s}, k}^{(0)} | H_J | \psi_{\vec{s}, k}^{(0)} \rangle = 0
,\fe
which implies that the first order shift in the energy level $E_{\vec{s}, k}^{(1)} = 0$. The energy eigenvalue is thus
\ie
E_{\vec{s}, k} = h^z(N_\uparrow - N_\downarrow) + 2\gamma \cos\left(\frac{\pi k}{M+1}\right) + \mathcal{O}\left(\frac{J^2}{h^z}\right)
.\fe
In the limit $h^z \gg J \sim \gamma$, the effect of $\gamma$ appears at $\mathcal{O}(\gamma)$, while the effect of $J$ appears at $\mathcal{O}(J^2)$. This implies that at large $h^z$, the dynamics are dictated predominantly by $H_h$ and $H_\gamma$, while the effect of $H_J$ can be omitted (at least to order $\mathcal{O}(J^2)$)). Since both $H_h$ and $H_\gamma$ preserve the spin configuration $s_2, \cdots, s_M$, the evolution will result in an almost static spin configuration.

\subsection{$H_{\bar{\gamma} \bar{J} h}$: random magnetic field induced MBL}
\label{subsection: MBL from the magnetic field}

So far, we have shown that switching on a non-zero uniform magnetic field for the spins causes the Krylov subspaces to mix, thereby leading the $Q = 1$ charge sector to obey the ETH. We now consider the effect of a random magnetic field and numerically find compelling evidence for an MBL transition. First, let us define the relevant diagnostics. 

An important and often used indicator of chaos is the level spacing statistics (LSS), which is defined as the statistical distribution $p(\delta_E)$ of the nearest neighboring energy level spacing
\ie
\delta_E(\alpha) = E(\alpha+1) - E(\alpha)
.\fe
The LSS of integrable systems shows a Poisson distribution~\cite{Berry1977LevelCI}
\ie
p(\delta_E) \propto e^{-\delta_E/\lambda_0}
,\fe
where $\lambda_0$ is a system-dependent constant. On the other hand, the LSS of chaotic systems shows a Wigner-Dyson (WD) distribution
\ie
p(\delta_E) \propto \delta_E^n e^{-\delta_E^2/\lambda_0^2}
,\fe
where $n=1, 2, 4$ corresponds to the Gaussian orthogonal ensemble (GOE), Gaussian unitary ensemble (GUE), and Gaussian sympletic ensemble (GSE) respectively~\cite{bohigas1984,wigner1967,dyson1970}. The mean level spacing ratio $\langle r \rangle$ enables us to determine the distribution of level spacing more quantitatively, which is defined as the mean of the following ratio~\cite{oganesyan2007mbl}:
\ie
0 \le r(\alpha) = \min \left( \frac{\delta_E(\alpha+1)}{\delta_E(\alpha)}, \frac{\delta_E(\alpha)}{\delta_E(\alpha+1)} \right) \le 1
.\fe
For Poisson, GOE, GUE, GSE distributions, the level spacing ratio $\langle r \rangle$ are $0.39, 0.53, 0.60, 0.67$, respectively~\cite{atas2013distribution}, and a change in the $r$-ratio is typically expected to signal a phase transition (or cross-over) from an ergodic to a non-ergodic (or vice-versa) phase (or regime).

\begin{figure}[t]
\centering
\subfigure[]{\label{fig: exp_M12_Q1_t10_J8_hzR}\includegraphics[width=38mm]{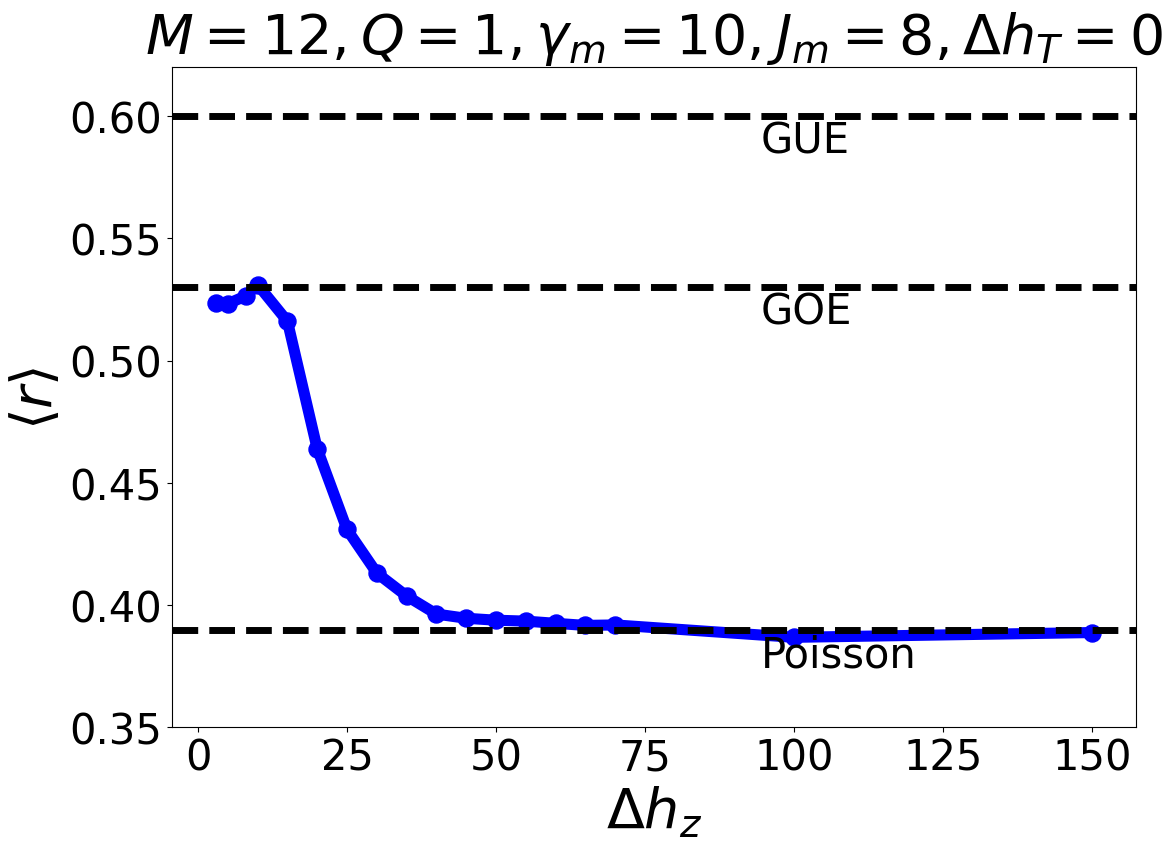}}
\subfigure[]{\label{fig: exp_M12_Q1_t10_J8_hzhT}\includegraphics[width=45mm]{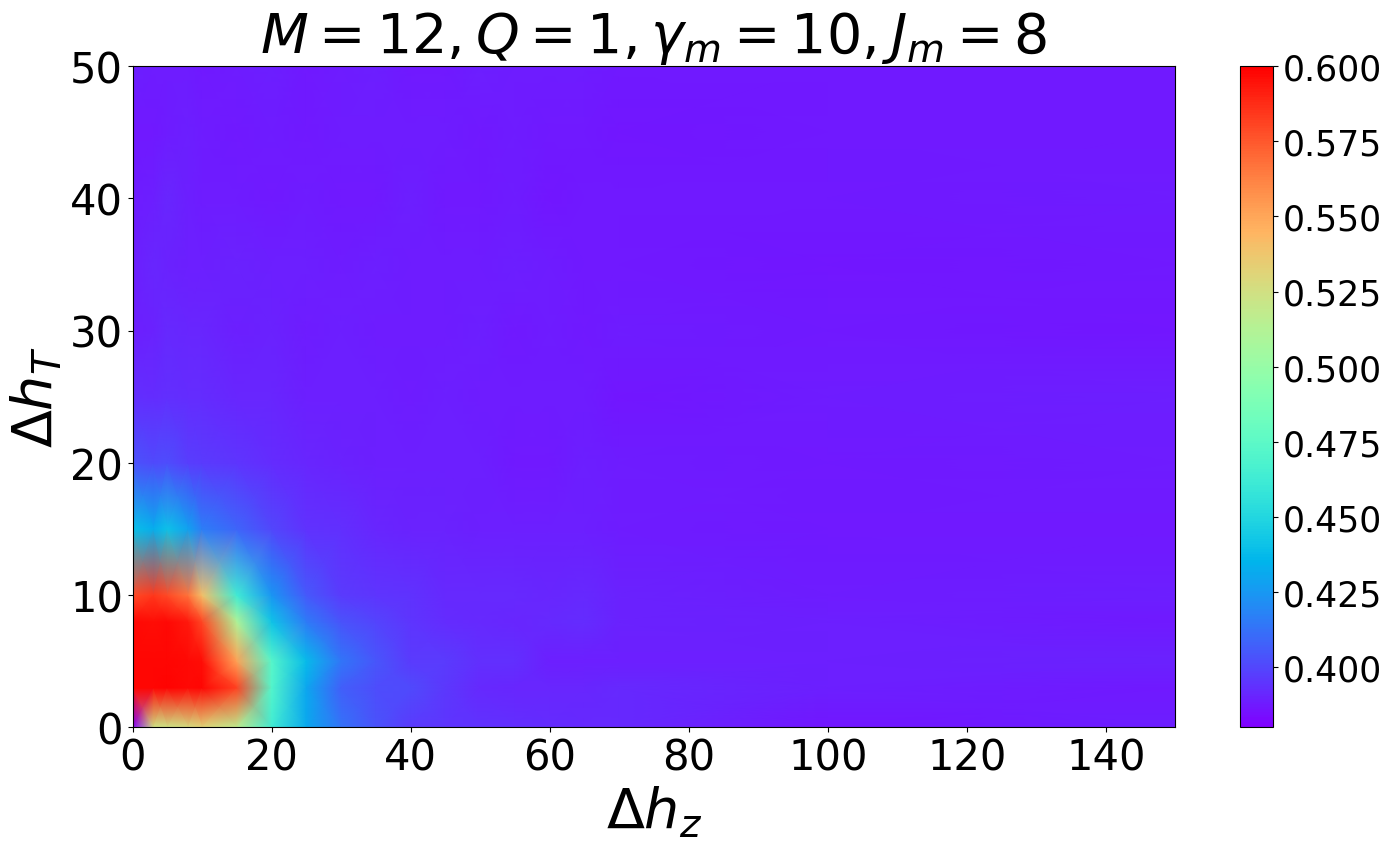}}
\subfigure[]{\label{fig: exp_M10_Q2_t10_J8_hzR}\includegraphics[width=38mm]{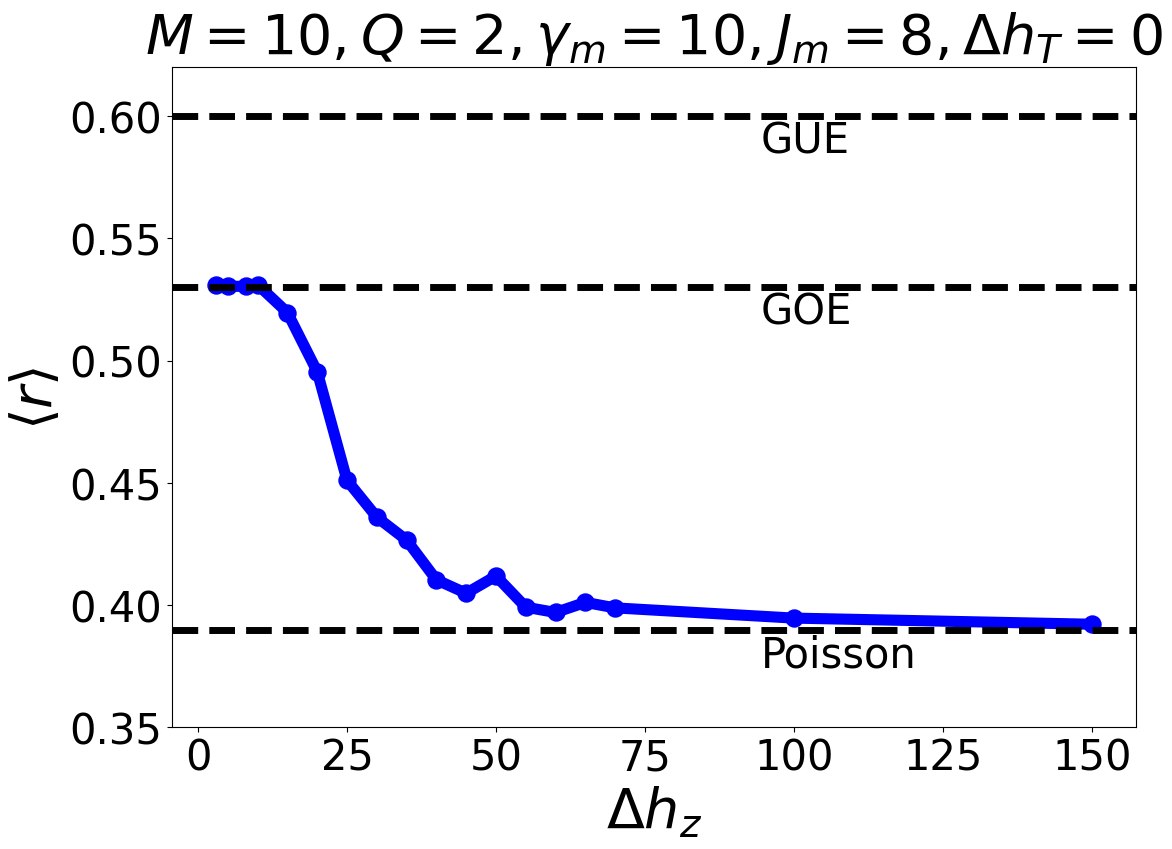}}
\subfigure[]{\label{fig: exp_M10_Q2_t10_J8_hzhT}\includegraphics[width=45mm]{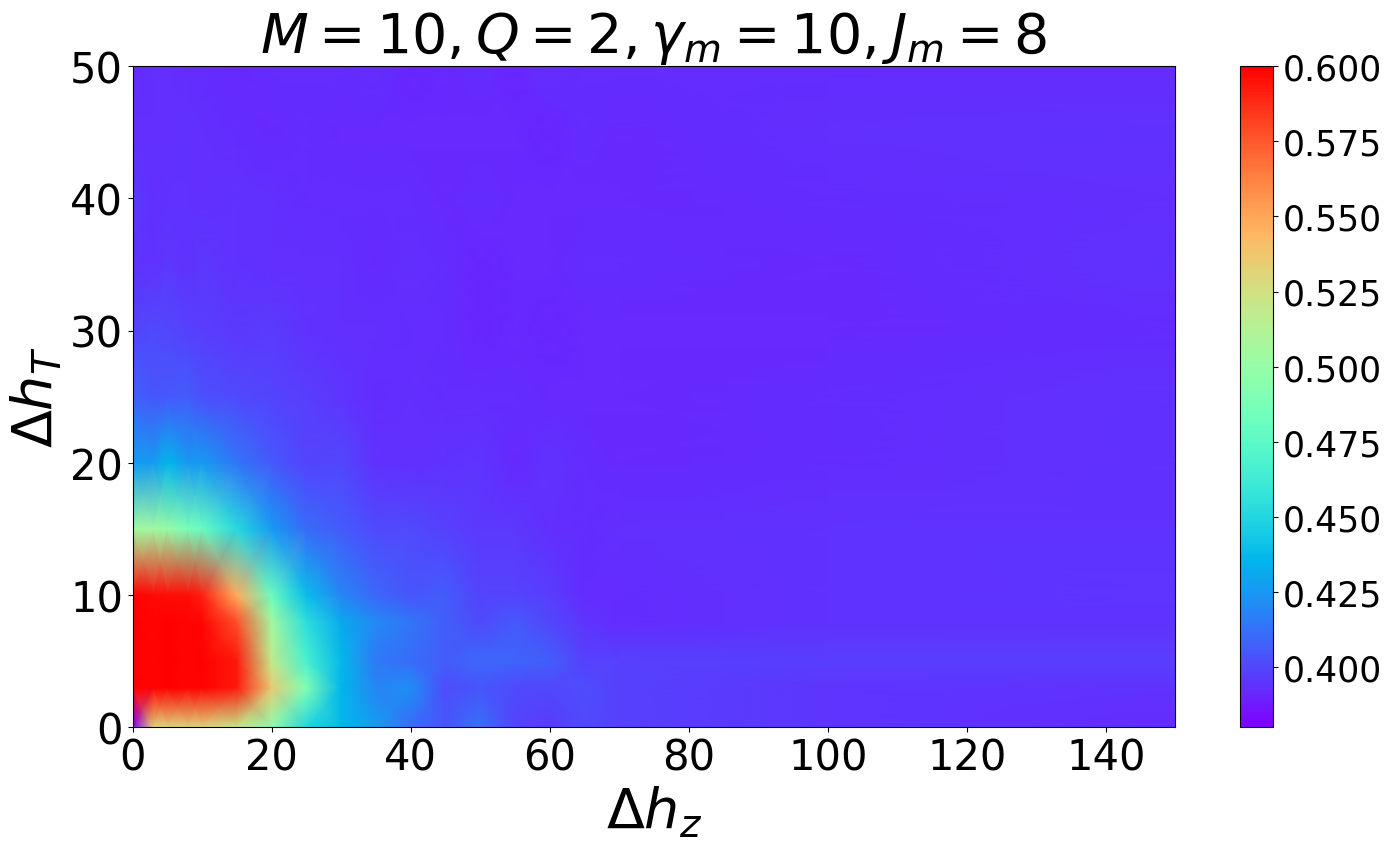}}
\subfigure[]{\label{fig: exp_M9_Q3_t10_J8_hzR}\includegraphics[width=38mm]{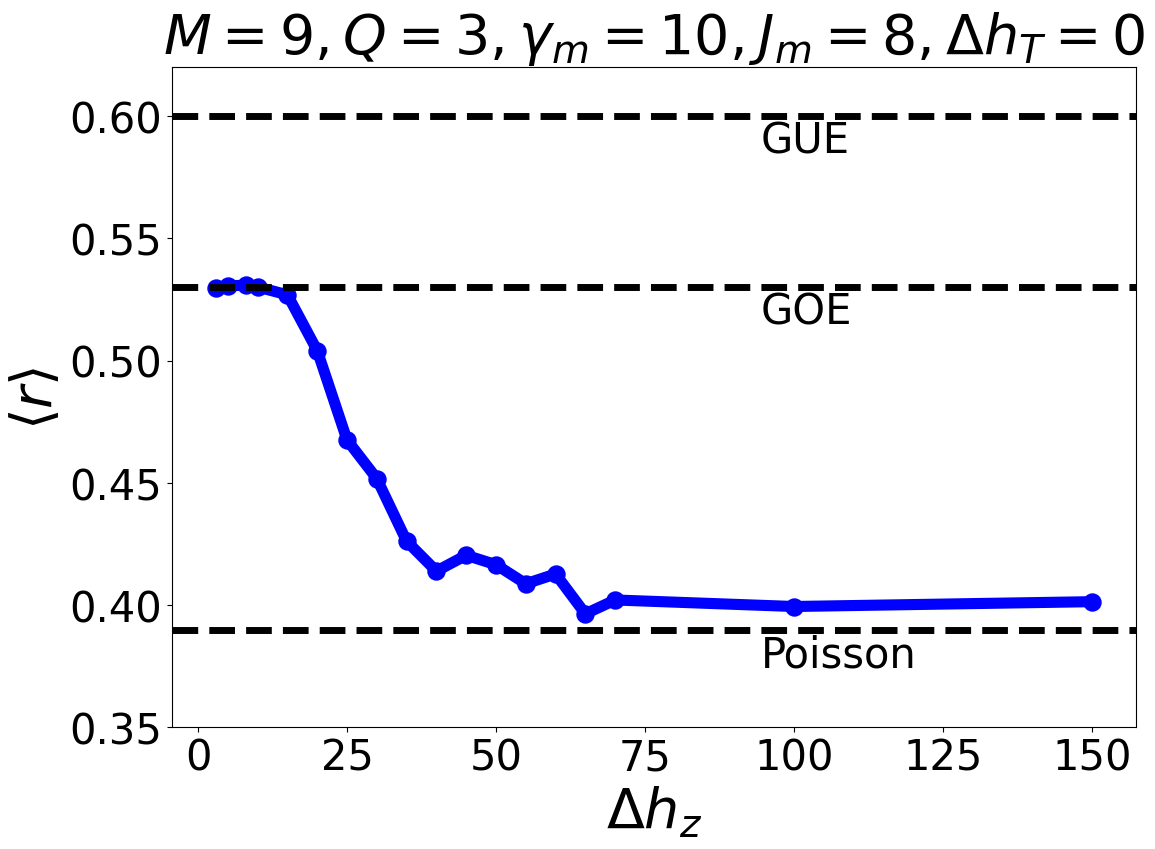}}
\subfigure[]{\label{fig: exp_M9_Q3_t10_J8_hzhT}\includegraphics[width=45mm]{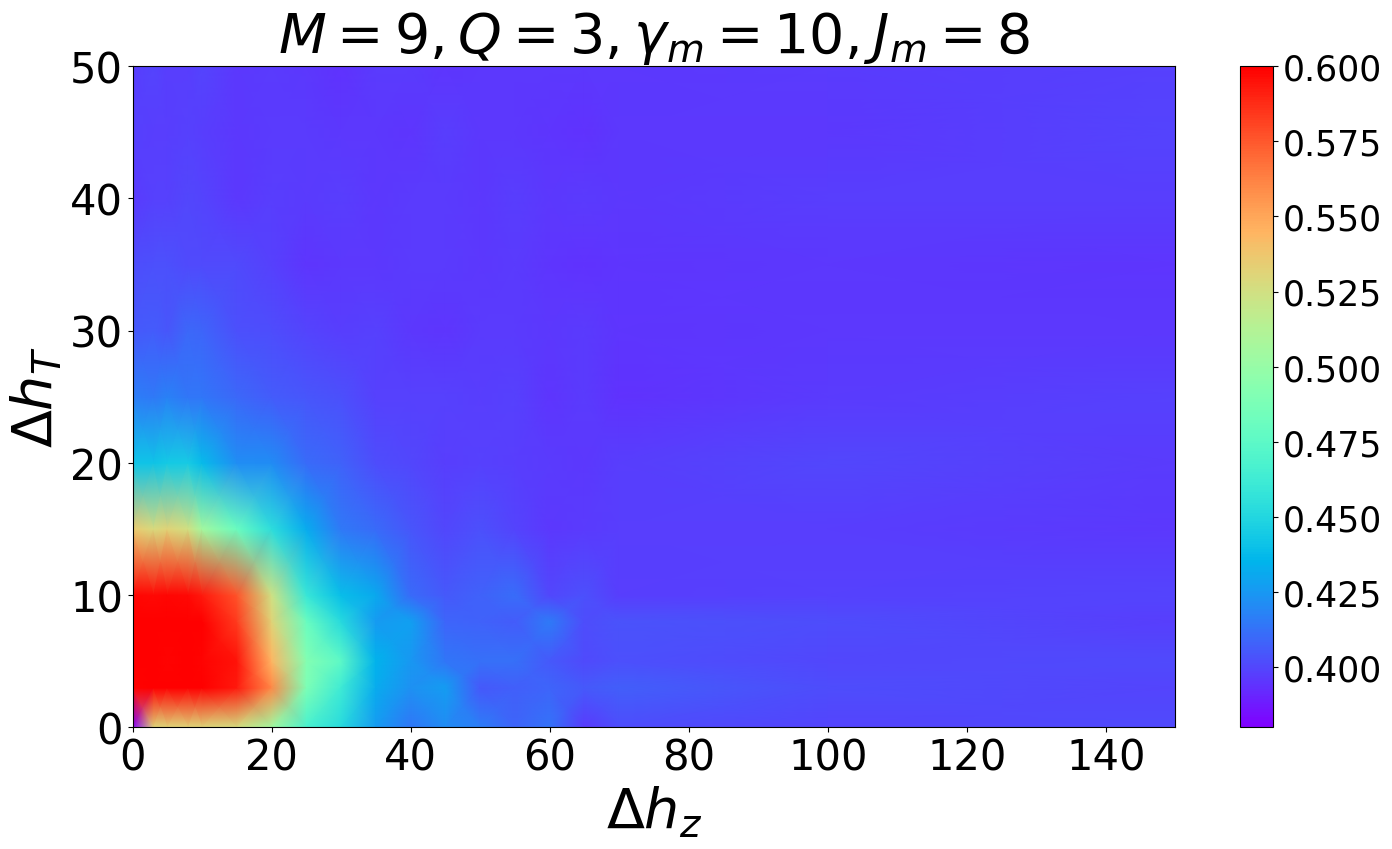}}
\caption{The level spacing statistics $r$-ratio for a (a-b) $(M, Q)=(12, 1)$ model, (c-d) $(M, Q)=(10, 2)$ model, and (e-f) a $(M, Q)=(9, 3)$ model with uniform $\gamma_m=10$, $J_m=8$, and a random magnetic field $\vec h = (h^x, h^y, h^z)$, where $\Delta h^x = \Delta h^y = \Delta h_T$. Each data point is averaged over 15 realizations.}
\label{fig: r-ratio of hzhT}
\end{figure}

In \figref{fig: r-ratio of hzhT}, we show the mean level spacing ratio for models with random magnetic fields $\vec{h} = (h^x, h^y, h^z)$, whose components are sampled from normal distributions of mean $0$ and standard deviation $\Delta h^T$ for $h^x$ and $h^y$ and $\Delta h^z$ for $h^z$. For $\Delta h_T = 0$, the values are shown in \figref{fig: exp_M12_Q1_t10_J8_hzR}, \figref{fig: exp_M10_Q2_t10_J8_hzR}, and \figref{fig: exp_M9_Q3_t10_J8_hzR}. As shown in these figures, we numerically observe a crossover from $\langle r \rangle \sim 0.53$ to $\langle r \rangle \sim 0.39$ as $\Delta h^z$ is increased i.e., the level spacing statistics smoothly crossover from a GOE (chaotic) distribution to a Poisson (non-ergodic) distribution. This crossover persists even as $\Delta h_T$ is increased (see \figref{fig: exp_M12_Q1_t10_J8_hzhT}, \figref{fig: exp_M10_Q2_t10_J8_hzhT}, and \figref{fig: exp_M9_Q3_t10_J8_hzhT}). However, in contrast to the $\Delta h_T=0$ case, the $r$-ratio now changes from $\langle r \rangle \sim 0.60$ to $\langle r \rangle \sim 0.39$, signalling a crossover from a GUE distribution to a Poisson distribution when the random magnetic field is not along the $z$-axis. For large enough $\Delta h_T$, there is no longer a crossover as a function of $h^z$ and $\langle r \rangle \sim 0.39$ regardless of $\Delta h^z$, indicating that the system is always in a localized phase.

Note that if we choose $h^y=0$ while keeping $h^x \neq 0$ and $h^z \neq 0$, the level spacing statistics will exhibits a GOE distribution rather than GUE. This is because in this case, the system possesses an additional antiunitary symmetry (see, e.g., Ref. \cite{PhysRevB.93.174202}): time reversal plus a $\pi$ rotation symmetry about the $y-$axis, i.e., 
\ie
(\sigma^x, \sigma^y, \sigma^z) \rightarrow (\sigma^x, -\sigma^y, \sigma^z)
.\fe
A similar situation occurs when we choose $h^x=0$, $h^y \neq 0$, and $h^z \neq 0$, which will result in a GOE distribution as well.

Finally, we note that at $\Delta h^z \rightarrow 0$ and $\Delta h^T \rightarrow 0$, there is a small region that is neither Poisson distributed nor WD distributed. This regime corresponds to the zero magnetic field case where we have established the presence of HSF in prior sections; we do not expect conventional level statistics here due to the exponentially many fragmented sectors~\cite{moudgalya2019thermalization}.

\subsection{$H_{\bar{\gamma} \bar{J} \mu h}$: disorder induced MBL}
\label{subsection: MBL from disorder W}


\begin{figure}[t]
\centering
\subfigure[]{\label{fig: exp_M12_Q1_t10_hz10_JW}\includegraphics[width=70mm]{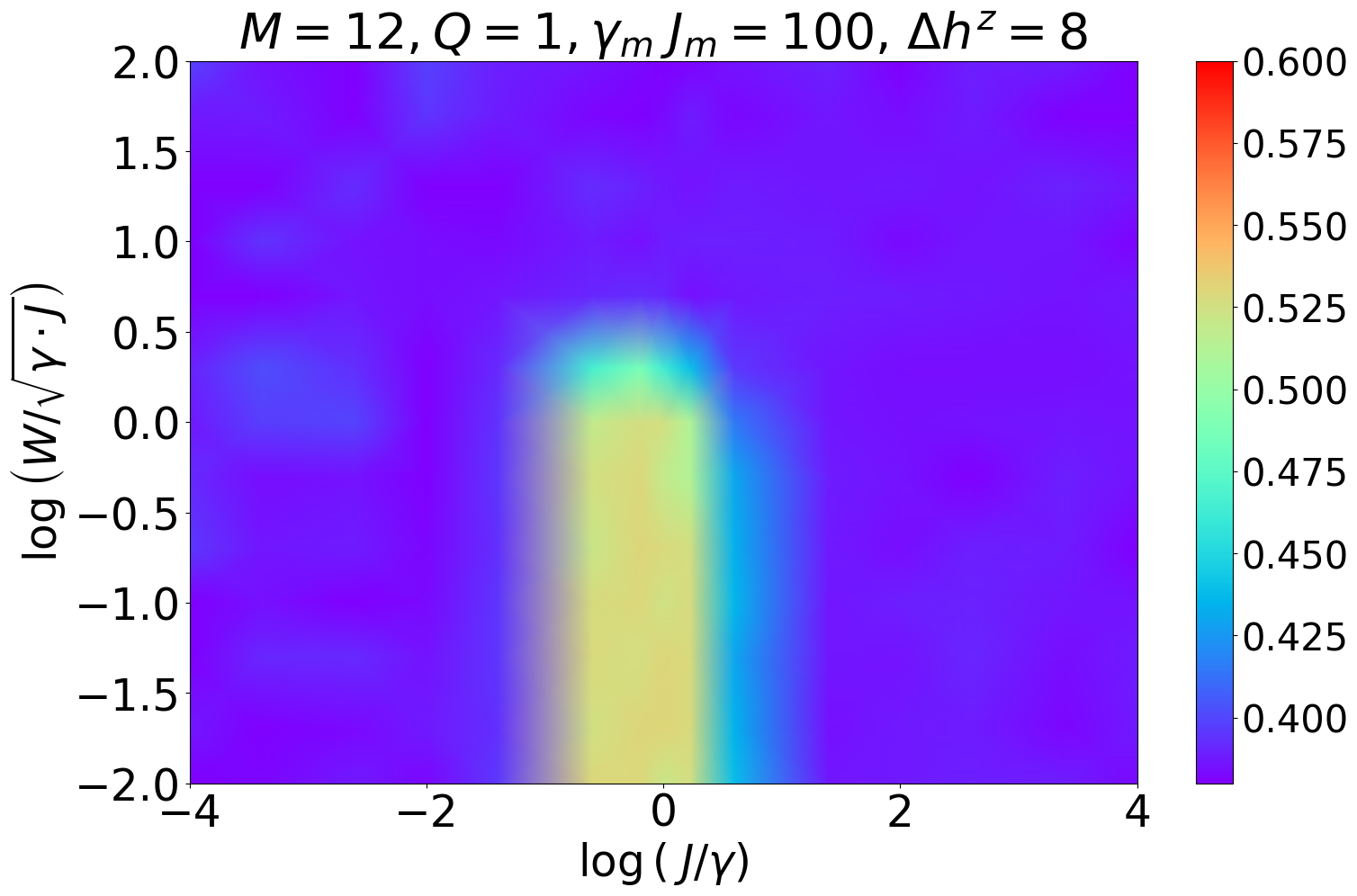}}
\subfigure[]{\label{fig: exp_M10_Q2_t10_hz10_JW}\includegraphics[width=70mm]{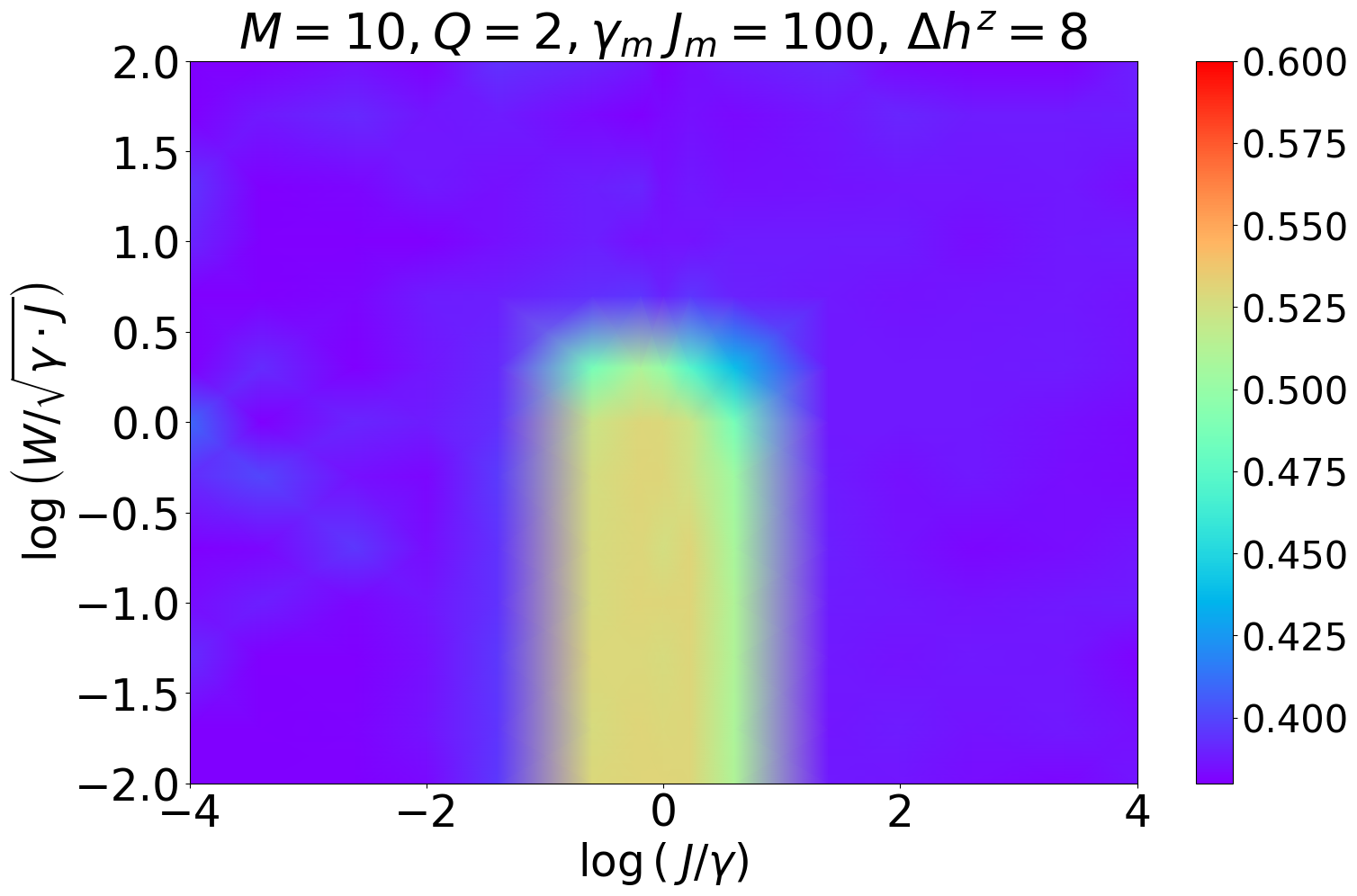}}
\subfigure[]{\label{fig: exp_M9_Q3_t10_hz10_JW}\includegraphics[width=70mm]{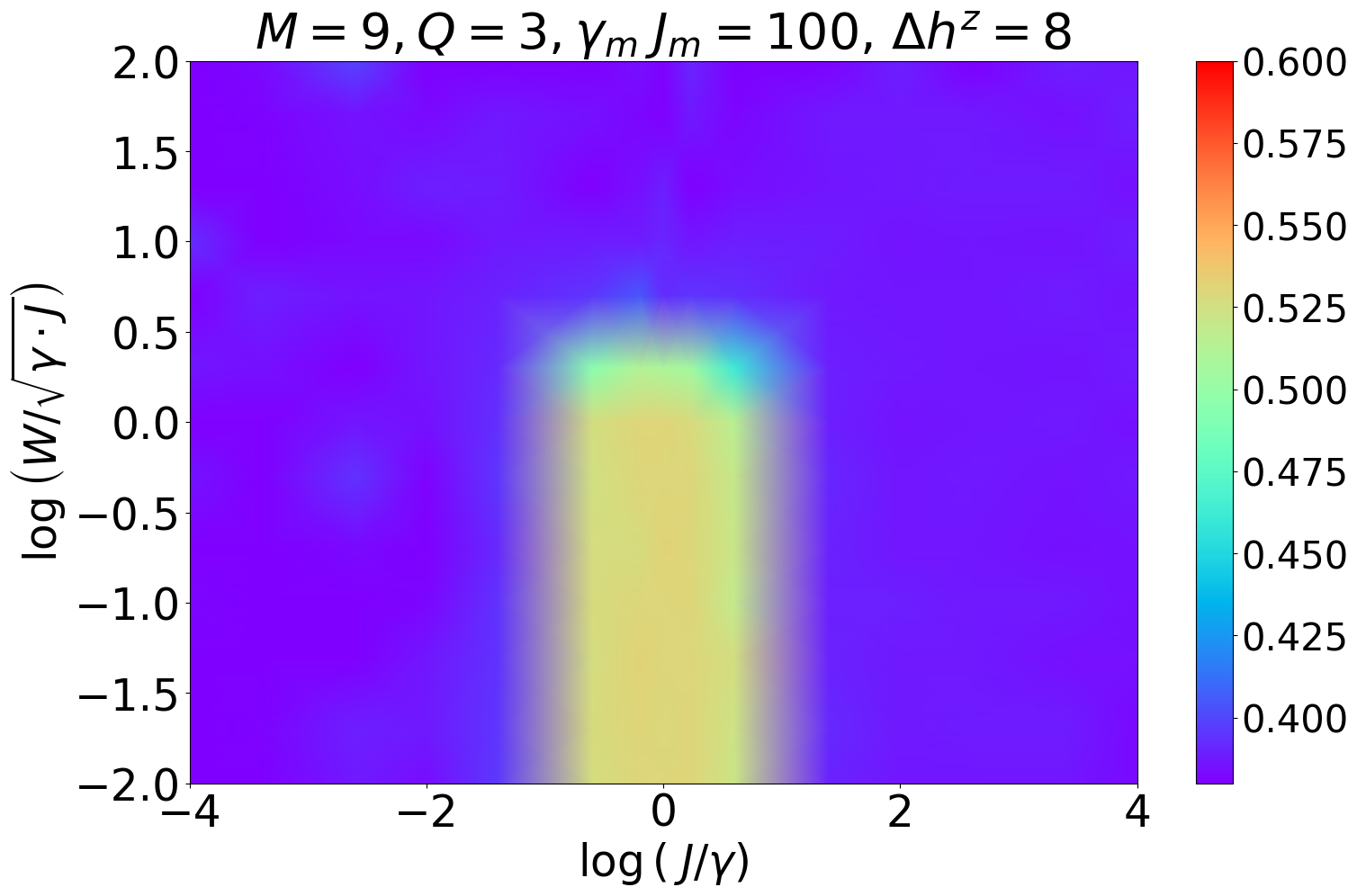}}
\caption{The level spacing statistics $r$-ratio for models with uniform $\gamma_m \cdot J_m=100$ and random magnetic field $\Delta h^z=8$. Each data point is averaged over 15 realizations.}
\label{fig: r-ratio of JW}
\end{figure}

Besides the random magnetic field on the spins, we can also induce localization in the system through static disorder. Let us first consider the model $H_{\bar{\gamma} \bar{J} \mu h} = H_{\bar{\gamma}} + H_{\bar{J}} + H_\mu + H_h$ with random static disorder $\mu_m$ and a non-zero random magnetic field in the $z$-direction (i.e., $h_x = h_y = 0$). This model has two limits: the large $\gamma/J$ limit and the large $J/\gamma$ limit. In the large $\gamma/J$ limit, the model approaches two decoupled chains, one comprised of fermions and the other of spin-$1/2$s. Due to the random on-site potential $H_\mu$, the fermionic chain will exhibit conventional Anderson localization. Concurrently, the spin-chain remains non-interacting in this limit and since both chains are non-thermal, so is the entire system. On the other hand, in the large $J/\gamma$ limit the existence of blockades prevents the system from thermalizing (see Sec.~\ref{section: disconnected subspaces by blocking} for details). Thus, we expect that the system will display non-thermal behavior in either of these limits. 

However, in the intermediate regime $\gamma \sim J$, we find that the model exhibits chaotic behavior. In \figref{fig: r-ratio of JW}, we show the LSS for models with uniform $\gamma_m = \gamma$ and $J_m=J$, and a random $z$-component magnetic field (we hold the product $\gamma\times J$ fixed). As anticipated, the observed $r$-ratio shows Poisson distributed LSS in the limits $J/\gamma \rightarrow 0$ and $J/\gamma \rightarrow \infty$, since $\langle r \rangle \sim 0.39$. At intermediate $J/\gamma$, however, the $r$-ratio indicates that the LSS is GOE distributed, since $\langle r \rangle \sim 0.53$. As a function of $J/\gamma$, the system thus shows re-entrant behavior as it first crosses over from a localized state to a chaotic state, followed by another crossover into a localized state. When the disorder $W$ is the largest scale, the $r$-ratio indicates that the LSS is Poisson distributed regardless of the interaction strength $J_m$.

\begin{figure}[t]
\centering
\subfigure[]{\label{fig: exp_M12_Q1_t10_J10_W}\includegraphics[width=75mm]{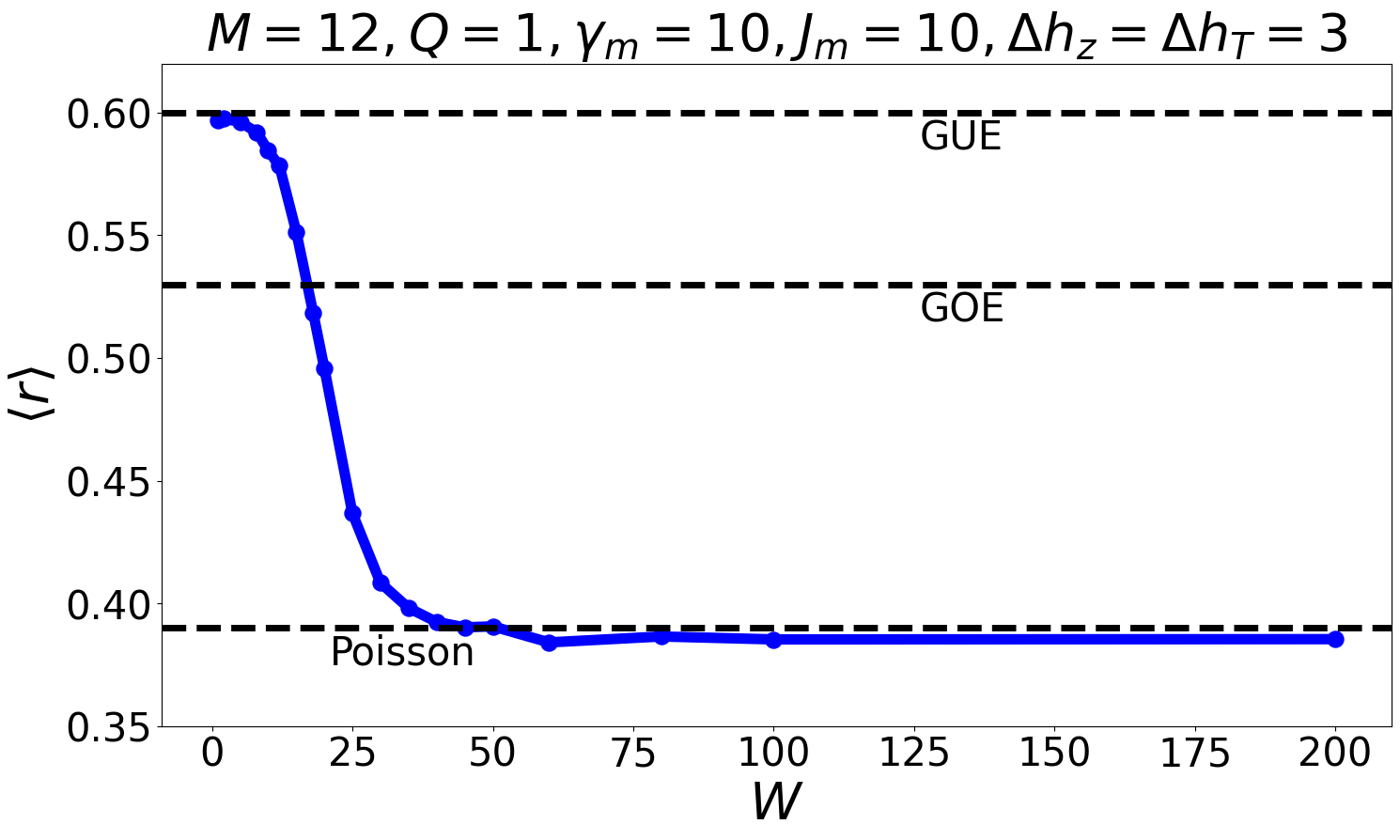}}
\subfigure[]{\label{fig: exp_M10_Q2_t10_J10_W}\includegraphics[width=75mm]{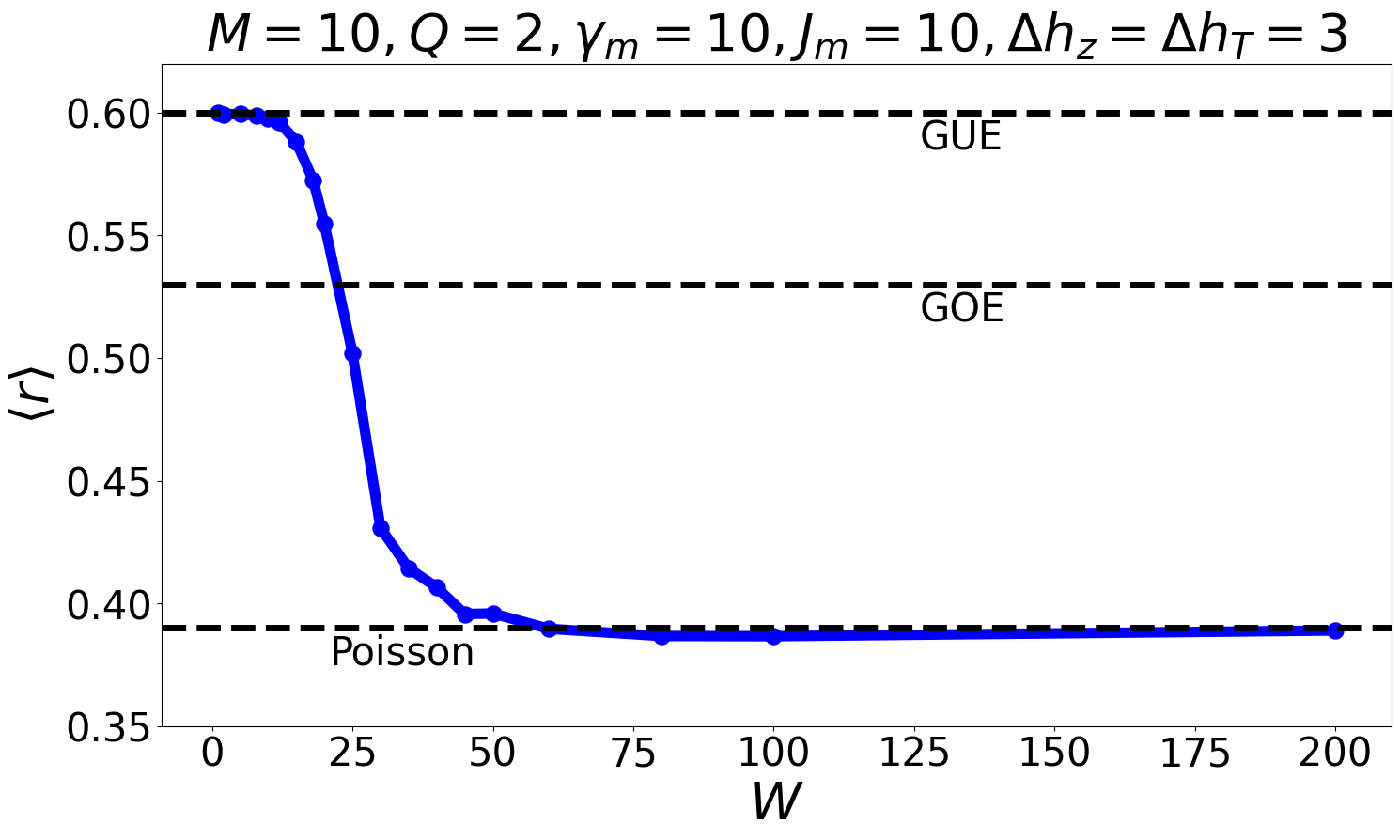}}
\subfigure[]{\label{fig: exp_M9_Q3_t10_J10_W}\includegraphics[width=75mm]{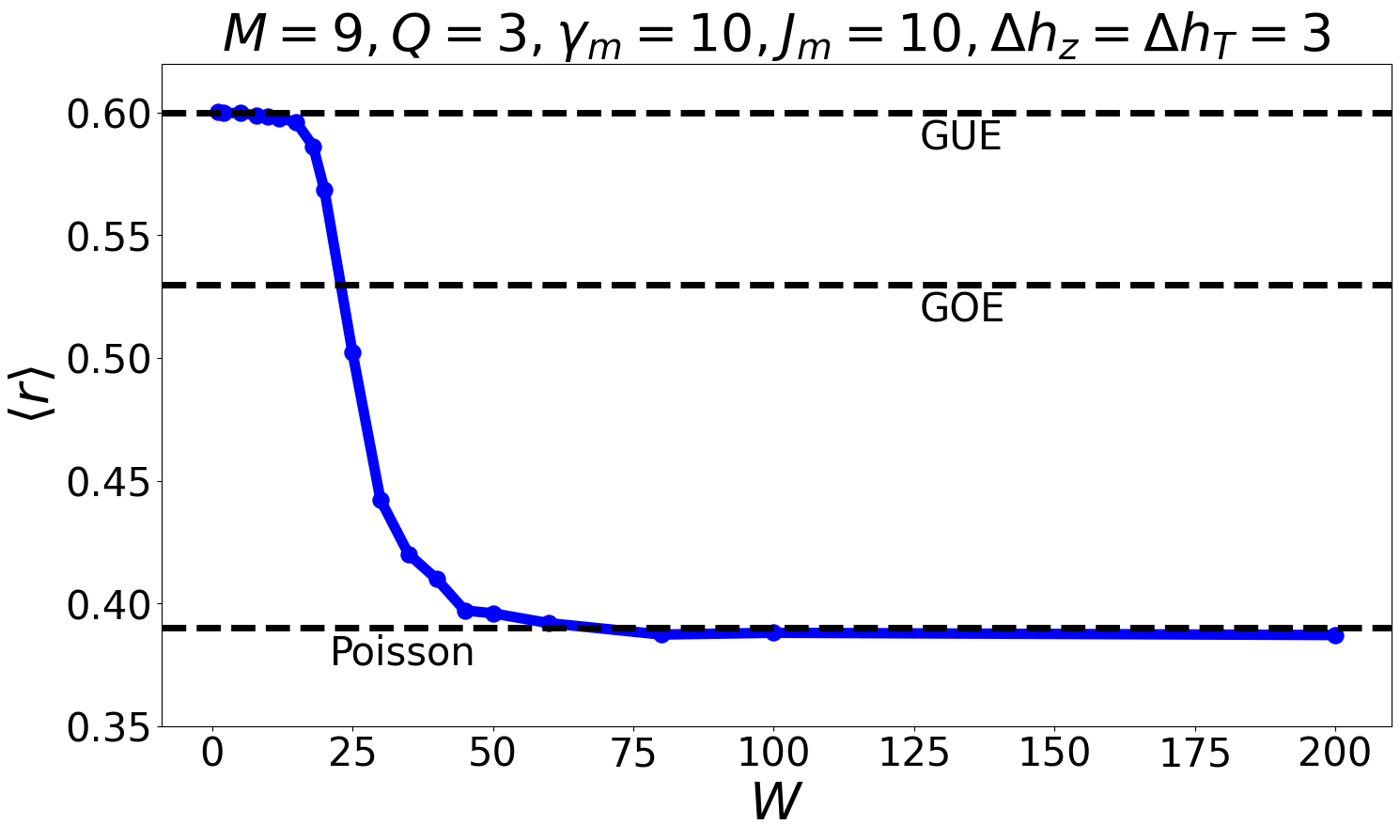}}
\caption{The level spacing statistics $r$-ratio for models with uniform $\gamma_m = \cdot J_m = 10$ and random magnetic field $\Delta h^z= \Delta h^T= 3$ for (a) $(M,Q)=(12,1)$, (b) $(M,Q)=(10,2)$, and (c) $(M,Q)=(9,3)$. Each data point is averaged over 20 realizations.}
\label{fig: r-ratio of t10_J10_hz3_hT3_W}
\end{figure}

When $h_x, h_y \ne 0$, there are no longer any blockades in the large $J/\gamma$ limit, as discussed in Sec.~\ref{section: disconnected subspaces by blocking}. In addition, from \figref{fig: exp_M12_Q1_t10_J8_hzhT}, \figref{fig: exp_M10_Q2_t10_J8_hzhT}, and \figref{fig: exp_M9_Q3_t10_J8_hzhT}, we see that the model demonstrates chaotic behavior within a specific region of the parameter space spanned by $\Delta h_z$ and $\Delta h_T$ if the disorder potential $W=0$. A natural question is whether the model will exhibit MBL as $W$ is increased. Our numerics suggest the answer is affirmative: in \figref{fig: r-ratio of t10_J10_hz3_hT3_W}, we show the $r$-ratio of models with $\Delta h^z = \Delta h^T = 3$ for various disorder strengths $W$. While the $r$-ratio indicates GUE distributed LSS at small $W$, we clearly observe a Poisson distribution at large $W$, consistent with MBL.

\subsection{Entanglement Entropy: MBL}

In this subsection, we present additional evidence on entanglement entropy to show that, at sufficiently large random on-site potential, the system exhibits MBL behavior.

The growth of the EE in MBL systems is different from that in chaotic systems. As pointed out in Refs.~\cite{rahul2015review,znidaric2008mbl}, the spreading of EE from a non-entangled initial condition is logarithmic for MBL systems, whereas in chaotic systems, the spreading follows a power-law. \figref{fig: expEE MBL} shows the EE before reaching their maximum values. The blue curves represent systems with no random on-site potential $W=0$, and the EE growth follows a power law: $S_A(t) \propto t$. For large random on-site potentials $W=60$, the $r$-ratio indicates a Poisson distribution of the LSS, see Section.~\ref{subsection: MBL from disorder W}. The green curves represents the EE growth of those systems, and it shows a logarithmic behavior. The EE provides another evidence that large random on-site potentials make the system becomes localized.

\begin{figure}[t]
\centering
\subfigure[]{\label{fig: expEE_M12_Q1_t10_J10_hz8}\includegraphics[width=42mm]{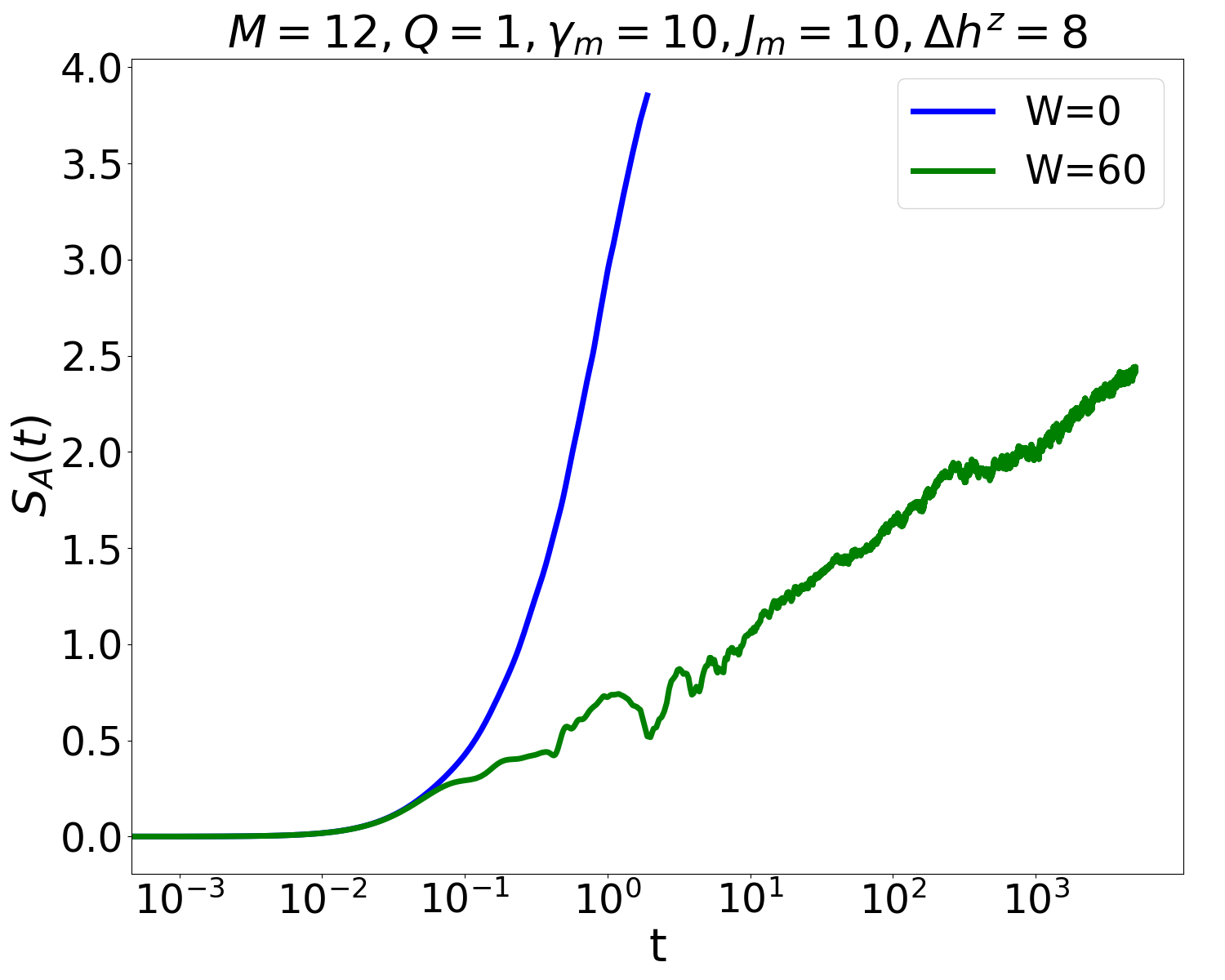}}
\subfigure[]{\label{fig: expEE_M12_Q1_t10_J10_hz3_hT3}\includegraphics[width=42mm]{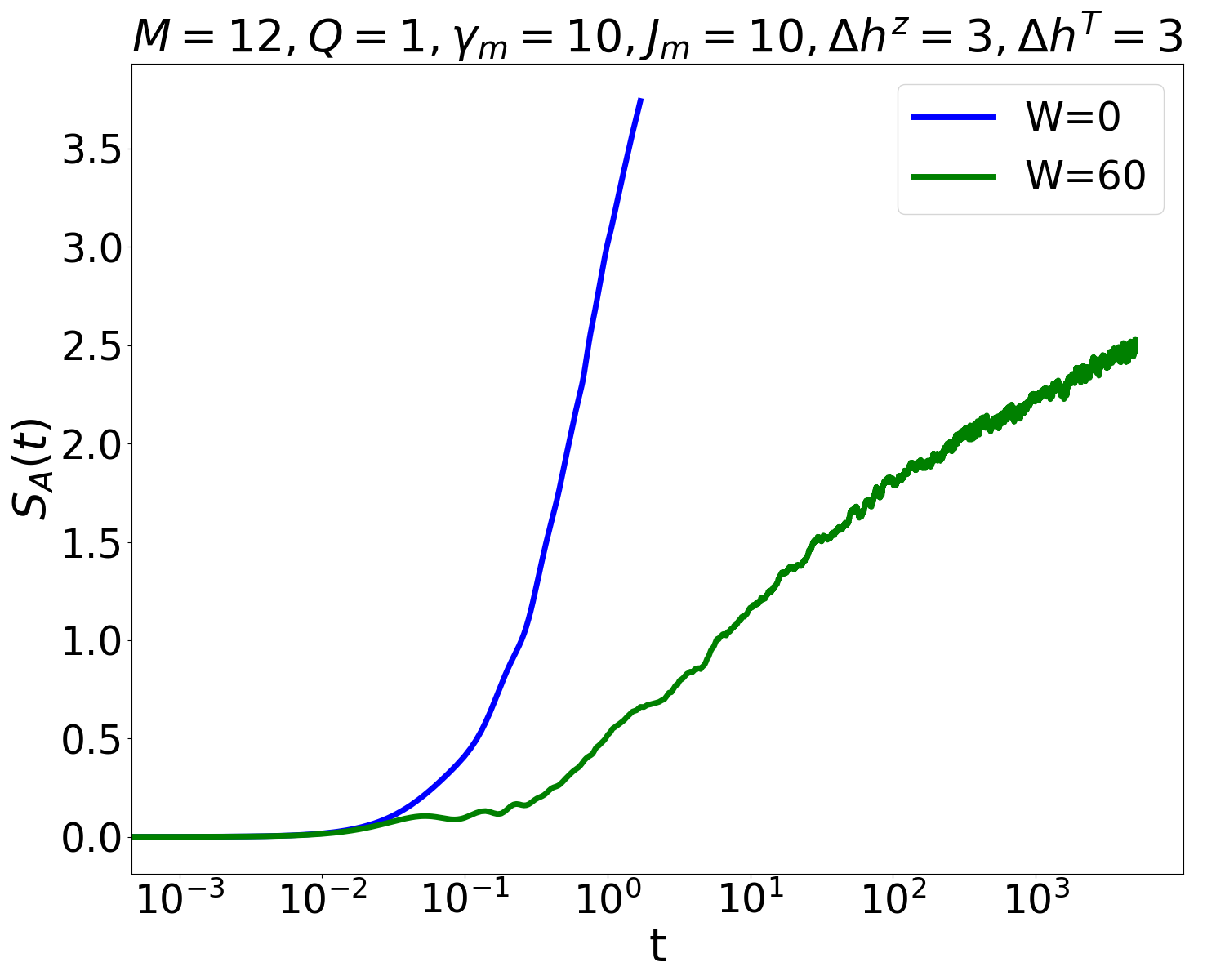}}
\subfigure[]{\label{fig: expEE_M10_Q2_t10_J10_hz8}\includegraphics[width=42mm]{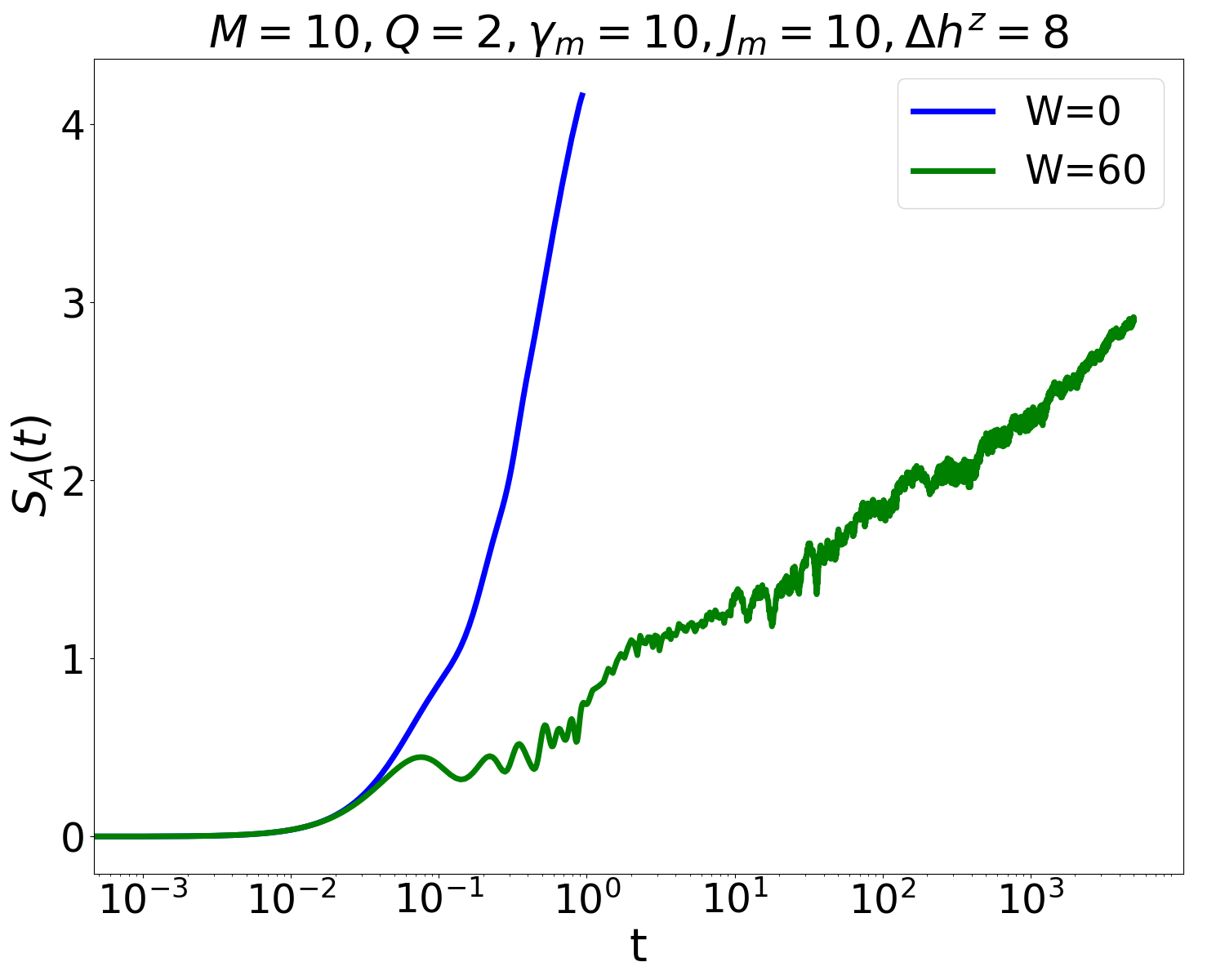}}
\subfigure[]{\label{fig: expEE_M10_Q2_t10_J10_hz3_hT3}\includegraphics[width=42mm]{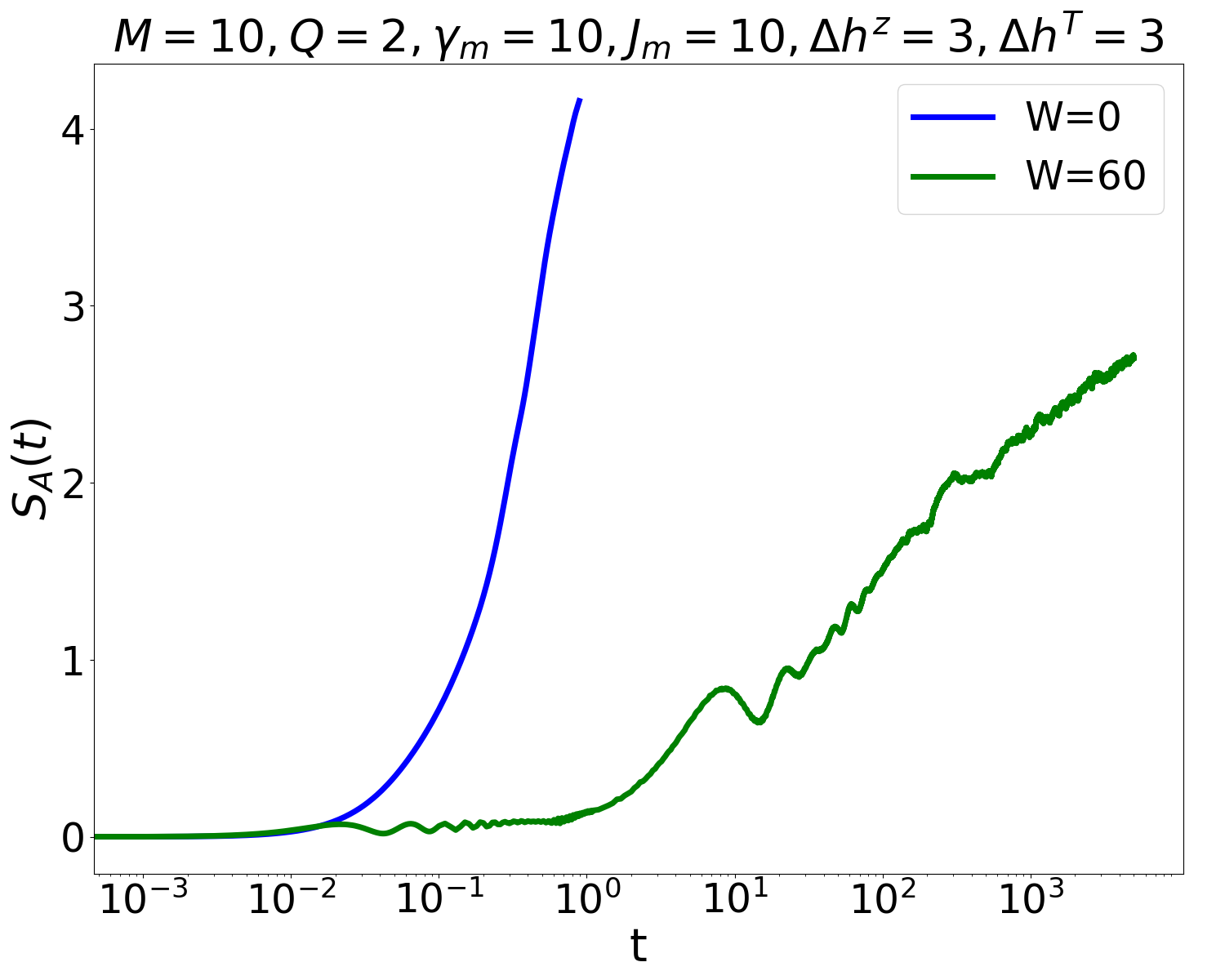}}
\subfigure[]{\label{fig: expEE_M9_Q3_t10_J10_hz8}\includegraphics[width=42mm]{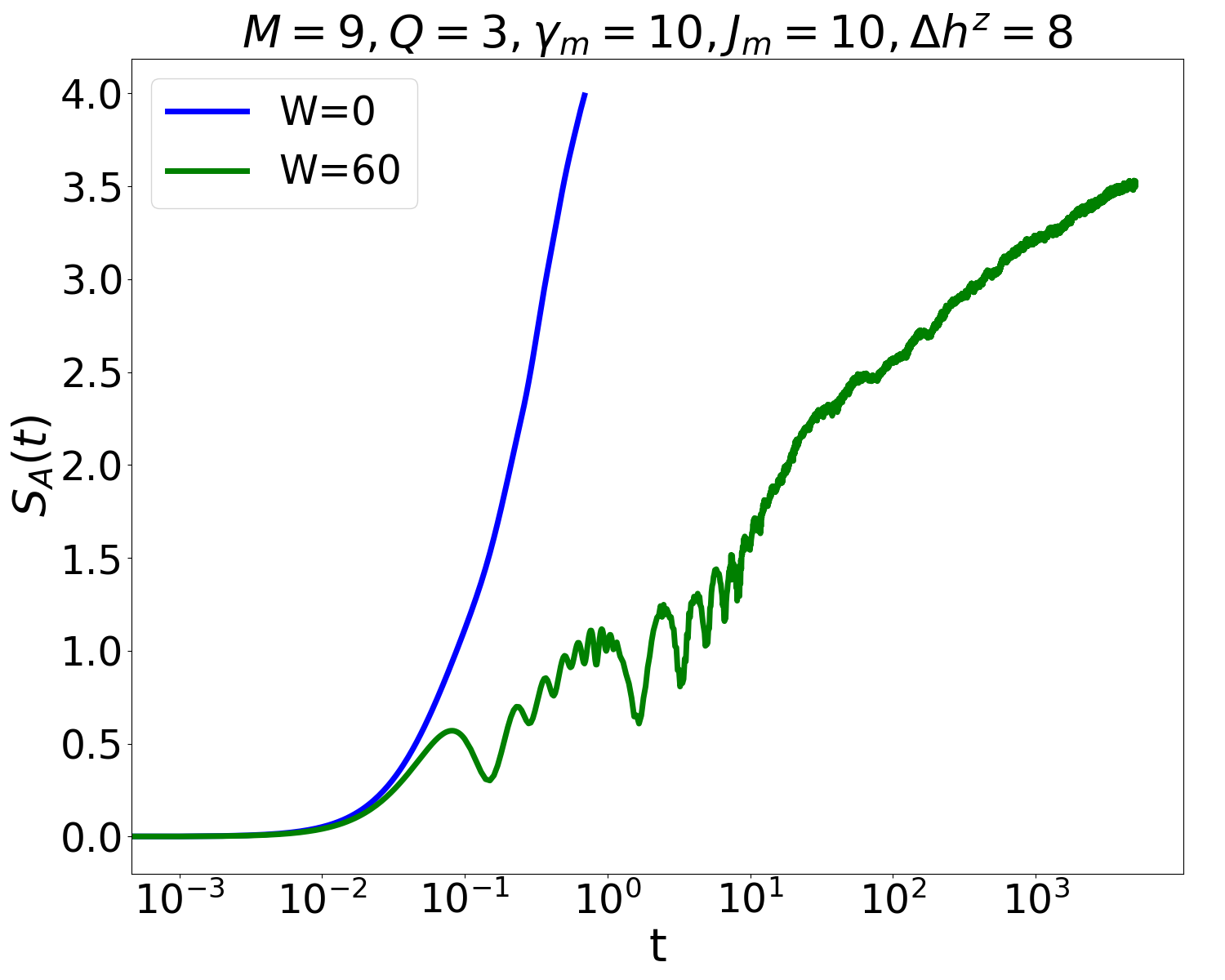}}
\subfigure[]{\label{fig: expEE_M9_Q3_t10_J10_hz3_hT3}\includegraphics[width=42mm]{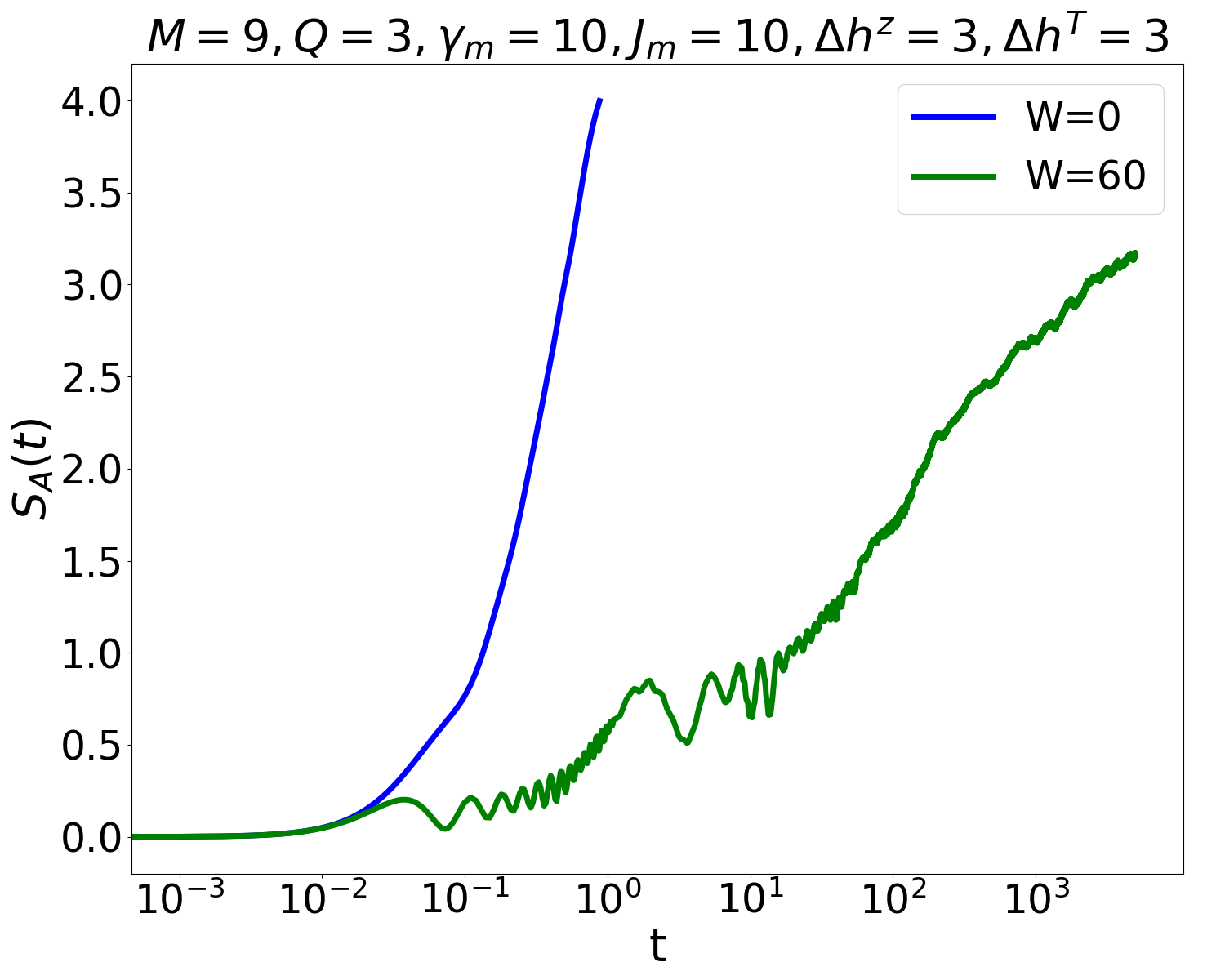}}
\caption{The entanglement entropy growth in systems featuring random on-site potentials with a standard deviation of $W=60$ (in green). As a reference, we also provide the EE growth in the absence of random on-site potentials, i.e., $W=0$ (in blue). The titles indicate the system parameters with the same interaction and hoppings: $\gamma_m=J_m=10$. The values for $(M, Q, \Delta h^z, \Delta h^T)$ are as follows: (a) $(12, 1, 8, 0)$, (b) $(12, 1, 3, 3)$, (c) $(10, 2, 8, 1)$, (d) $(10, 2, 3, 3)$, (e) $(9, 3, 8, 0)$, and (f) $(9, 3, 3, 3)$.}
\label{fig: expEE MBL}
\end{figure}

\section{Summary and Future Directions}
\label{section: Discussion}

In this paper, we have introduced a 1D extended quantum breakdown model that features a spatially asymmetric interaction between fermionic and spin-$1/2$ degrees of freedom and is descended from the recently introduced breakdown model~\cite{Lian:2022nqj}. Through a combination of exact analytic calculations and numerical simulations, we have shown that this model exhibits a host of intriguing dynamical phenomena, prominent amongst which is the Hilbert space fragmentation caused by the presence of exponentially many closed Krylov subspaces within most symmetry sectors. A key aspect of this model is that the fragmentation occurs in a basis of entangled states and thus provides another example of a ``quantum fragmented" system. Moreover, it exhibits extensive numbers of degenerate Krylov subspaces. The lack of thermalization is also revealed by considering the long-time entanglement dynamics of random initial Haar product states, which fail to saturate the expected Page value when the given symmetry sector is strongly fragmented, even if the initial state has a weight over \emph{all} Krylov subspaces. The saturation of the Page entropy is restored for symmetry sectors exhibiting weak fragmentation.

Within the $Q = 1$ sector, we also showed analytically that the fragmentation survives the addition of an on-site random potential (for uniform couplings $\gamma$ and $J$) and exponentially many Krylov subspaces persist.

When we apply a non-zero magnetic field (along the $z$-direction) to the spin degrees of freedom, the Krylov subspaces merge and the system thermalizes. We numerically observe a crossover from a dynamical spin configuration to one where the spins are frozen as the field strength is increased, and provide a perturbative argument that well-describes the observed crossover. For a random magnetic field, we numerically observe a transition from a chaotic to an MBL phase within certain sectors, which we probe through level spacing statistics. Finally, we study the effects of static disorder on the system and again find an MBL phase using level spacing statistics and long-time dynamics of the entanglement entropy.

The extended quantum breakdown model exhibits rich dynamical behavior and hence provides a playground for exploring transitions between distinct dynamical phases. Here, we have initiated a study of these phases but several open questions remain: firstly, the physics of larger Krylov subspaces remains unclear and understanding these better would explain when the fragmentation is strong rather than weak. Further, in the presence of either a random magnetic field or static disorder, it remains to be understood how the critical field strength or disorder strength scales with system size at the MBL transition. For instance, in a different model with Hilbert space fragmentation, Ref.~\cite{herviouMBL2021} observed a size-independent scaling of the critical disorder strength in certain sectors --- further investigations of the extended breakdown model could shed light on whether this behavior is generic for strongly fragmented systems.

Finally, it would be interesting to establish a precise relation between the extended breakdown model and so-called ``quantum link" models~\cite{Chandrasekharan:1996ih}. While the former hosts fermions interacting asymmetrically with spins, the latter consists of fermions interacting symmetrically with dynamical gauge degrees of freedom and has recently enabled the simulation of 1D quantum electrodynamics on Rydberg atom chains~\cite{surace2020qlm}. Understanding whether there exist limiting conditions under which the quantum link model reduces to the extended breakdown model would thus also provide a natural route towards the latter's experimental realization.

\section*{Acknowledgements}
We are especially grateful to Yu-Min Hu for pointing out a relation between our model and the quantum link model. A.P. and N. R. acknowledge B. A. Bernevig and S. Moudgalya for previous collaboration on related topics. This material is based upon work supported by the Sivian Fund at the Institute for Advanced Study and the U.S. Department of Energy, Office of Science, Office of High Energy Physics under Award Number DE-SC0009988 (A. P.). A. P. thanks the KITP, which is supported by the National Science Foundation under Grant No. NSF PHY-1748958, for its hospitality during the ``Topology, Symmetry, and Interactions in Crystals" program, when part of this work was completed. B.L. acknowledge the Alfred P. Sloan Foundation, the National Science Foundation through Princeton University’s Materials Research Science and Engineering Center DMR-2011750, and the National Science Foundation under award DMR-2141966. Additional support is provided by the Gordon and Betty Moore Foundation through Grant GBMF8685 towards the Princeton theory program.



\bibliography{mybib}


\appendix

\section{Relation to the original quantum breakdown model}
\label{appendix: the original breakdown model}
A quantum breakdown model for fermions with asymmetric interaction was introduced in Ref.~\cite{Lian:2022nqj}, which consists of a 1D chain of $M$ sites, and each site consists of $N$ fermionic degrees of freedom. The Hamiltonian takes the form
\ie
\label{eqn: the hamiltonian of the original model}
H = H_J + H_\mu + H_{\text{dis}}
.\fe

The first part is the asymmetric breakdown interaction:
\ie
\label{eqn: HI of the original model}
H_J = 
\sum_{m=1}^{M-1} \sum_{i<j<k}^{N} \sum_{l=1}^{N} 
\left( J_{m, l}^{ijk} a_{m+1, i}^{\dagger} a_{m+1, j}^{\dagger} a_{m+1, k}^{\dagger} a_{m, l} 
+ h.c. \right)
,\fe
where $a_{m,i}$ and $a_{m,i}^{\dagger}$ are the annihilation and creation operators of the $i$-th fermion mode on the $m$-th site ($1 \le i \le N$, $1 \le m \le M$), and $J_{m, l}^{ijk}$ are complex parameters which are totally anti-symmetric in indices $i$, $j$, $k$.

The second part of the Hamiltonian is an on-site potential uniform within each site:
\ie
\label{eqn: Hmu of the original model}
H_\mu = \sum_{m=1}^{M} \mu_m \hat{n}_m
,\fe
where $\hat{n}_m = \sum_{i=1}^{N} a_{m, i}^{\dagger} a_{m, i}$ is the fermion number operator on the $m$-th site, and $\mu_m$ is the on-site potential on the $m$-th site.

The last part is a disorder potential:
\ie
\label{eqn: Hdis of the original model}
H_{\text{dis}} = \sum_{m=1}^{M} \sum_{i=1}^{N} \nu_{m, i} a_{m, i}^{\dagger} a_{m, i}
,\fe
where $\nu_{m,i}$ are Gaussian random potentials with mean value and variance given by
\ie
\langle \nu_{m, i} \rangle = 0 \;, \;\;\; \langle \nu_{m, i}^2 \rangle = W^2
.\fe

The quantum breakdown model in Eq.~\eqref{eqn: the hamiltonian of the original model} has a conserved global $U(1)$ charge
\ie
\label{eqn: Q in the original model}
q = \sum_{m=1}^{M} 3^{M-m} \hat{n}_m
.\fe
The conserved charge allows us to block diagonalize the Hamiltonian, reducing the computational complexity significantly.

For models with $N=3$, if the disorder potential is absent, the Hamiltonian $H=H_\mu + H_I$ can be simplified by $U(3)$ rotations on $a_{m, i}$ on each site. We define a new fermion basis $c_{m, i}$ by
\ie
a_{m, l} = \sum_{l^\prime=1}^{3} U_{l l^\prime}^{(m)} c_{m, l^\prime}
,\fe
where the matrices $U^{(m)}$ satisfy
\ie
\sum_{l} J_{m, l}^{123} U_{l l^\prime}^{(m)} = \sqrt{3} J_m \delta_{1, l^\prime} \det\left( U^{(m+1)} \right)
,\fe
with
\ie
J_m \equiv \sqrt{\frac{1}{3} \sum_{l=1}^{3} \left| J_{m, l}^{123} \right|^2 }
.\fe
In the new basis, the Hamiltonian can be rewritten as
\ie
\label{eqn: original breakdown model simplified}
& H_J = \sum_{m=1}^{M-1}{\sqrt{3} J_m c_{m+1, 1}^\dagger c_{m+1, 2}^\dagger c_{m+1, 3}^\dagger c_{m, 1} + h.c.}
\\
& H_\mu = \sum_{m=1}^{M}{\sum_{l=1}^{3} \mu_m c_{m, l}^\dagger c_{m, l}}
.\fe

The extended breakdown model originates from the breakdown model with $N=3$. We allow electrons to hop between nearest neighbors, which renders Eq.~\eqref{eqn: Q in the original model} a non-conserved charge. The Hilbert space dimension therefore cannot be divided into different conserved charge sectors, c.f., Eq.~\eqref{eqn: Q in the original model}. If we wish to calculate the energy spectrum, we will have to diagonalize the Hamiltonian matrix with dimension $2^{3M} \times 2^{3M}$. The dimension grows rapidly as the system size $M$ increases, which makes exact diagonalization numerically impractical for large system sizes.

In Eq.~\eqref{eqn: original breakdown model simplified}, the fermion modes of $i=2,3$ are always simultaneously excited. Therefore, the states where either all $i=2,3$ fermion modes are fully filled or completely empty on each site form closed subspaces. In these subspaces, we represent the fully-filled/empty states on site $m$ by spin eigenstates of the Pauli matrix $\sigma_m^z$, i.e.,
\ie
\begin{cases}
c_{m, 1}^\dagger c_{m, 2}^\dagger c_{m, 3}^\dagger |\Omega\rangle \mapsto c_{m, 1}^\dagger |\Omega\rangle \otimes | \uparrow \rangle_m \equiv c_m^\dagger |\Omega\rangle \otimes | \uparrow \rangle_m
\\
c_{m, 1}^\dagger |\Omega\rangle \mapsto c_{m, 1}^\dagger |\Omega\rangle \otimes | \downarrow \rangle_m \equiv c_m^\dagger |\Omega\rangle \otimes | \downarrow \rangle_m
\end{cases}
,\fe
where $|\Omega\rangle$ is the vacuum state of the fermion chain (no fermion), and we re-defined $c_{m, 1}^\dagger$ as $c_m^\dagger$. These subspaces constitute the primary focus of this paper. In this paper, we denote the spin state on the $m$-th site as $s_m = \{ \ket{\uparrow}_m, \ket{\downarrow}_m \}$ (or $s_m = \{ | 1 \rangle_m, | 0 \rangle_m \}$).

\section{Explicit form of spin states in Sec.~\ref{section: Notations}}
\label{appendix: Explicit forms of the notations}

\begin{widetext}

\subsection{Notations for $H_{\gamma J}$: $\mathcal{K}^{(1)}$}
\label{appendix: Explicit forms of the notations for K1}
The explicit form of the spin states with arrows defined in \ref{subsubsection: Notations for K1} are:
\begin{alignat}{2}
\begin{tikzpicture}[baseline={(0, -0.1)}]
\draw (0.4, 0) circle (.2);
\draw (1.2, 0) circle (.2);
\draw (0.6, 0) -- (1.0, 0);
\node at (0, 0) {$\cdots$};
\node at (1.7, 0) {$\cdots$};
\node at (0.8, -0.2) {\small $0$};
\node at (1.2, -0.4) {\small $m$};
\draw[<-] (0.65, 0.15) -- (0.95, 0.15);
\end{tikzpicture}
\equiv
&\left(-\frac{1}{\gamma_{m-1}}\right)\times
&&\begin{tikzpicture}[baseline={(0, -0.1)}]
\draw (0.4, 0) circle (.2);
\draw (1.2, 0) circle (.2);
\draw (0.6, 0) -- (1.0, 0);
\node at (0, 0) {$\cdots$};
\node at (1.7, 0) {$\cdots$};
\node at (0.8, -0.2) {\small $0$};
\node at (1.2, -0.4) {\small $m$};
\end{tikzpicture}
+
\left(-\frac{J_{m-1}}{\gamma_{m-1}^2}\right)\times
\begin{tikzpicture}[baseline={(0, -0.1)}]
\draw (0.4, 0) circle (.2);
\draw (1.2, 0) circle (.2);
\draw (0.6, 0) -- (1.0, 0);
\node at (0, 0) {$\cdots$};
\node at (1.7, 0) {$\cdots$};
\node at (0.8, -0.2) {\small $1$};
\node at (1.2, -0.4) {\small $m$};
\end{tikzpicture}
\notag\\
\begin{tikzpicture}[baseline={(0, -0.1)}]
\draw (0.4, 0) circle (.2);
\draw (1.2, 0) circle (.2);
\draw (0.6, 0) -- (1.0, 0);
\node at (0, 0) {$\cdots$};
\node at (1.7, 0) {$\cdots$};
\node at (0.8, -0.2) {\small $1$};
\node at (1.2, -0.4) {\small $m$};
\draw[<-] (0.65, 0.15) -- (0.95, 0.15);
\end{tikzpicture}
\equiv
&\left(-\frac{1}{\gamma_{m-1}}\right) \times
&&\begin{tikzpicture}[baseline={(0, -0.1)}]
\draw (0.4, 0) circle (.2);
\draw (1.2, 0) circle (.2);
\draw (0.6, 0) -- (1.0, 0);
\node at (0, 0) {$\cdots$};
\node at (1.7, 0) {$\cdots$};
\node at (0.8, -0.2) {\small $1$};
\node at (1.2, -0.4) {\small $m$};
\end{tikzpicture}
\notag\\
\begin{tikzpicture}[baseline={(0, -0.1)}]
\draw (0.4, 0) circle (.2);
\draw (1.2, 0) circle (.2);
\draw (0.6, 0) -- (1.0, 0);
\node at (0, 0) {$\cdots$};
\node at (1.7, 0) {$\cdots$};
\node at (0.8, -0.2) {\small $0$};
\node at (1.2, -0.4) {\small $m$};
\draw[->] (0.65, 0.15) -- (0.95, 0.15);
\end{tikzpicture}
\equiv
&\left(-\frac{1}{\gamma_{m-1}}\right) \times
&&\begin{tikzpicture}[baseline={(0, -0.1)}]
\draw (0.4, 0) circle (.2);
\draw (1.2, 0) circle (.2);
\draw (0.6, 0) -- (1.0, 0);
\node at (0, 0) {$\cdots$};
\node at (1.7, 0) {$\cdots$};
\node at (0.8, -0.2) {\small $0$};
\node at (1.2, -0.4) {\small $m$};
\end{tikzpicture}
\notag\\
\begin{tikzpicture}[baseline={(0, -0.1)}]
\draw (0.4, 0) circle (.2);
\draw (1.2, 0) circle (.2);
\draw (0.6, 0) -- (1.0, 0);
\node at (0, 0) {$\cdots$};
\node at (1.7, 0) {$\cdots$};
\node at (0.8, -0.2) {\small $1$};
\node at (1.2, -0.4) {\small $m$};
\draw[->] (0.65, 0.15) -- (0.95, 0.15);
\end{tikzpicture}
\equiv
&\left(-\frac{J_{m-1}}{\gamma_{m-1}^2}\right)\times
&&\begin{tikzpicture}[baseline={(0, -0.1)}]
\draw (0.4, 0) circle (.2);
\draw (1.2, 0) circle (.2);
\draw (0.6, 0) -- (1.0, 0);
\node at (0, 0) {$\cdots$};
\node at (1.7, 0) {$\cdots$};
\node at (0.8, -0.2) {\small $0$};
\node at (1.2, -0.4) {\small $m$};
\end{tikzpicture}
+
\left(-\frac{1}{\gamma_{m-1}}\right)\times
\begin{tikzpicture}[baseline={(0, -0.1)}]
\draw (0.4, 0) circle (.2);
\draw (1.2, 0) circle (.2);
\draw (0.6, 0) -- (1.0, 0);
\node at (0, 0) {$\cdots$};
\node at (1.7, 0) {$\cdots$};
\node at (0.8, -0.2) {\small $1$};
\node at (1.2, -0.4) {\small $m$};
\end{tikzpicture}
.\end{alignat}
The states are designed such that Eq.~\eqref{eqn: H acting on arrow states} holds.

\subsection{The Krylov subspaces for models with $M=3$}

The characteristic polynomial for models with $(M, Q)=(3, 1)$ and uniform hopping coefficients and interactions $\gamma_m = \gamma$, $J_m = J$ can be factorized:
\ie
\label{eqn: characteristic polynomial of M=3}
|\det \left(\lambda \mathds{1} - H \right)| = \lambda^4 \left[ \lambda^2-\left( 2\gamma^2+J^2 \right) \right]^2 \left[ \lambda^4 - \left( 4\gamma^2+2J^2 \right) \lambda^2 + 4\gamma^4 \right]
.\fe
This implies that there are 4 zero modes, 2 Krylov subspaces with dimension 2, and 1 with dimension 4. One can easily work out what the root states of the subspaces are:
\begin{alignat}{1}
\label{eqn: root states of Krylov subspaces of a M, Q = 3, 1 model}
&\mathcal{K}^{(1)}:
\begin{cases}
    \begin{tikzpicture}[baseline={(0, -0.1)}]
    \draw (0.4, 0) circle (.2);
    \draw (1.2, 0) circle (.2);
    \draw (2.0, 0) circle (.2);
    \fill (0.4, 0) circle (2pt);
    \draw (0.6, 0) -- (1.0, 0);
    \draw (1.4, 0) -- (1.8, 0);
    \node at (0.8, -0.2) {\small $0$};
    \node at (1.6, -0.2) {\small $0$};
    \draw[<-] (0.65, 0.15) -- (0.95, 0.15);
    \end{tikzpicture}
    -
    \begin{tikzpicture}[baseline={(0, -0.1)}]
    \draw (0.4, 0) circle (.2);
    \draw (1.2, 0) circle (.2);
    \draw (2.0, 0) circle (.2);
    \fill (2.0, 0) circle (2pt);
    \draw (0.6, 0) -- (1.0, 0);
    \draw (1.4, 0) -- (1.8, 0);
    \node at (0.8, -0.2) {\small $0$};
    \node at (1.6, -0.2) {\small $0$};
    \draw[->] (1.45, 0.15) -- (1.75, 0.15);
    \end{tikzpicture}
    , \\
    \begin{tikzpicture}[baseline={(0, -0.1)}]
    \draw (0.4, 0) circle (.2);
    \draw (1.2, 0) circle (.2);
    \draw (2.0, 0) circle (.2);
    \fill (0.4, 0) circle (2pt);
    \draw (0.6, 0) -- (1.0, 0);
    \draw (1.4, 0) -- (1.8, 0);
    \node at (0.8, -0.2) {\small $0$};
    \node at (1.6, -0.2) {\small $1$};
    \draw[<-] (0.65, 0.15) -- (0.95, 0.15);
    \end{tikzpicture}
    -
    \begin{tikzpicture}[baseline={(0, -0.1)}]
    \draw (0.4, 0) circle (.2);
    \draw (1.2, 0) circle (.2);
    \draw (2.0, 0) circle (.2);
    \fill (2.0, 0) circle (2pt);
    \draw (0.6, 0) -- (1.0, 0);
    \draw (1.4, 0) -- (1.8, 0);
    \node at (0.8, -0.2) {\small $0$};
    \node at (1.6, -0.2) {\small $1$};
    \draw[->] (1.45, 0.15) -- (1.75, 0.15);
    \end{tikzpicture}
    , \\
    \begin{tikzpicture}[baseline={(0, -0.1)}]
    \draw (0.4, 0) circle (.2);
    \draw (1.2, 0) circle (.2);
    \draw (2.0, 0) circle (.2);
    \fill (0.4, 0) circle (2pt);
    \draw (0.6, 0) -- (1.0, 0);
    \draw (1.4, 0) -- (1.8, 0);
    \node at (0.8, -0.2) {\small $1$};
    \node at (1.6, -0.2) {\small $0$};
    \draw[<-] (0.65, 0.15) -- (0.95, 0.15);
    \end{tikzpicture}
    -
    \begin{tikzpicture}[baseline={(0, -0.1)}]
    \draw (0.4, 0) circle (.2);
    \draw (1.2, 0) circle (.2);
    \draw (2.0, 0) circle (.2);
    \fill (2.0, 0) circle (2pt);
    \draw (0.6, 0) -- (1.0, 0);
    \draw (1.4, 0) -- (1.8, 0);
    \node at (0.8, -0.2) {\small $1$};
    \node at (1.6, -0.2) {\small $0$};
    \draw[->] (1.45, 0.15) -- (1.75, 0.15);
    \end{tikzpicture}
    , \\
    \begin{tikzpicture}[baseline={(0, -0.1)}]
    \draw (0.4, 0) circle (.2);
    \draw (1.2, 0) circle (.2);
    \draw (2.0, 0) circle (.2);
    \fill (0.4, 0) circle (2pt);
    \draw (0.6, 0) -- (1.0, 0);
    \draw (1.4, 0) -- (1.8, 0);
    \node at (0.8, -0.2) {\small $1$};
    \node at (1.6, -0.2) {\small $1$};
    \draw[<-] (0.65, 0.15) -- (0.95, 0.15);
    \end{tikzpicture}
    -
    \begin{tikzpicture}[baseline={(0, -0.1)}]
    \draw (0.4, 0) circle (.2);
    \draw (1.2, 0) circle (.2);
    \draw (2.0, 0) circle (.2);
    \fill (2.0, 0) circle (2pt);
    \draw (0.6, 0) -- (1.0, 0);
    \draw (1.4, 0) -- (1.8, 0);
    \node at (0.8, -0.2) {\small $1$};
    \node at (1.6, -0.2) {\small $1$};
    \draw[->] (1.45, 0.15) -- (1.75, 0.15);
    \end{tikzpicture}
    , \\
\end{cases}
\\
&\mathcal{K}^{(2)}:
\begin{cases}
    J \;
    \begin{tikzpicture}[baseline={(0, -0.1)}]
    \draw (0.4, 0) circle (.2);
    \draw (1.2, 0) circle (.2);
    \draw (2.0, 0) circle (.2);
    \fill (1.2, 0) circle (2pt);
    \draw (0.6, 0) -- (1.0, 0);
    \draw (1.4, 0) -- (1.8, 0);
    \node at (0.8, -0.2) {\small $0$};
    \node at (1.6, -0.2) {\small $0$};
    \end{tikzpicture}
    - \gamma \;
    \begin{tikzpicture}[baseline={(0, -0.1)}]
    \draw (0.4, 0) circle (.2);
    \draw (1.2, 0) circle (.2);
    \draw (2.0, 0) circle (.2);
    \fill (1.2, 0) circle (2pt);
    \draw (0.6, 0) -- (1.0, 0);
    \draw (1.4, 0) -- (1.8, 0);
    \node at (0.8, -0.2) {\small $0$};
    \node at (1.6, -0.2) {\small $1$};
    \end{tikzpicture}
    + \gamma \;
    \begin{tikzpicture}[baseline={(0, -0.1)}]
    \draw (0.4, 0) circle (.2);
    \draw (1.2, 0) circle (.2);
    \draw (2.0, 0) circle (.2);
    \fill (1.2, 0) circle (2pt);
    \draw (0.6, 0) -- (1.0, 0);
    \draw (1.4, 0) -- (1.8, 0);
    \node at (0.8, -0.2) {\small $1$};
    \node at (1.6, -0.2) {\small $0$};
    \end{tikzpicture}
    , \\
    J \;
    \begin{tikzpicture}[baseline={(0, -0.1)}]
    \draw (0.4, 0) circle (.2);
    \draw (1.2, 0) circle (.2);
    \draw (2.0, 0) circle (.2);
    \fill (1.2, 0) circle (2pt);
    \draw (0.6, 0) -- (1.0, 0);
    \draw (1.4, 0) -- (1.8, 0);
    \node at (0.8, -0.2) {\small $1$};
    \node at (1.6, -0.2) {\small $1$};
    \end{tikzpicture}
    - \gamma \;
    \begin{tikzpicture}[baseline={(0, -0.1)}]
    \draw (0.4, 0) circle (.2);
    \draw (1.2, 0) circle (.2);
    \draw (2.0, 0) circle (.2);
    \fill (1.2, 0) circle (2pt);
    \draw (0.6, 0) -- (1.0, 0);
    \draw (1.4, 0) -- (1.8, 0);
    \node at (0.8, -0.2) {\small $0$};
    \node at (1.6, -0.2) {\small $1$};
    \end{tikzpicture}
    + \gamma \;
    \begin{tikzpicture}[baseline={(0, -0.1)}]
    \draw (0.4, 0) circle (.2);
    \draw (1.2, 0) circle (.2);
    \draw (2.0, 0) circle (.2);
    \fill (1.2, 0) circle (2pt);
    \draw (0.6, 0) -- (1.0, 0);
    \draw (1.4, 0) -- (1.8, 0);
    \node at (0.8, -0.2) {\small $1$};
    \node at (1.6, -0.2) {\small $0$};
    \end{tikzpicture}
    ,
\end{cases}
\\
&\mathcal{K}^{(4)}:
    J \;
    \begin{tikzpicture}[baseline={(0, -0.1)}]
    \draw (0.4, 0) circle (.2);
    \draw (1.2, 0) circle (.2);
    \draw (2.0, 0) circle (.2);
    \fill (1.2, 0) circle (2pt);
    \draw (0.6, 0) -- (1.0, 0);
    \draw (1.4, 0) -- (1.8, 0);
    \node at (0.8, -0.2) {\small $1$};
    \node at (1.6, -0.2) {\small $0$};
    \end{tikzpicture}
    - \gamma \;
    \begin{tikzpicture}[baseline={(0, -0.1)}]
    \draw (0.4, 0) circle (.2);
    \draw (1.2, 0) circle (.2);
    \draw (2.0, 0) circle (.2);
    \fill (1.2, 0) circle (2pt);
    \draw (0.6, 0) -- (1.0, 0);
    \draw (1.4, 0) -- (1.8, 0);
    \node at (0.8, -0.2) {\small $0$};
    \node at (1.6, -0.2) {\small $0$};
    \end{tikzpicture}
    - \gamma \;
    \begin{tikzpicture}[baseline={(0, -0.1)}]
    \draw (0.4, 0) circle (.2);
    \draw (1.2, 0) circle (.2);
    \draw (2.0, 0) circle (.2);
    \fill (1.2, 0) circle (2pt);
    \draw (0.6, 0) -- (1.0, 0);
    \draw (1.4, 0) -- (1.8, 0);
    \node at (0.8, -0.2) {\small $1$};
    \node at (1.6, -0.2) {\small $1$};
    \end{tikzpicture}
    .
\end{alignat}

\subsection{Notations for $H_{\bar{\gamma} \bar{J}}$: $\mathcal{K}^{(2)}$}
\label{appendix: Explicit forms of the notations for K2}
Motivated by the root states of $\mathcal{K}^{(2)}$ in Eq.~\eqref{eqn: root states of Krylov subspaces of a M, Q = 3, 1 model}, the singly-rectangularized states that meet the requirements in \ref{subsubsection: Notations for K2} are:
\begin{alignat}{8}
\begin{tikzpicture}[baseline={(0, -0.1)}]
\draw (0.0, 0.25) rectangle (2.4, -0.25);
\draw (0.4, 0) circle (.2);
\draw (1.2, 0) circle (.2);
\draw (2.0, 0) circle (.2);
\draw (0.6, 0) -- (1.0, 0);
\draw (1.4, 0) -- (1.8, 0);
\node at (0, 0) {$\cdots$};
\node at (2.5, 0) {$\cdots$};
\node at (1.2, 0.4) {\small $0$};
\end{tikzpicture}
\equiv \; &
J
\begin{tikzpicture}[baseline={(0, -0.1)}]
\draw (0.4, 0) circle (.2);
\draw (1.2, 0) circle (.2);
\draw (2.0, 0) circle (.2);
\draw (0.6, 0) -- (1.0, 0);
\draw (1.4, 0) -- (1.8, 0);
\node at (0, 0) {$\cdots$};
\node at (2.5, 0) {$\cdots$};
\node at (0.8, -0.2) {\small $0$};
\node at (1.6, -0.2) {\small $0$};
\end{tikzpicture}
- \gamma
\begin{tikzpicture}[baseline={(0, -0.1)}]
\draw (0.4, 0) circle (.2);
\draw (1.2, 0) circle (.2);
\draw (2.0, 0) circle (.2);
\draw (0.6, 0) -- (1.0, 0);
\draw (1.4, 0) -- (1.8, 0);
\node at (0, 0) {$\cdots$};
\node at (2.5, 0) {$\cdots$};
\node at (0.8, -0.2) {\small $0$};
\node at (1.6, -0.2) {\small $1$};
\end{tikzpicture}
+ \gamma
\begin{tikzpicture}[baseline={(0, -0.1)}]
\draw (0.4, 0) circle (.2);
\draw (1.2, 0) circle (.2);
\draw (2.0, 0) circle (.2);
\draw (0.6, 0) -- (1.0, 0);
\draw (1.4, 0) -- (1.8, 0);
\node at (0, 0) {$\cdots$};
\node at (2.5, 0) {$\cdots$};
\node at (0.8, -0.2) {\small $1$};
\node at (1.6, -0.2) {\small $0$};
\end{tikzpicture}
,
\notag \\
\begin{tikzpicture}[baseline={(0, -0.1)}]
\draw (0.0, 0.25) rectangle (2.4, -0.25);
\draw (0.4, 0) circle (.2);
\draw (1.2, 0) circle (.2);
\draw (2.0, 0) circle (.2);
\draw (0.6, 0) -- (1.0, 0);
\draw (1.4, 0) -- (1.8, 0);
\draw[<-] (0.65, 0.4) -- (0.95, 0.4);
\node at (0, 0) {$\cdots$};
\node at (2.5, 0) {$\cdots$};
\node at (1.2, 0.4) {\small $0$};
\end{tikzpicture}
\equiv \; &
\gamma^2
\begin{tikzpicture}[baseline={(0, -0.1)}]
\draw (0.4, 0) circle (.2);
\draw (1.2, 0) circle (.2);
\draw (2.0, 0) circle (.2);
\draw (0.6, 0) -- (1.0, 0);
\draw (1.4, 0) -- (1.8, 0);
\node at (0, 0) {$\cdots$};
\node at (2.5, 0) {$\cdots$};
\node at (0.8, -0.2) {\small $0$};
\node at (1.6, -0.2) {\small $1$};
\end{tikzpicture}
- \gamma^2
\begin{tikzpicture}[baseline={(0, -0.1)}]
\draw (0.4, 0) circle (.2);
\draw (1.2, 0) circle (.2);
\draw (2.0, 0) circle (.2);
\draw (0.6, 0) -- (1.0, 0);
\draw (1.4, 0) -- (1.8, 0);
\node at (0, 0) {$\cdots$};
\node at (2.5, 0) {$\cdots$};
\node at (0.8, -0.2) {\small $1$};
\node at (1.6, -0.2) {\small $0$};
\end{tikzpicture}
,
\notag \\
\begin{tikzpicture}[baseline={(0, -0.1)}]
\draw (0.0, 0.25) rectangle (2.4, -0.25);
\draw (0.4, 0) circle (.2);
\draw (1.2, 0) circle (.2);
\draw (2.0, 0) circle (.2);
\draw (0.6, 0) -- (1.0, 0);
\draw (1.4, 0) -- (1.8, 0);
\draw[->] (1.45, 0.4) -- (1.75, 0.4);
\node at (0, 0) {$\cdots$};
\node at (2.5, 0) {$\cdots$};
\node at (1.2, 0.4) {\small $0$};
\end{tikzpicture}
\equiv \; &
- \gamma J
\begin{tikzpicture}[baseline={(0, -0.1)}]
\draw (0.4, 0) circle (.2);
\draw (1.2, 0) circle (.2);
\draw (2.0, 0) circle (.2);
\draw (0.6, 0) -- (1.0, 0);
\draw (1.4, 0) -- (1.8, 0);
\node at (0, 0) {$\cdots$};
\node at (2.5, 0) {$\cdots$};
\node at (0.8, -0.2) {\small $0$};
\node at (1.6, -0.2) {\small $0$};
\end{tikzpicture}
+ (\gamma^2+J^2)
\begin{tikzpicture}[baseline={(0, -0.1)}]
\draw (0.4, 0) circle (.2);
\draw (1.2, 0) circle (.2);
\draw (2.0, 0) circle (.2);
\draw (0.6, 0) -- (1.0, 0);
\draw (1.4, 0) -- (1.8, 0);
\node at (0, 0) {$\cdots$};
\node at (2.5, 0) {$\cdots$};
\node at (0.8, -0.2) {\small $0$};
\node at (1.6, -0.2) {\small $1$};
\end{tikzpicture}
\notag \\
& - \gamma^2
\begin{tikzpicture}[baseline={(0, -0.1)}]
\draw (0.4, 0) circle (.2);
\draw (1.2, 0) circle (.2);
\draw (2.0, 0) circle (.2);
\draw (0.6, 0) -- (1.0, 0);
\draw (1.4, 0) -- (1.8, 0);
\node at (0, 0) {$\cdots$};
\node at (2.5, 0) {$\cdots$};
\node at (0.8, -0.2) {\small $1$};
\node at (1.6, -0.2) {\small $0$};
\end{tikzpicture}
+ \gamma J
\begin{tikzpicture}[baseline={(0, -0.1)}]
\draw (0.4, 0) circle (.2);
\draw (1.2, 0) circle (.2);
\draw (2.0, 0) circle (.2);
\draw (0.6, 0) -- (1.0, 0);
\draw (1.4, 0) -- (1.8, 0);
\node at (0, 0) {$\cdots$};
\node at (2.5, 0) {$\cdots$};
\node at (0.8, -0.2) {\small $1$};
\node at (1.6, -0.2) {\small $1$};
\end{tikzpicture},
\notag \\
\begin{tikzpicture}[baseline={(0, -0.1)}]
\draw (0.0, 0.25) rectangle (2.4, -0.25);
\draw (0.4, 0) circle (.2);
\draw (1.2, 0) circle (.2);
\draw (2.0, 0) circle (.2);
\draw (0.6, 0) -- (1.0, 0);
\draw (1.4, 0) -- (1.8, 0);
\node at (0, 0) {$\cdots$};
\node at (2.5, 0) {$\cdots$};
\node at (1.2, 0.4) {\small $1$};
\end{tikzpicture}
\equiv \; &
- \gamma
\begin{tikzpicture}[baseline={(0, -0.1)}]
\draw (0.4, 0) circle (.2);
\draw (1.2, 0) circle (.2);
\draw (2.0, 0) circle (.2);
\draw (0.6, 0) -- (1.0, 0);
\draw (1.4, 0) -- (1.8, 0);
\node at (0, 0) {$\cdots$};
\node at (2.5, 0) {$\cdots$};
\node at (0.8, -0.2) {\small $0$};
\node at (1.6, -0.2) {\small $1$};
\end{tikzpicture}
+ \gamma
\begin{tikzpicture}[baseline={(0, -0.1)}]
\draw (0.4, 0) circle (.2);
\draw (1.2, 0) circle (.2);
\draw (2.0, 0) circle (.2);
\draw (0.6, 0) -- (1.0, 0);
\draw (1.4, 0) -- (1.8, 0);
\node at (0, 0) {$\cdots$};
\node at (2.5, 0) {$\cdots$};
\node at (0.8, -0.2) {\small $1$};
\node at (1.6, -0.2) {\small $0$};
\end{tikzpicture}
+ J
\begin{tikzpicture}[baseline={(0, -0.1)}]
\draw (0.4, 0) circle (.2);
\draw (1.2, 0) circle (.2);
\draw (2.0, 0) circle (.2);
\draw (0.6, 0) -- (1.0, 0);
\draw (1.4, 0) -- (1.8, 0);
\node at (0, 0) {$\cdots$};
\node at (2.5, 0) {$\cdots$};
\node at (0.8, -0.2) {\small $1$};
\node at (1.6, -0.2) {\small $1$};
\end{tikzpicture}
\notag \\
\begin{tikzpicture}[baseline={(0, -0.1)}]
\draw (0.0, 0.25) rectangle (2.4, -0.25);
\draw (0.4, 0) circle (.2);
\draw (1.2, 0) circle (.2);
\draw (2.0, 0) circle (.2);
\draw (0.6, 0) -- (1.0, 0);
\draw (1.4, 0) -- (1.8, 0);
\draw[<-] (0.65, 0.4) -- (0.95, 0.4);
\node at (0, 0) {$\cdots$};
\node at (2.5, 0) {$\cdots$};
\node at (1.2, 0.4) {\small $1$};
\end{tikzpicture}
\equiv \; &
\gamma J
\begin{tikzpicture}[baseline={(0, -0.1)}]
\draw (0.4, 0) circle (.2);
\draw (1.2, 0) circle (.2);
\draw (2.0, 0) circle (.2);
\draw (0.6, 0) -- (1.0, 0);
\draw (1.4, 0) -- (1.8, 0);
\node at (0, 0) {$\cdots$};
\node at (2.5, 0) {$\cdots$};
\node at (0.8, -0.2) {\small $0$};
\node at (1.6, -0.2) {\small $0$};
\end{tikzpicture}
+ (\gamma^2 + J^2)
\begin{tikzpicture}[baseline={(0, -0.1)}]
\draw (0.4, 0) circle (.2);
\draw (1.2, 0) circle (.2);
\draw (2.0, 0) circle (.2);
\draw (0.6, 0) -- (1.0, 0);
\draw (1.4, 0) -- (1.8, 0);
\node at (0, 0) {$\cdots$};
\node at (2.5, 0) {$\cdots$};
\node at (0.8, -0.2) {\small $0$};
\node at (1.6, -0.2) {\small $1$};
\end{tikzpicture}
\notag \\
& - \gamma^2
\begin{tikzpicture}[baseline={(0, -0.1)}]
\draw (0.4, 0) circle (.2);
\draw (1.2, 0) circle (.2);
\draw (2.0, 0) circle (.2);
\draw (0.6, 0) -- (1.0, 0);
\draw (1.4, 0) -- (1.8, 0);
\node at (0, 0) {$\cdots$};
\node at (2.5, 0) {$\cdots$};
\node at (0.8, -0.2) {\small $1$};
\node at (1.6, -0.2) {\small $0$};
\end{tikzpicture}
- \gamma J
\begin{tikzpicture}[baseline={(0, -0.1)}]
\draw (0.4, 0) circle (.2);
\draw (1.2, 0) circle (.2);
\draw (2.0, 0) circle (.2);
\draw (0.6, 0) -- (1.0, 0);
\draw (1.4, 0) -- (1.8, 0);
\node at (0, 0) {$\cdots$};
\node at (2.5, 0) {$\cdots$};
\node at (0.8, -0.2) {\small $1$};
\node at (1.6, -0.2) {\small $1$};
\end{tikzpicture}
\notag \\
\begin{tikzpicture}[baseline={(0, -0.1)}]
\draw (0.0, 0.25) rectangle (2.4, -0.25);
\draw (0.4, 0) circle (.2);
\draw (1.2, 0) circle (.2);
\draw (2.0, 0) circle (.2);
\draw (0.6, 0) -- (1.0, 0);
\draw (1.4, 0) -- (1.8, 0);
\draw[->] (1.45, 0.4) -- (1.75, 0.4);
\node at (0, 0) {$\cdots$};
\node at (2.5, 0) {$\cdots$};
\node at (1.2, 0.4) {\small $1$};
\end{tikzpicture}
\equiv \; &
\gamma^2
\begin{tikzpicture}[baseline={(0, -0.1)}]
\draw (0.4, 0) circle (.2);
\draw (1.2, 0) circle (.2);
\draw (2.0, 0) circle (.2);
\draw (0.6, 0) -- (1.0, 0);
\draw (1.4, 0) -- (1.8, 0);
\node at (0, 0) {$\cdots$};
\node at (2.5, 0) {$\cdots$};
\node at (0.8, -0.2) {\small $0$};
\node at (1.6, -0.2) {\small $1$};
\end{tikzpicture}
- \gamma^2
\begin{tikzpicture}[baseline={(0, -0.1)}]
\draw (0.4, 0) circle (.2);
\draw (1.2, 0) circle (.2);
\draw (2.0, 0) circle (.2);
\draw (0.6, 0) -- (1.0, 0);
\draw (1.4, 0) -- (1.8, 0);
\node at (0, 0) {$\cdots$};
\node at (2.5, 0) {$\cdots$};
\node at (0.8, -0.2) {\small $1$};
\node at (1.6, -0.2) {\small $0$};
\end{tikzpicture}
.\end{alignat}
One can easily verify that
\begin{alignat}{1}
\hat{h}_m^{(L)} \;
\begin{tikzpicture}[baseline={(0, -0.1)}]
\draw (0.0, 0.25) rectangle (2.4, -0.25);
\draw (0.4, 0) circle (.2);
\draw (1.2, 0) circle (.2);
\draw (2.0, 0) circle (.2);
\fill (1.2, 0) circle (2pt);
\draw (0.6, 0) -- (1.0, 0);
\draw (1.4, 0) -- (1.8, 0);
\node at (0, 0) {$\cdots$};
\node at (2.5, 0) {$\cdots$};
\node at (1.2, 0.4) {\small $r_j$};
\node at (1.2, -0.4) {\small $m$};
\end{tikzpicture}
&=
\begin{tikzpicture}[baseline={(0, -0.1)}]
\draw (0.0, 0.25) rectangle (2.4, -0.25);
\draw (0.4, 0) circle (.2);
\draw (1.2, 0) circle (.2);
\draw (2.0, 0) circle (.2);
\fill (0.4, 0) circle (2pt);
\draw (0.6, 0) -- (1.0, 0);
\draw (1.4, 0) -- (1.8, 0);
\draw[<-] (0.65, 0.4) -- (0.95, 0.4);
\node at (0, 0) {$\cdots$};
\node at (2.5, 0) {$\cdots$};
\node at (1.2, 0.4) {\small $r_j$};
\node at (1.2, -0.4) {\small $m$};
\end{tikzpicture}
, \notag \\
\hat{h}_m^{(R)}
\begin{tikzpicture}[baseline={(0, -0.1)}]
\draw (0.0, 0.25) rectangle (2.4, -0.25);
\draw (0.4, 0) circle (.2);
\draw (1.2, 0) circle (.2);
\draw (2.0, 0) circle (.2);
\fill (1.2, 0) circle (2pt);
\draw (0.6, 0) -- (1.0, 0);
\draw (1.4, 0) -- (1.8, 0);
\node at (0, 0) {$\cdots$};
\node at (2.5, 0) {$\cdots$};
\node at (1.2, 0.4) {\small $r_j$};
\node at (1.2, -0.4) {\small $m$};
\end{tikzpicture}
&=
\begin{tikzpicture}[baseline={(0, -0.1)}]
\draw (0.0, 0.25) rectangle (2.4, -0.25);
\draw (0.4, 0) circle (.2);
\draw (1.2, 0) circle (.2);
\draw (2.0, 0) circle (.2);
\fill (2.0, 0) circle (2pt);
\draw (0.6, 0) -- (1.0, 0);
\draw (1.4, 0) -- (1.8, 0);
\draw[->] (1.45, 0.4) -- (1.75, 0.4);
\node at (0, 0) {$\cdots$};
\node at (2.5, 0) {$\cdots$};
\node at (1.2, 0.4) {\small $r_j$};
\node at (1.2, -0.4) {\small $m$};
\end{tikzpicture}
.\end{alignat}
In addition, it is also straight forward to show that
\begin{alignat}{1}
\label{eqn: The coefficients for spanned space}
& \hat{h}_{m-1}^{(R)} \;
\begin{tikzpicture}[baseline={(0, -0.1)}]
\draw (0.0, 0.25) rectangle (2.4, -0.25);
\draw (0.4, 0) circle (.2);
\draw (1.2, 0) circle (.2);
\draw (2.0, 0) circle (.2);
\fill (0.4, 0) circle (2pt);
\draw (0.6, 0) -- (1.0, 0);
\draw (1.4, 0) -- (1.8, 0);
\draw[<-] (0.65, 0.4) -- (0.95, 0.4);
\node at (0, 0) {$\cdots$};
\node at (2.5, 0) {$\cdots$};
\node at (1.2, 0.4) {\small $0$};
\node at (1.2, -0.4) {\small $m$};
\end{tikzpicture}
=
\gamma^2
\begin{tikzpicture}[baseline={(0, -0.1)}]
\draw (0.0, 0.25) rectangle (2.4, -0.25);
\draw (0.4, 0) circle (.2);
\draw (1.2, 0) circle (.2);
\draw (2.0, 0) circle (.2);
\fill (1.2, 0) circle (2pt);
\draw (0.6, 0) -- (1.0, 0);
\draw (1.4, 0) -- (1.8, 0);
\node at (0, 0) {$\cdots$};
\node at (2.5, 0) {$\cdots$};
\node at (1.2, 0.4) {\small $1$};
\end{tikzpicture}
, \notag \\
& \hat{h}_{m+1}^{(L)} \;
\begin{tikzpicture}[baseline={(0, -0.1)}]
\draw (0.0, 0.25) rectangle (2.4, -0.25);
\draw (0.4, 0) circle (.2);
\draw (1.2, 0) circle (.2);
\draw (2.0, 0) circle (.2);
\fill (2.0, 0) circle (2pt);
\draw (0.6, 0) -- (1.0, 0);
\draw (1.4, 0) -- (1.8, 0);
\draw[->] (1.45, 0.4) -- (1.75, 0.4);
\node at (0, 0) {$\cdots$};
\node at (2.5, 0) {$\cdots$};
\node at (1.2, 0.4) {\small $0$};
\node at (1.2, -0.4) {\small $m$};
\end{tikzpicture}
=
\left( 2\gamma^2+J^2 \right)
\begin{tikzpicture}[baseline={(0, -0.1)}]
\draw (0.0, 0.25) rectangle (2.4, -0.25);
\draw (0.4, 0) circle (.2);
\draw (1.2, 0) circle (.2);
\draw (2.0, 0) circle (.2);
\fill (1.2, 0) circle (2pt);
\draw (0.6, 0) -- (1.0, 0);
\draw (1.4, 0) -- (1.8, 0);
\node at (0, 0) {$\cdots$};
\node at (2.5, 0) {$\cdots$};
\node at (1.2, 0.4) {\small $0$};
\end{tikzpicture}
- \gamma^2
\begin{tikzpicture}[baseline={(0, -0.1)}]
\draw (0.0, 0.25) rectangle (2.4, -0.25);
\draw (0.4, 0) circle (.2);
\draw (1.2, 0) circle (.2);
\draw (2.0, 0) circle (.2);
\fill (1.2, 0) circle (2pt);
\draw (0.6, 0) -- (1.0, 0);
\draw (1.4, 0) -- (1.8, 0);
\node at (0, 0) {$\cdots$};
\node at (2.5, 0) {$\cdots$};
\node at (1.2, 0.4) {\small $1$};
\end{tikzpicture}
,\notag \\
& \hat{h}_{m-1}^{(R)} \;
\begin{tikzpicture}[baseline={(0, -0.1)}]
\draw (0.0, 0.25) rectangle (2.4, -0.25);
\draw (0.4, 0) circle (.2);
\draw (1.2, 0) circle (.2);
\draw (2.0, 0) circle (.2);
\fill (0.4, 0) circle (2pt);
\draw (0.6, 0) -- (1.0, 0);
\draw (1.4, 0) -- (1.8, 0);
\draw[<-] (0.65, 0.4) -- (0.95, 0.4);
\node at (0, 0) {$\cdots$};
\node at (2.5, 0) {$\cdots$};
\node at (1.2, 0.4) {\small $1$};
\node at (1.2, -0.4) {\small $m$};
\end{tikzpicture}
=
- \gamma^2
\begin{tikzpicture}[baseline={(0, -0.1)}]
\draw (0.0, 0.25) rectangle (2.4, -0.25);
\draw (0.4, 0) circle (.2);
\draw (1.2, 0) circle (.2);
\draw (2.0, 0) circle (.2);
\fill (1.2, 0) circle (2pt);
\draw (0.6, 0) -- (1.0, 0);
\draw (1.4, 0) -- (1.8, 0);
\node at (0, 0) {$\cdots$};
\node at (2.5, 0) {$\cdots$};
\node at (1.2, 0.4) {\small $0$};
\end{tikzpicture}
+ \left( 2\gamma^2+J^2 \right)
\begin{tikzpicture}[baseline={(0, -0.1)}]
\draw (0.0, 0.25) rectangle (2.4, -0.25);
\draw (0.4, 0) circle (.2);
\draw (1.2, 0) circle (.2);
\draw (2.0, 0) circle (.2);
\fill (1.2, 0) circle (2pt);
\draw (0.6, 0) -- (1.0, 0);
\draw (1.4, 0) -- (1.8, 0);
\node at (0, 0) {$\cdots$};
\node at (2.5, 0) {$\cdots$};
\node at (1.2, 0.4) {\small $1$};
\end{tikzpicture}
, \notag \\
& \hat{h}_{m+1}^{(L)} \;
\begin{tikzpicture}[baseline={(0, -0.1)}]
\draw (0.0, 0.25) rectangle (2.4, -0.25);
\draw (0.4, 0) circle (.2);
\draw (1.2, 0) circle (.2);
\draw (2.0, 0) circle (.2);
\fill (2.0, 0) circle (2pt);
\draw (0.6, 0) -- (1.0, 0);
\draw (1.4, 0) -- (1.8, 0);
\draw[->] (1.45, 0.4) -- (1.75, 0.4);
\node at (0, 0) {$\cdots$};
\node at (2.5, 0) {$\cdots$};
\node at (1.2, 0.4) {\small $1$};
\node at (1.2, -0.4) {\small $m$};
\end{tikzpicture}
=
\gamma^2
\begin{tikzpicture}[baseline={(0, -0.1)}]
\draw (0.0, 0.25) rectangle (2.4, -0.25);
\draw (0.4, 0) circle (.2);
\draw (1.2, 0) circle (.2);
\draw (2.0, 0) circle (.2);
\fill (1.2, 0) circle (2pt);
\draw (0.6, 0) -- (1.0, 0);
\draw (1.4, 0) -- (1.8, 0);
\node at (0, 0) {$\cdots$};
\node at (2.5, 0) {$\cdots$};
\node at (1.2, 0.4) {\small $0$};
\end{tikzpicture}
,\end{alignat}
which then implies Eq.~\eqref{eqn: H^2 acting on rectangularized states}.

\section{$H_{\bar{\gamma}\bar{J}}$: $Q=1$, $\mathcal{K}^{(4)}$ subspaces}
\label{appendix: Q=1, odd M, K=4}

Motivated by the root state of $\mathcal{K}^{(4)}$ in Eq.~\eqref{eqn: root states of Krylov subspaces of a M, Q = 3, 1 model}, we define the doubly neutral-/left-/right-rectangularized states that enable us to construct 4-dimensional Krylov subspaces as follows:

\begin{alignat}{12}
\label{eqn: Doubly rectangular exact form}
\begin{tikzpicture}[baseline={(0, -0.1)}]
\draw (0.05, 0.25) rectangle (2.35, -0.25);
\draw (0.0, 0.3) rectangle (2.4, -0.3);
\draw (0.4, 0) circle (.2);
\draw (1.2, 0) circle (.2);
\draw (2.0, 0) circle (.2);
\draw (0.6, 0) -- (1.0, 0);
\draw (1.4, 0) -- (1.8, 0);
\node at (0, 0) {$\cdots$};
\node at (2.5, 0) {$\cdots$};
\node at (1.2, 0.45) {\small $0$};
\end{tikzpicture}
\equiv \; &
- \gamma
\begin{tikzpicture}[baseline={(0, -0.1)}]
\draw (0.4, 0) circle (.2);
\draw (1.2, 0) circle (.2);
\draw (2.0, 0) circle (.2);
\draw (0.6, 0) -- (1.0, 0);
\draw (1.4, 0) -- (1.8, 0);
\node at (0, 0) {$\cdots$};
\node at (2.5, 0) {$\cdots$};
\node at (0.8, -0.2) {\small $0$};
\node at (1.6, -0.2) {\small $0$};
\end{tikzpicture}
+ J
\begin{tikzpicture}[baseline={(0, -0.1)}]
\draw (0.4, 0) circle (.2);
\draw (1.2, 0) circle (.2);
\draw (2.0, 0) circle (.2);
\draw (0.6, 0) -- (1.0, 0);
\draw (1.4, 0) -- (1.8, 0);
\node at (0, 0) {$\cdots$};
\node at (2.5, 0) {$\cdots$};
\node at (0.8, -0.2) {\small $1$};
\node at (1.6, -0.2) {\small $0$};
\end{tikzpicture}
- \gamma
\begin{tikzpicture}[baseline={(0, -0.1)}]
\draw (0.4, 0) circle (.2);
\draw (1.2, 0) circle (.2);
\draw (2.0, 0) circle (.2);
\draw (0.6, 0) -- (1.0, 0);
\draw (1.4, 0) -- (1.8, 0);
\node at (0, 0) {$\cdots$};
\node at (2.5, 0) {$\cdots$};
\node at (0.8, -0.2) {\small $1$};
\node at (1.6, -0.2) {\small $1$};
\end{tikzpicture}
,
\notag \\
\begin{tikzpicture}[baseline={(0, -0.1)}]
\draw (0.05, 0.25) rectangle (2.35, -0.25);
\draw (0.0, 0.3) rectangle (2.4, -0.3);
\draw (0.4, 0) circle (.2);
\draw (1.2, 0) circle (.2);
\draw (2.0, 0) circle (.2);
\draw (0.6, 0) -- (1.0, 0);
\draw (1.4, 0) -- (1.8, 0);
\draw[<-] (0.65, 0.45) -- (0.95, 0.45);
\node at (0, 0) {$\cdots$};
\node at (2.5, 0) {$\cdots$};
\node at (1.2, 0.45) {\small $0$};
\end{tikzpicture}
\equiv \; &
(\gamma^2 + J^2)
\begin{tikzpicture}[baseline={(0, -0.1)}]
\draw (0.4, 0) circle (.2);
\draw (1.2, 0) circle (.2);
\draw (2.0, 0) circle (.2);
\draw (0.6, 0) -- (1.0, 0);
\draw (1.4, 0) -- (1.8, 0);
\node at (0, 0) {$\cdots$};
\node at (2.5, 0) {$\cdots$};
\node at (0.8, -0.2) {\small $0$};
\node at (1.6, -0.2) {\small $0$};
\end{tikzpicture}
- \gamma J
\begin{tikzpicture}[baseline={(0, -0.1)}]
\draw (0.4, 0) circle (.2);
\draw (1.2, 0) circle (.2);
\draw (2.0, 0) circle (.2);
\draw (0.6, 0) -- (1.0, 0);
\draw (1.4, 0) -- (1.8, 0);
\node at (0, 0) {$\cdots$};
\node at (2.5, 0) {$\cdots$};
\node at (0.8, -0.2) {\small $0$};
\node at (1.6, -0.2) {\small $1$};
\end{tikzpicture}
\notag \\ &
- \gamma J
\begin{tikzpicture}[baseline={(0, -0.1)}]
\draw (0.4, 0) circle (.2);
\draw (1.2, 0) circle (.2);
\draw (2.0, 0) circle (.2);
\draw (0.6, 0) -- (1.0, 0);
\draw (1.4, 0) -- (1.8, 0);
\node at (0, 0) {$\cdots$};
\node at (2.5, 0) {$\cdots$};
\node at (0.8, -0.2) {\small $1$};
\node at (1.6, -0.2) {\small $0$};
\end{tikzpicture}
+ \gamma^2
\begin{tikzpicture}[baseline={(0, -0.1)}]
\draw (0.4, 0) circle (.2);
\draw (1.2, 0) circle (.2);
\draw (2.0, 0) circle (.2);
\draw (0.6, 0) -- (1.0, 0);
\draw (1.4, 0) -- (1.8, 0);
\node at (0, 0) {$\cdots$};
\node at (2.5, 0) {$\cdots$};
\node at (0.8, -0.2) {\small $1$};
\node at (1.6, -0.2) {\small $1$};
\end{tikzpicture}
,
\notag \\
\begin{tikzpicture}[baseline={(0, -0.1)}]
\draw (0.05, 0.25) rectangle (2.35, -0.25);
\draw (0.0, 0.3) rectangle (2.4, -0.3);
\draw (0.4, 0) circle (.2);
\draw (1.2, 0) circle (.2);
\draw (2.0, 0) circle (.2);
\draw (0.6, 0) -- (1.0, 0);
\draw (1.4, 0) -- (1.8, 0);
\draw[->] (1.45, 0.45) -- (1.75, 0.45);
\node at (0, 0) {$\cdots$};
\node at (2.5, 0) {$\cdots$};
\node at (1.2, 0.45) {\small $0$};
\end{tikzpicture}
\equiv \; &
- \gamma^2
\begin{tikzpicture}[baseline={(0, -0.1)}]
\draw (0.4, 0) circle (.2);
\draw (1.2, 0) circle (.2);
\draw (2.0, 0) circle (.2);
\draw (0.6, 0) -- (1.0, 0);
\draw (1.4, 0) -- (1.8, 0);
\node at (0, 0) {$\cdots$};
\node at (2.5, 0) {$\cdots$};
\node at (0.8, -0.2) {\small $0$};
\node at (1.6, -0.2) {\small $0$};
\end{tikzpicture}
- \gamma J
\begin{tikzpicture}[baseline={(0, -0.1)}]
\draw (0.4, 0) circle (.2);
\draw (1.2, 0) circle (.2);
\draw (2.0, 0) circle (.2);
\draw (0.6, 0) -- (1.0, 0);
\draw (1.4, 0) -- (1.8, 0);
\node at (0, 0) {$\cdots$};
\node at (2.5, 0) {$\cdots$};
\node at (0.8, -0.2) {\small $0$};
\node at (1.6, -0.2) {\small $1$};
\end{tikzpicture}
\notag \\ &
- \gamma J
\begin{tikzpicture}[baseline={(0, -0.1)}]
\draw (0.4, 0) circle (.2);
\draw (1.2, 0) circle (.2);
\draw (2.0, 0) circle (.2);
\draw (0.6, 0) -- (1.0, 0);
\draw (1.4, 0) -- (1.8, 0);
\node at (0, 0) {$\cdots$};
\node at (2.5, 0) {$\cdots$};
\node at (0.8, -0.2) {\small $1$};
\node at (1.6, -0.2) {\small $0$};
\end{tikzpicture}
+ (\gamma^2 + J^2)
\begin{tikzpicture}[baseline={(0, -0.1)}]
\draw (0.4, 0) circle (.2);
\draw (1.2, 0) circle (.2);
\draw (2.0, 0) circle (.2);
\draw (0.6, 0) -- (1.0, 0);
\draw (1.4, 0) -- (1.8, 0);
\node at (0, 0) {$\cdots$};
\node at (2.5, 0) {$\cdots$};
\node at (0.8, -0.2) {\small $1$};
\node at (1.6, -0.2) {\small $1$};
\end{tikzpicture},
\notag \\
\begin{tikzpicture}[baseline={(0, -0.1)}]
\draw (0.05, 0.25) rectangle (2.35, -0.25);
\draw (0.0, 0.3) rectangle (2.4, -0.3);
\draw (0.4, 0) circle (.2);
\draw (1.2, 0) circle (.2);
\draw (2.0, 0) circle (.2);
\draw (0.6, 0) -- (1.0, 0);
\draw (1.4, 0) -- (1.8, 0);
\node at (0, 0) {$\cdots$};
\node at (2.5, 0) {$\cdots$};
\node at (1.2, 0.45) {\small $1$};
\end{tikzpicture}
\equiv \; &
-2 \gamma (\gamma^2 + J^2)
\begin{tikzpicture}[baseline={(0, -0.1)}]
\draw (0.4, 0) circle (.2);
\draw (1.2, 0) circle (.2);
\draw (2.0, 0) circle (.2);
\draw (0.6, 0) -- (1.0, 0);
\draw (1.4, 0) -- (1.8, 0);
\node at (0, 0) {$\cdots$};
\node at (2.5, 0) {$\cdots$};
\node at (0.8, -0.2) {\small $0$};
\node at (1.6, -0.2) {\small $0$};
\end{tikzpicture}
+ 2 J \gamma^2
\begin{tikzpicture}[baseline={(0, -0.1)}]
\draw (0.4, 0) circle (.2);
\draw (1.2, 0) circle (.2);
\draw (2.0, 0) circle (.2);
\draw (0.6, 0) -- (1.0, 0);
\draw (1.4, 0) -- (1.8, 0);
\node at (0, 0) {$\cdots$};
\node at (2.5, 0) {$\cdots$};
\node at (0.8, -0.2) {\small $0$};
\node at (1.6, -0.2) {\small $1$};
\end{tikzpicture}
\notag \\ &
+ 2 J (2\gamma^2 + J^2)
\begin{tikzpicture}[baseline={(0, -0.1)}]
\draw (0.4, 0) circle (.2);
\draw (1.2, 0) circle (.2);
\draw (2.0, 0) circle (.2);
\draw (0.6, 0) -- (1.0, 0);
\draw (1.4, 0) -- (1.8, 0);
\node at (0, 0) {$\cdots$};
\node at (2.5, 0) {$\cdots$};
\node at (0.8, -0.2) {\small $1$};
\node at (1.6, -0.2) {\small $0$};
\end{tikzpicture}
- 2 \gamma (\gamma^2 + J^2)
\begin{tikzpicture}[baseline={(0, -0.1)}]
\draw (0.4, 0) circle (.2);
\draw (1.2, 0) circle (.2);
\draw (2.0, 0) circle (.2);
\draw (0.6, 0) -- (1.0, 0);
\draw (1.4, 0) -- (1.8, 0);
\node at (0, 0) {$\cdots$};
\node at (2.5, 0) {$\cdots$};
\node at (0.8, -0.2) {\small $1$};
\node at (1.6, -0.2) {\small $1$};
\end{tikzpicture}
\notag \\
\begin{tikzpicture}[baseline={(0, -0.1)}]
\draw (0.05, 0.25) rectangle (2.35, -0.25);
\draw (0.0, 0.3) rectangle (2.4, -0.3);
\draw (0.4, 0) circle (.2);
\draw (1.2, 0) circle (.2);
\draw (2.0, 0) circle (.2);
\draw (0.6, 0) -- (1.0, 0);
\draw (1.4, 0) -- (1.8, 0);
\draw[<-] (0.65, 0.45) -- (0.95, 0.45);
\node at (0, 0) {$\cdots$};
\node at (2.5, 0) {$\cdots$};
\node at (1.2, 0.45) {\small $1$};
\end{tikzpicture}
\equiv \; &
2 (J^4 + 3 J^2 \gamma^2 + \gamma^4)
\begin{tikzpicture}[baseline={(0, -0.1)}]
\draw (0.4, 0) circle (.2);
\draw (1.2, 0) circle (.2);
\draw (2.0, 0) circle (.2);
\draw (0.6, 0) -- (1.0, 0);
\draw (1.4, 0) -- (1.8, 0);
\node at (0, 0) {$\cdots$};
\node at (2.5, 0) {$\cdots$};
\node at (0.8, -0.2) {\small $0$};
\node at (1.6, -0.2) {\small $0$};
\end{tikzpicture}
- 2 \gamma J (2\gamma^2 + J^2)
\begin{tikzpicture}[baseline={(0, -0.1)}]
\draw (0.4, 0) circle (.2);
\draw (1.2, 0) circle (.2);
\draw (2.0, 0) circle (.2);
\draw (0.6, 0) -- (1.0, 0);
\draw (1.4, 0) -- (1.8, 0);
\node at (0, 0) {$\cdots$};
\node at (2.5, 0) {$\cdots$};
\node at (0.8, -0.2) {\small $0$};
\node at (1.6, -0.2) {\small $1$};
\end{tikzpicture}
\notag \\ &
- 2 \gamma J (2\gamma^2 + J^2)
\begin{tikzpicture}[baseline={(0, -0.1)}]
\draw (0.4, 0) circle (.2);
\draw (1.2, 0) circle (.2);
\draw (2.0, 0) circle (.2);
\draw (0.6, 0) -- (1.0, 0);
\draw (1.4, 0) -- (1.8, 0);
\node at (0, 0) {$\cdots$};
\node at (2.5, 0) {$\cdots$};
\node at (0.8, -0.2) {\small $1$};
\node at (1.6, -0.2) {\small $0$};
\end{tikzpicture}
+ 2 \gamma^2 (\gamma^2 + J^2)
\begin{tikzpicture}[baseline={(0, -0.1)}]
\draw (0.4, 0) circle (.2);
\draw (1.2, 0) circle (.2);
\draw (2.0, 0) circle (.2);
\draw (0.6, 0) -- (1.0, 0);
\draw (1.4, 0) -- (1.8, 0);
\node at (0, 0) {$\cdots$};
\node at (2.5, 0) {$\cdots$};
\node at (0.8, -0.2) {\small $1$};
\node at (1.6, -0.2) {\small $1$};
\end{tikzpicture}
\notag \\
\begin{tikzpicture}[baseline={(0, -0.1)}]
\draw (0.05, 0.25) rectangle (2.35, -0.25);
\draw (0.0, 0.3) rectangle (2.4, -0.3);
\draw (0.4, 0) circle (.2);
\draw (1.2, 0) circle (.2);
\draw (2.0, 0) circle (.2);
\draw (0.6, 0) -- (1.0, 0);
\draw (1.4, 0) -- (1.8, 0);
\draw[->] (1.45, 0.45) -- (1.75, 0.45);
\node at (0, 0) {$\cdots$};
\node at (2.5, 0) {$\cdots$};
\node at (1.2, 0.45) {\small $1$};
\end{tikzpicture}
\equiv \; &
2 \gamma^2 (\gamma^2 + J^2)
\begin{tikzpicture}[baseline={(0, -0.1)}]
\draw (0.4, 0) circle (.2);
\draw (1.2, 0) circle (.2);
\draw (2.0, 0) circle (.2);
\draw (0.6, 0) -- (1.0, 0);
\draw (1.4, 0) -- (1.8, 0);
\node at (0, 0) {$\cdots$};
\node at (2.5, 0) {$\cdots$};
\node at (0.8, -0.2) {\small $0$};
\node at (1.6, -0.2) {\small $0$};
\end{tikzpicture}
- 2 \gamma J (2\gamma^2 + J^2)
\begin{tikzpicture}[baseline={(0, -0.1)}]
\draw (0.4, 0) circle (.2);
\draw (1.2, 0) circle (.2);
\draw (2.0, 0) circle (.2);
\draw (0.6, 0) -- (1.0, 0);
\draw (1.4, 0) -- (1.8, 0);
\node at (0, 0) {$\cdots$};
\node at (2.5, 0) {$\cdots$};
\node at (0.8, -0.2) {\small $0$};
\node at (1.6, -0.2) {\small $1$};
\end{tikzpicture}
\notag \\ &
- 2 \gamma J (2\gamma^2 + J^2)
\begin{tikzpicture}[baseline={(0, -0.1)}]
\draw (0.4, 0) circle (.2);
\draw (1.2, 0) circle (.2);
\draw (2.0, 0) circle (.2);
\draw (0.6, 0) -- (1.0, 0);
\draw (1.4, 0) -- (1.8, 0);
\node at (0, 0) {$\cdots$};
\node at (2.5, 0) {$\cdots$};
\node at (0.8, -0.2) {\small $1$};
\node at (1.6, -0.2) {\small $0$};
\end{tikzpicture}
+ 2 (J^4 + 3 J^2 \gamma^2 + \gamma^4)
\begin{tikzpicture}[baseline={(0, -0.1)}]
\draw (0.4, 0) circle (.2);
\draw (1.2, 0) circle (.2);
\draw (2.0, 0) circle (.2);
\draw (0.6, 0) -- (1.0, 0);
\draw (1.4, 0) -- (1.8, 0);
\node at (0, 0) {$\cdots$};
\node at (2.5, 0) {$\cdots$};
\node at (0.8, -0.2) {\small $1$};
\node at (1.6, -0.2) {\small $1$};
\end{tikzpicture}
.\end{alignat}

The doubly-rectangularized states satisfy:

\begin{enumerate}
\item When acted by the Hamiltonian once, the neutral-rectangularized spin states transform as
\begin{alignat}{8}
    \hat{h}_{m}^{(L)} 
    \begin{tikzpicture}[baseline={(0, -0.1)}]
    \draw (0.05, 0.25) rectangle (2.35, -0.25);
    \draw (0.0, 0.3) rectangle (2.4, -0.3);
    \draw (0.4, 0) circle (.2);
    \draw (1.2, 0) circle (.2);
    \draw (2.0, 0) circle (.2);
    \fill (1.2, 0) circle (2pt);
    \draw (0.6, 0) -- (1.0, 0);
    \draw (1.4, 0) -- (1.8, 0);
    \node at (0, 0) {$\cdots$};
    \node at (2.5, 0) {$\cdots$};
    \node at (1.2, 0.45) {\small $0$};
    \node at (1.2, -0.45) {\small $m$};
    \end{tikzpicture}
    &=
    \begin{tikzpicture}[baseline={(0, -0.1)}]
    \draw (0.05, 0.25) rectangle (2.35, -0.25);
    \draw (0.0, 0.3) rectangle (2.4, -0.3);
    \draw (0.4, 0) circle (.2);
    \draw (1.2, 0) circle (.2);
    \draw (2.0, 0) circle (.2);
    \fill (0.4, 0) circle (2pt);
    \draw (0.6, 0) -- (1.0, 0);
    \draw (1.4, 0) -- (1.8, 0);
    \draw[<-] (0.65, 0.45) -- (0.95, 0.45);
    \node at (0, 0) {$\cdots$};
    \node at (2.5, 0) {$\cdots$};
    \node at (1.2, 0.45) {\small $0$};
    \node at (1.2, -0.45) {\small $m$};
    \end{tikzpicture}
    \notag \\
    \hat{h}_{m}^{(R)} 
    \begin{tikzpicture}[baseline={(0, -0.1)}]
    \draw (0.05, 0.25) rectangle (2.35, -0.25);
    \draw (0.0, 0.3) rectangle (2.4, -0.3);
    \draw (0.4, 0) circle (.2);
    \draw (1.2, 0) circle (.2);
    \draw (2.0, 0) circle (.2);
    \fill (1.2, 0) circle (2pt);
    \draw (0.6, 0) -- (1.0, 0);
    \draw (1.4, 0) -- (1.8, 0);
    \node at (0, 0) {$\cdots$};
    \node at (2.5, 0) {$\cdots$};
    \node at (1.2, 0.45) {\small $0$};
    \node at (1.2, -0.45) {\small $m$};
    \end{tikzpicture}
    &=
    \begin{tikzpicture}[baseline={(0, -0.1)}]
    \draw (0.05, 0.25) rectangle (2.35, -0.25);
    \draw (0.0, 0.3) rectangle (2.4, -0.3);
    \draw (0.4, 0) circle (.2);
    \draw (1.2, 0) circle (.2);
    \draw (2.0, 0) circle (.2);
    \fill (2.0, 0) circle (2pt);
    \draw (0.6, 0) -- (1.0, 0);
    \draw (1.4, 0) -- (1.8, 0);
    \draw[->] (1.45, 0.45) -- (1.75, 0.45);
    \node at (0, 0) {$\cdots$};
    \node at (2.5, 0) {$\cdots$};
    \node at (1.2, 0.45) {\small $0$};
    \node at (1.2, -0.45) {\small $m$};
    \end{tikzpicture}
    \notag \\
    \hat{h}_{m}^{(L)} 
    \begin{tikzpicture}[baseline={(0, -0.1)}]
    \draw (0.05, 0.25) rectangle (2.35, -0.25);
    \draw (0.0, 0.3) rectangle (2.4, -0.3);
    \draw (0.4, 0) circle (.2);
    \draw (1.2, 0) circle (.2);
    \draw (2.0, 0) circle (.2);
    \fill (1.2, 0) circle (2pt);
    \draw (0.6, 0) -- (1.0, 0);
    \draw (1.4, 0) -- (1.8, 0);
    \node at (0, 0) {$\cdots$};
    \node at (2.5, 0) {$\cdots$};
    \node at (1.2, 0.45) {\small $1$};
    \node at (1.2, -0.45) {\small $m$};
    \end{tikzpicture}
    &=
    \begin{tikzpicture}[baseline={(0, -0.1)}]
    \draw (0.05, 0.25) rectangle (2.35, -0.25);
    \draw (0.0, 0.3) rectangle (2.4, -0.3);
    \draw (0.4, 0) circle (.2);
    \draw (1.2, 0) circle (.2);
    \draw (2.0, 0) circle (.2);
    \fill (0.4, 0) circle (2pt);
    \draw (0.6, 0) -- (1.0, 0);
    \draw (1.4, 0) -- (1.8, 0);
    \draw[<-] (0.65, 0.45) -- (0.95, 0.45);
    \node at (0, 0) {$\cdots$};
    \node at (2.5, 0) {$\cdots$};
    \node at (1.2, 0.45) {\small $1$};
    \node at (1.2, -0.45) {\small $m$};
    \end{tikzpicture}
    \notag \\
    \hat{h}_{m}^{(R)} 
    \begin{tikzpicture}[baseline={(0, -0.1)}]
    \draw (0.05, 0.25) rectangle (2.35, -0.25);
    \draw (0.0, 0.3) rectangle (2.4, -0.3);
    \draw (0.4, 0) circle (.2);
    \draw (1.2, 0) circle (.2);
    \draw (2.0, 0) circle (.2);
    \fill (1.2, 0) circle (2pt);
    \draw (0.6, 0) -- (1.0, 0);
    \draw (1.4, 0) -- (1.8, 0);
    \node at (0, 0) {$\cdots$};
    \node at (2.5, 0) {$\cdots$};
    \node at (1.2, 0.45) {\small $1$};
    \node at (1.2, -0.45) {\small $m$};
    \end{tikzpicture}
    &=
    \begin{tikzpicture}[baseline={(0, -0.1)}]
    \draw (0.05, 0.25) rectangle (2.35, -0.25);
    \draw (0.0, 0.3) rectangle (2.4, -0.3);
    \draw (0.4, 0) circle (.2);
    \draw (1.2, 0) circle (.2);
    \draw (2.0, 0) circle (.2);
    \fill (2.0, 0) circle (2pt);
    \draw (0.6, 0) -- (1.0, 0);
    \draw (1.4, 0) -- (1.8, 0);
    \draw[->] (1.45, 0.45) -- (1.75, 0.45);
    \node at (0, 0) {$\cdots$};
    \node at (2.5, 0) {$\cdots$};
    \node at (1.2, 0.45) {\small $1$};
    \node at (1.2, -0.45) {\small $m$};
    \end{tikzpicture}
.\end{alignat}

\item When acted by the Hamiltonian twice, the neutral-rectangularized spin states transform as
\begin{alignat}{8}
    \left( \hat{h}_{m-1}^{(R)} \hat{h}_m^{(L)} + \hat{h}_{m+1}^{(L)} \hat{h}_m^{(R)} \right) \;
    \begin{tikzpicture}[baseline={(0, -0.1)}]
    \draw (0.05, 0.25) rectangle (2.35, -0.25);
    \draw (0.0, 0.3) rectangle (2.4, -0.3);
    \draw (0.4, 0) circle (.2);
    \draw (1.2, 0) circle (.2);
    \draw (2.0, 0) circle (.2);
    \fill (1.2, 0) circle (2pt);
    \draw (0.6, 0) -- (1.0, 0);
    \draw (1.4, 0) -- (1.8, 0);
    \node at (0, 0) {$\cdots$};
    \node at (2.5, 0) {$\cdots$};
    \node at (1.2, 0.45) {\small $0$};
    \node at (1.2, -0.45) {\small $m$};
    \end{tikzpicture}
    = \; &
    \hat{h}_{m-1}^{(R)} \bigg(
    \begin{tikzpicture}[baseline={(0, -0.1)}]
    \draw (0.05, 0.25) rectangle (2.35, -0.25);
    \draw (0.0, 0.3) rectangle (2.4, -0.3);
    \draw (0.4, 0) circle (.2);
    \draw (1.2, 0) circle (.2);
    \draw (2.0, 0) circle (.2);
    \fill (0.4, 0) circle (2pt);
    \draw (0.6, 0) -- (1.0, 0);
    \draw (1.4, 0) -- (1.8, 0);
    \draw[<-] (0.65, 0.45) -- (0.95, 0.45);
    \node at (0, 0) {$\cdots$};
    \node at (2.5, 0) {$\cdots$};
    \node at (1.2, 0.45) {\small $0$};
    \node at (1.2, -0.45) {\small $m$};
    \end{tikzpicture}
    \bigg)
    +
    \hat{h}_{m+1}^{(L)}\bigg(
    \begin{tikzpicture}[baseline={(0, -0.1)}]
    \draw (0.05, 0.25) rectangle (2.35, -0.25);
    \draw (0.0, 0.3) rectangle (2.4, -0.3);
    \draw (0.4, 0) circle (.2);
    \draw (1.2, 0) circle (.2);
    \draw (2.0, 0) circle (.2);
    \fill (2.0, 0) circle (2pt);
    \draw (0.6, 0) -- (1.0, 0);
    \draw (1.4, 0) -- (1.8, 0);
    \draw[->] (1.45, 0.45) -- (1.75, 0.45);
    \node at (0, 0) {$\cdots$};
    \node at (2.5, 0) {$\cdots$};
    \node at (1.2, 0.45) {\small $0$};
    \node at (1.2, -0.45) {\small $m$};
    \end{tikzpicture}
    \bigg)
    \notag \\
    = \; &
    \begin{tikzpicture}[baseline={(0, -0.1)}]
    \draw (0.05, 0.25) rectangle (2.35, -0.25);
    \draw (0.0, 0.3) rectangle (2.4, -0.3);
    \draw (0.4, 0) circle (.2);
    \draw (1.2, 0) circle (.2);
    \draw (2.0, 0) circle (.2);
    \fill (1.2, 0) circle (2pt);
    \draw (0.6, 0) -- (1.0, 0);
    \draw (1.4, 0) -- (1.8, 0);
    \node at (0, 0) {$\cdots$};
    \node at (2.5, 0) {$\cdots$};
    \node at (1.2, 0.45) {\small $1$};
    \node at (1.2, -0.45) {\small $m$};
    \end{tikzpicture}
    \notag \\
    \left( \hat{h}_{m-1}^{(R)} \hat{h}_m^{(L)} + \hat{h}_{m+1}^{(L)} \hat{h}_m^{(R)} \right) \;
    \begin{tikzpicture}[baseline={(0, -0.1)}]
    \draw (0.05, 0.25) rectangle (2.35, -0.25);
    \draw (0.0, 0.3) rectangle (2.4, -0.3);
    \draw (0.4, 0) circle (.2);
    \draw (1.2, 0) circle (.2);
    \draw (2.0, 0) circle (.2);
    \fill (1.2, 0) circle (2pt);
    \draw (0.6, 0) -- (1.0, 0);
    \draw (1.4, 0) -- (1.8, 0);
    \node at (0, 0) {$\cdots$};
    \node at (2.5, 0) {$\cdots$};
    \node at (1.2, 0.45) {\small $1$};
    \node at (1.2, -0.45) {\small $m$};
    \end{tikzpicture}    
    = \; &
    \hat{h}_{m-1}^{(R)} \bigg(
    \begin{tikzpicture}[baseline={(0, -0.1)}]
    \draw (0.05, 0.25) rectangle (2.35, -0.25);
    \draw (0.0, 0.3) rectangle (2.4, -0.3);
    \draw (0.4, 0) circle (.2);
    \draw (1.2, 0) circle (.2);
    \draw (2.0, 0) circle (.2);
    \fill (0.4, 0) circle (2pt);
    \draw (0.6, 0) -- (1.0, 0);
    \draw (1.4, 0) -- (1.8, 0);
    \draw[<-] (0.65, 0.45) -- (0.95, 0.45);
    \node at (0, 0) {$\cdots$};
    \node at (2.5, 0) {$\cdots$};
    \node at (1.2, 0.45) {\small $1$};
    \node at (1.2, -0.45) {\small $m$};
    \end{tikzpicture}
    \bigg)
    +
    \hat{h}_{m+1}^{(L)}\bigg(
    \begin{tikzpicture}[baseline={(0, -0.1)}]
    \draw (0.05, 0.25) rectangle (2.35, -0.25);
    \draw (0.0, 0.3) rectangle (2.4, -0.3);
    \draw (0.4, 0) circle (.2);
    \draw (1.2, 0) circle (.2);
    \draw (2.0, 0) circle (.2);
    \fill (2.0, 0) circle (2pt);
    \draw (0.6, 0) -- (1.0, 0);
    \draw (1.4, 0) -- (1.8, 0);
    \draw[->] (1.45, 0.45) -- (1.75, 0.45);
    \node at (0, 0) {$\cdots$};
    \node at (2.5, 0) {$\cdots$};
    \node at (1.2, 0.45) {\small $1$};
    \node at (1.2, -0.45) {\small $m$};
    \end{tikzpicture}
    \bigg)
    \notag \\
    = \; &
    (-4 \gamma^4)
    \begin{tikzpicture}[baseline={(0, -0.1)}]
    \draw (0.05, 0.25) rectangle (2.35, -0.25);
    \draw (0.0, 0.3) rectangle (2.4, -0.3);
    \draw (0.4, 0) circle (.2);
    \draw (1.2, 0) circle (.2);
    \draw (2.0, 0) circle (.2);
    \fill (1.2, 0) circle (2pt);
    \draw (0.6, 0) -- (1.0, 0);
    \draw (1.4, 0) -- (1.8, 0);
    \node at (0, 0) {$\cdots$};
    \node at (2.5, 0) {$\cdots$};
    \node at (1.2, 0.45) {\small $0$};
    \node at (1.2, -0.45) {\small $m$};
    \end{tikzpicture}
    + (4 \gamma^2 + 2J^2)
    \begin{tikzpicture}[baseline={(0, -0.1)}]
    \draw (0.05, 0.25) rectangle (2.35, -0.25);
    \draw (0.0, 0.3) rectangle (2.4, -0.3);
    \draw (0.4, 0) circle (.2);
    \draw (1.2, 0) circle (.2);
    \draw (2.0, 0) circle (.2);
    \fill (1.2, 0) circle (2pt);
    \draw (0.6, 0) -- (1.0, 0);
    \draw (1.4, 0) -- (1.8, 0);
    \node at (0, 0) {$\cdots$};
    \node at (2.5, 0) {$\cdots$};
    \node at (1.2, 0.45) {\small $1$};
    \node at (1.2, -0.45) {\small $m$};
    \end{tikzpicture}
\end{alignat}
\end{enumerate}

\end{widetext}

For models with system size $M=4Z+3$ ($Z \in \mathbb{Z}_{\geq 0}$), an extensive number of root states $| \Psi \rangle$ of $\mathcal{K}^{(4)}$ satisfy:
\ie
    \langle \Psi | H^2 \hat{n}_m H^2 | \Psi \rangle &= 0
    \\
    \langle \Psi | H^4 \hat{n}_m H^4 | \Psi \rangle &= 0
,\fe
for $m \in 4Z$ ($Z \in \mathbb{Z}_{>0}$). These root states can be constructed as follows:
\begin{enumerate}
    \item We express the root state $| \Psi \rangle$ as 
    \ie
        | \Psi \rangle = \sum_{n=0}^{Z} \sum_{r_i=0, 1} (-1)^n c_{\vec{r}}^{(n)} | \Phi_{\{ \vec{r} \}}^{(4n+2)} \rangle
    ,\fe
    where $\vec{r}=\{ r_1, r_2, \cdots, r_{Z+1} \}$. Each of the $(Z+1) \cdot 2^{Z+1}$ component $\Phi_{\{ \vec{r} \}}^{(4n+2)}$ is a product state. For the components $| \Phi_{\{ \vec{r} \}}^{(4n+2)} \rangle$, a fermion occupies the $(4n+2)$-th site.

    \item For the components $| \Phi_{\{ \vec{r} \}}^{(4n+2)} \rangle$, where the fermion occupies the $f_1=(4n+2)$-th site, we construct the components from the rectangular building blocks introduced in Eq.~\eqref{eqn: Doubly rectangular exact form} as follows:
    \ie
    \begin{cases}
        \text{Doubly right-rectangularize:}
        \\
        \;\;\;\{ (m-1, m, m+1) | m \in 4\mathbb{Z}_{\geq 0}+2 \text{ and } m<f_1 \},
        \\
        \text{doubly neutral-rectangularize:}
        \\
        \;\;\;(f_1-1, f_1, f_1+1),
        \\
        \text{doubly left-rectangularize:}
        \\
        \;\;\;\{ (m-1, m, m+1) | m \in 4\mathbb{Z}_{\geq 0}+2 \text{ and } m>f_1 \}.
    \end{cases}
    \fe
    There are then a total of $(Z+1)$ rectangles. Fix the rectangular labels indicated by the subscript of the components: $r_1, \cdots r_{Z+1}$.
    
    \item Arbitrarily choose the spin states 
    \ie
    \label{eqn: zero modes of K4, spin states}
    \{(s_m, s_{m+1}) | m \in 4\mathbb{Z}_{>0} \}
    \fe
    to be either $0$ or $1$ ($2Z$ in total). For components $| \Phi_{\{ \vec{r} \}}^{(4n+2)} \rangle$, where the fermion occupies the $f_1=(4n+2)$-th site, fix the spin states as in Eq.~\eqref{eqn: spin states for rectangularized states}.

    \item In order to make the root states satisfy $\langle \Psi | H^2 \hat{n}_m H^2 | \Psi \rangle = 0$ for $m \in 4\mathbb{Z}_{>0}$, the coefficients $c_{\vec{r}}^{(n)}$ are fixed to be
    \ie
        c_{\vec{r}}^{(0)} = c_{\vec{r}}^{(1)} = \cdots = c_{\vec{r}}^{(N)} \equiv c_{\vec{r}}
    \fe

    \item In order to make the root states satisfy $\langle \Psi | H^4 \hat{n}_m H^4 | \Psi \rangle = 0$ for $m \in 4\mathbb{Z}_{>0}$, the coefficients $c_{\vec{r}}$ with the same number of $0$'s and $1$'s in $\vec{r}$ should be equal to each other:
    \ie
        c_{0,0,\cdots, 0} &\equiv \alpha_0
        \\
        c_{1,0,\cdots, 0} = c_{0,1,\cdots, 0} = \cdots = c_{0,0,\cdots, 1} &\equiv \alpha_1
        \\
        & \vdots
        \\
        c_{1,1,\cdots, 1} &\equiv \alpha_{Z+1}
    .\fe
    In addition, for $i \in [0, Z-1]$, we require
    \ie
    \alpha_i + (4\gamma^2 + 2J^2) \alpha_{i+1} + 4\gamma^4 \alpha_{i+2} = 0
    .\fe
\end{enumerate}
Given the above constructions, it can be shown that
\ie
    H^4 | \Psi \rangle = (4\gamma^2 + 2J^2) H^2 | \Psi \rangle - 4\gamma^4 | \Psi \rangle
,\fe
making $| \Psi \rangle$ a root state of $\mathcal{K}^{(4)}$. Since we can arbitrarily choose $2Z$ spin states in Eq.~\eqref{eqn: zero modes of K4, spin states}, the number of the root states constructed from the above procedure is $2^{2Z}$. Hence, the number of $\mathcal{K}^{(4)}$ in the symmetry sector, i.e., $D(\mathcal{K}^{(4)})$, increases exponentially with the system size.

\section{Characteristic polynomial factorization for models with uniform $\gamma_m$ and $J_m$}
\label{appendix: characteristic polynomial factorization}

In this Appendix, we demonstrate polynomial factorization for models with Hilbert space dimensions of up to $\lesssim 12000$. We substitute integers for $\gamma_m=\gamma$ and $J_m=J$ in the Hamiltonian, and factorize the characteristic polynomial explicitly to determine the structure of the Krylov subspaces. We consider several sets of integers $\left\{ \gamma_m, J_m \right\}$ to check that the factorization does not depend on some fine-tuned choice.

Given the characteristic polynomial for a Hamiltonian $H$, we denote the number of distinct factors as $I$, and denote the factors as $f_{i}$:
\ie
\label{eqn: characteristic polynomial in a product form}
\det(\lambda \mathds{1} - H) = \prod_{i=1}^{I} \left(f_{i}(\lambda)\right)^{d_{i}}
,\fe
where the integers $d_{i}$ denote the degeneracies of the factor $f_{i}(\lambda)$. In addition, we denote the degree of $f_{i}(\lambda)$ as $n_{i}$. Taking Eq.~\eqref{eqn: characteristic polynomial of M=3} as an example, we have
\ie
\begin{cases}
    f_{1}(\lambda) = \lambda, \; n_{1}=1 , \; d_{1} = 4;
    \\
    f_{2}(\lambda) = \left[ \lambda^2-\left( 2\gamma^2+J^2 \right) \right], \; n_{2}=2 , \; d_{2} = 2;
    \\
    f_{3}(\lambda) = \left[ \lambda^4 - \left( 4\gamma^2+2J^2 \right) \lambda^2 + 4\gamma^4 \right] , \; n_{3}=4, \; d_{3} = 1;
\end{cases}
\fe

In the following tables, we present the degeneracies of the factors in different symmetry sectors. In the row labeled ``factor degrees,'' we display various values of $1 \le n_i \le |\mathcal{K}^{\text{max}}|$. In the row labeled ``factor degeneracies,'' we provide the corresponding factor degeneracies.
\ie
d_{i_1} + d_{i_2} + \cdots + d_{i_p}
,\fe
where $n_{i_1} = n_{i_2} = \cdots = n_{i_p}$. If $d_{i_1} = d_{i_2} = \cdots = d_{i_p}$, we impose the following abbreviation
\ie
d_{i_1}+\cdots+d_{i_p} \mapsto  d_{i_1} \times p
.\fe

The degeneracy plays an important role in the discussion of the dimension of the Krylov subspaces generated by
arbitrary states, which is discussed in Appendix~\ref{appendix: Krylov subspace generated by arbitrary state}.

\subsection{$Q=1$}
\label{appendix: Q=1 numerical factorization}
\begin{itemize}

\item $M=3$: $\dim(\mathcal{H}) = 12$

\begin{center}
\begin{tabular}{ || c || c | c | } 
\hline
$\#$ of zero modes & \multicolumn{2}{|c|}{4} \\
\hline
\hline
Factor degrees & $2$ & $4$ \\ 
\hline
Factor degeneracies & $2$ & $1$ \\ 
\hline
\end{tabular}
\end{center}

\item $M=4$: $\dim(\mathcal{H}) = 32$

\begin{center}
\begin{tabular}{ |c || c|c| } 
\hline
$\#$ of zero modes & \multicolumn{2}{|c|}{0} \\
\hline
\hline
Factor degrees & $4$ & $8$ \\ 
\hline
Factor degeneracies & $1+1+1+1 $ & $2$ \\ 
\hline
\end{tabular}
\end{center}

\item $M=5$: $\dim(\mathcal{H}) = 80$

\begin{center}
\begin{tabular}{ |c || c | c | c |} 
\hline
$\#$ of zero modes & \multicolumn{3}{|c|}{16} \\
\hline
\hline
Factor degrees & $2$ & $4$ & $8$ \\ 
\hline
Factor degeneracies & $2 + 2 $ & $1+2+3+4$ & 2 \\ 
\hline
\end{tabular}
\end{center}

\item $M=6$: $\dim(\mathcal{H}) = 192$

\begin{center}
\begin{tabular}{ |c || c | c |} 
\hline
$\#$ of zero modes & \multicolumn{2}{|c|}{0} \\
\hline
\hline
Factor degrees & $6$ & $12$ \\ 
\hline
Factor degeneracies & $ 1 \times 8 $ & $ 2 \times 6$  \\ 
\hline
\end{tabular}
\end{center}

\item $M=7$: $\dim(\mathcal{H}) = 448$

\begin{center}
\begin{tabular}{ |c || c | c | c | c | c |} 
\hline
$\#$ of zero modes & \multicolumn{5}{|c|}{64} \\
\hline
\hline
Factor degrees & $2$ & $4$ & $6$ & $8$ & $12$ \\ 
\hline
Factor degeneracies & $18$ & $2+4+6$ & $2$ & $1\times 3 + 2 + 4$ & $2\times 9$ \\ 
\hline
\end{tabular}
\end{center}

\begin{widetext}

\item $M=8$: $\dim(\mathcal{H}) = 1024$

\begin{center}
\begin{tabular}{ |c || c | c | c | c | c |} 
\hline
$\#$ of zero modes & \multicolumn{5}{|c|}{0} \\
\hline
\hline
Factor degrees & $2$ & $6$ & $8$ & $12$ & $16$ \\ 
\hline
Factor degeneracies & $16+16$ & $1\times 6 + 3 \times 2$ & $1\times 8$ & $2\times 3 + 4$ & $2\times 22$ \\ 
\hline
\end{tabular}
\end{center}

\item $M=9$: $\dim(\mathcal{H}) = 2304$

\begin{center}
\begin{tabular}{ |c || c | c | c |} 
\hline
$\#$ of zero modes & \multicolumn{3}{|c|}{256} \\
\hline
\hline
Factor degrees & $4$ &  $8$  & $16$ \\ 
\hline
Factor degeneracies & $2\times2+4\times4$ & $1\times4+2\times6+3+4+5\times3+16$ & $2\times48$  \\ 
\hline
\end{tabular}
\end{center}

\item $M=10$: $\dim(\mathcal{H}) = 5120$

\begin{center}
\begin{tabular}{ |c || c | c |} 
\hline
$\#$ of zero modes & \multicolumn{2}{|c|}{0} \\
\hline
\hline
Factor degrees & $10$ &  $20$ \\ 
\hline
Factor degeneracies & $1\times32$ & $2\times116+4\times2$ \\ 
\hline
\end{tabular}
\end{center}

\item $M=11$: $\dim(\mathcal{H}) = 11264$

\begin{center}
\begin{tabular}{ |c || c | c | c |} 
\hline
$\#$ of zero modes & \multicolumn{3}{|c|}{1024} \\
\hline
\hline
Factor degrees & $2$ &  $4$ & $6$  \\ 
\hline
Factor degeneracies & $22\times2+64\times2+168$ & $2+4+16+18\times2+46$ & $6\times1$ \\ 
\hline
\hline
Factor degrees & $8$ &  $10$ & $12$ \\ 
\hline
Factor degeneracies & $1+2\times5+3\times2+4\times3+16$ & $2\times5$ & $1\times7+2\times7+3+4\times2+6\times2+8\times2$ \\ 
\hline
\hline
Factor degrees &  $16$ &  $20$ & \\ 
\hline
Factor degeneracies  &  $2\times45+4\times2$ & $2\times159$ & \\ 
\hline
\end{tabular}
\end{center}

\end{widetext}

\end{itemize}

\subsection{$Q=2$}

\begin{itemize}

\item $M=4$: $\dim(\mathcal{H}) = 48$

\begin{center}
\begin{tabular}{ |c || c | c |} 
\hline
$\#$ of zero modes & \multicolumn{2}{|c|}{16} \\
\hline
\hline
Factor degrees & $4$ &  $8$ \\ 
\hline
Factor degeneracies & $2+2$ & $1+1$ \\ 
\hline
\end{tabular}
\end{center}

\item $M=5$: $\dim(\mathcal{H}) = 160$

\begin{center}
\begin{tabular}{ |c || c | c | c | c | c |} 
\hline
$\#$ of zero modes & \multicolumn{5}{|c|}{32} \\
\hline
\hline
Factor degrees & $2$ & $4$ & $16$ & $32$ & $36$ \\ 
\hline
Factor degeneracies & $2+2$ & $1$ & $1$ & $2$ & $1$ \\ 
\hline
\end{tabular}
\end{center}

\item $M=6$: $\dim(\mathcal{H}) = 480$

\begin{center}
\begin{tabular}{ |c || c |} 
\hline
$\#$ of zero modes & \multicolumn{1}{|c|}{96} \\
\hline
\hline
Factor degrees & $96$ \\ 
\hline
Factor degeneracies & $1+1+2$ \\ 
\hline
\end{tabular}
\end{center}

\item $M=7$: $\dim(\mathcal{H}) = 1344$

\begin{center}
\begin{tabular}{ |c || c | c | c |} 
\hline
$\#$ of zero modes & \multicolumn{3}{|c|}{192} \\
\hline
\hline
Factor degrees & $264$ & $288$ & $312$ \\ 
\hline
Factor degeneracies & $1$ & $2$ & $1$ \\ 
\hline
\end{tabular}
\end{center}

\item $M=8$: $\dim(\mathcal{H}) = 3584$

\begin{center}
\begin{tabular}{ |c || c |} 
\hline
$\#$ of zero modes & \multicolumn{1}{|c|}{512} \\
\hline
\hline
Factor degrees & $768$ \\ 
\hline
Factor degeneracies & $1+1+2$ \\ 
\hline
\end{tabular}
\end{center}

\item $M=9$: $\dim(\mathcal{H}) = 9216$

\begin{center}
\begin{tabular}{ |c || c | c | c |} 
\hline
$\#$ of zero modes & \multicolumn{3}{|c|}{1024} \\
\hline
\hline
Factor degrees & $1984$ & $2048$ & $2112$ \\ 
\hline
Factor degeneracies & $1$ & $2$ & $1$ \\ 
\hline
\end{tabular}
\end{center}

\end{itemize}

\subsection{$Q=3$}

\begin{itemize}

\item $M=6$: $\dim(\mathcal{H}) = 640$

\begin{center}
\begin{tabular}{ |c || c | c | c |} 
\hline
$\#$ of zero modes & \multicolumn{3}{|c|}{0} \\
\hline
\hline
Factor degrees & $64$ & $80$ & $96$ \\ 
\hline
Factor degeneracies & $1+1$ & $2+2$ & $1+1$ \\ 
\hline
\end{tabular}
\end{center}

\item $M=7$: $\dim(\mathcal{H}) = 2240$

\begin{center}
\begin{tabular}{ |c || c |} 
\hline
$\#$ of zero modes & \multicolumn{1}{|c|}{192} \\
\hline
\hline
Factor degrees & $512$ \\ 
\hline
Factor degeneracies & $1+1+2$ \\ 
\hline
\end{tabular}
\end{center}

\item $M=8$: $\dim(\mathcal{H}) = 7168$

\begin{center}
\begin{tabular}{ |c || c |} 
\hline
$\#$ of zero modes & \multicolumn{1}{|c|}{0} \\
\hline
\hline
Factor degrees & $1792$ \\ 
\hline
Factor degeneracies & $1+1+2$ \\ 
\hline
\end{tabular}
\end{center}

\end{itemize}

\section{Krylov subspace generated by arbitrary states}
\label{appendix: Krylov subspace generated by arbitrary state}

Consider a Hamiltonian $H$ acting on a Hilbert space $\mathcal{H}$ with dimension $\mathcal{N}$. Generally, for an arbitrary state $| \Psi \rangle$, the column vectors of the matrix $U = \left\{ | \Psi \rangle, H | \Psi \rangle, \cdots, H^{\mathcal{N}-1} | \Psi \rangle \right\}$ span the entire Hilbert space, which is to say that $\rank \left( U \right) = \mathcal{N}$. However, as we will show, the degeneracies of the energy spectrum $d_{i}$ will lead to $\rank \left( U \right) < \mathcal{N}$.

In the factorization of the characteristic polynomial of Eq.~\eqref{eqn: characteristic polynomial in a product form}, each factor $f_{i}(\lambda)$ corresponds to $n_{i}$ different energy levels, each with a degeneracy of $d_{i}$. We denote these energy levels as $E_{i,n}$, where $n \in \left[1, n_i \right]$, and they are ordered as $E_{i,1} \le E_{i,2} \le \cdots \le E_{i,n_i}$. The eigenstates associated with these energy levels are represented as $| \psi_{i, n}^{k} \rangle$, where $k\in \left[ 1, d_i \right]$. We define the projection operator 
\ie
\hat{\Pi}_{i,n} \equiv \sum_{k=1}^{d_i} | \psi_{i, n}^{k} \rangle \langle \psi_{i, n}^{k} |
\fe
that projects an arbitrary state onto the subspace with energy $E_{i, n}$.

Given an arbitrary state $| \Psi \rangle$, although the projected state $| \psi_{i, n}^{\perp} \rangle \equiv \hat{\Pi}_{i,n} | \Psi \rangle$ is spanned by $d_{i}$ eigenstates $| \psi_{i, n}^{k} \rangle$, the rank of $U_{i, n}^{\perp} \equiv \{ | \psi_{i, n}^{\perp} \rangle, H | \psi_{i, n}^{\perp} \rangle, \cdots \}$ only has rank $1$, since $H | \psi_{i, n}^{\perp} \rangle \propto | \psi_{i, n}^{\perp} \rangle$. Consequently, the rank of the matrix $U = \{ | \Psi \rangle, \cdots, H^{\mathcal{N}-1} | \Psi \rangle \}$, where
\ie
| \Psi \rangle 
= \sum_{i=1}^{I} \sum_{n=1}^{n_i} | \psi_{i, n}^{\perp} \rangle
,\fe
will be upper bounded by
\ie
\rank \left(U\right) \le \sum_{i} n_i
.\fe
The upper bound is saturated when $| \Psi \rangle$ overlaps with all subspaces:
\ie
\forall (i, n), \; \hat{\Pi}_{i, n} | \Psi \rangle \neq 0
.\fe
From the tables displayed in Appendix~\ref{appendix: characteristic polynomial factorization}, we see that the upper bounds in the $Q=1$ sector are
\begin{center}
\begin{tabular}{ | c || c | c | c | c | c |} 
\hline
$M$ & $3$ & $4$ & $5$ & $6$ & $7$ \\
\hline
$\text{max } \rank(U)$ & $7$ & $24$ & $29$ & $120$ & $169$ \\
\hline
\hline
$M$ & $8$ & $9$ & $10$ & $11$ &  \\
\hline
$\text{max } \rank(U)$ & $516$ & $921$ & $2680$ & $4371$ &  \\
\hline
\end{tabular}
\end{center}



\end{document}